\documentclass[mnsc,nonblindrev]{informs3_nojournal}
\usepackage[utf8]{inputenc}

\usepackage{fix-cm}
\usepackage{lmodern}

\usepackage{algorithm}
\usepackage{algorithmic}
\usepackage{amsmath,amssymb,amsfonts}

\usepackage{graphicx}
\usepackage{xcolor}
\usepackage{subcaption}
\usepackage{booktabs}
\usepackage{mathtools}
\usepackage{url}
\usepackage{enumitem}
\usepackage{tikz}
\usetikzlibrary{arrows.meta,positioning,calc}

\usepackage{natbib}
 \bibpunct[, ]{(}{)}{,}{a}{}{,}%
 \def\bibfont{\small}%

\usepackage{tikz}
\usetikzlibrary{arrows.meta,positioning,calc}
\usepackage{pifont}
\usepackage{bm}   
\usepackage{xfrac}

\allowdisplaybreaks[4]
\def\BibTeX{{\rm B\kern-.05em{\sc i\kern-.025em b}\kern-.08em
    T\kern-.1667em\lower.7ex\hbox{E}\kern-.125emX}}

\renewcommand{\Halmos}{{$\blacksquare$}}

\OneAndAHalfSpacedXI

\TheoremsNumberedThrough
\EquationsNumberedThrough

\begin{document}
\def\thefootnote{*}\footnotetext{These authors contributed equally to this work.}
\def\thefootnote{\arabic{footnote}}
\RUNAUTHOR{Chen, Kim, Elmachtoub, and Xu}
\RUNTITLE{Fair Aggregation in Virtual Power Plants}
\TITLE{Fair Aggregation in Virtual Power Plants}

\ARTICLEAUTHORS{%
\AUTHOR{Liudong Chen\textsuperscript{*}}
\AFF{Department of Earth and Environmental Engineering, Columbia University, New York, NY 10027, \EMAIL{lc3671@columbia.edu}}
\AUTHOR{Hyemi Kim\textsuperscript{*}}
\AFF{Department of Industrial Engineering and Operations Research, Columbia University, New York, NY 10027, \EMAIL{hk3181@columbia.edu}}
\AUTHOR{Adam N. Elmachtoub}
\AFF{Department of Industrial Engineering and Operations Research and Data Science Institute, Columbia University, New York, NY 10027, \EMAIL{adam@ieor.columbia.edu}}
\AUTHOR{Bolun Xu}
\AFF{Department of Earth and Environmental Engineering, Columbia University, New York, NY 10027, \EMAIL{bx2177@columbia.edu}}
}

\ABSTRACT{
A virtual power plant (VPP) is operated by an aggregator that acts as a market intermediary to aggregate energy consumers to participate in wholesale power markets. 
By setting incentive prices, a VPP aggregator induces consumers to sell energy, and profits by providing this aggregated energy to the market. 
The energy supply is enabled by consumers’ flexibility to adjust electricity consumption in response to market conditions.
However, heterogeneity in consumers’ flexibility implies that profit-maximizing VPP pricing can generate inequalities in participation and benefit allocation across consumers.
In this paper, we develop a fairness-aware pricing framework to analyze how alternative fairness notions 
reshape the system performance, as measured by consumer Nash welfare, total consumer utility, and social welfare. 
We consider three fairness criteria: energy fairness, which ensures equitable energy provision, price fairness, which ensures consistency in incentive prices, and utility fairness, which ensures comparable levels of consumer utility. 
We model the aggregator-consumer interaction as a Stackelberg game and derive consumers’ optimal responses to incentive prices.
Using a stylized model, we show that profit-only pricing systematically disadvantages less flexible consumers. 
We show that incorporating energy fairness can either improve or worsen overall system performance across all measures, and the regimes that increase most system performance are attained only under moderate fairness levels. Surprisingly, price fairness never benefits less flexible consumers, even if it reduces the gap in incentive prices between consumers. On the other hand, utility fairness protects less flexible consumers without providing benefits to more flexible ones.
We validate our findings using data from an experiment in Norway under a tiered pricing scheme. Our results provide regulators and VPP operators with a systematic map linking fairness definitions and the level of enforcement to operational and welfare outcomes.

}


\maketitle
\section{Introduction}
Virtual power plant (VPP) aggregators are rapidly emerging as market intermediaries that pool consumer-owned distributed energy resources (DERs)--such as residential batteries, electric vehicles (EVs), and smart devices--and coordinate their participation in wholesale power markets. This intermediation arises from broader institutional constraints, including rules set by the Federal Energy Regulatory Commission (FERC) that limit the direct participation of small-scale resources in wholesale markets~\citep{FERC}. 
For example, the utility Consolidated Edison offers $\$18$/kW-month for committed electricity demand reductions~\citep{ConEd}, but participation requires a $50$~kW pledge, which is typically beyond a single household. 
By setting incentive prices, a VPP aggregator induces consumers to provide energy to the aggregator, which then aggregates this energy to support grid operations and generate revenue. The provided energy comes from discharging residential batteries or reducing electricity use through home appliances and thermostat adjustments~\citep{flexibility1}. 
Importantly, consumers differ in flexibility--their ability to provide energy in response to incentive prices--due to differences in technology, usage patterns, operational constraints, and socioeconomic characteristics. For example, higher flexibility can arise because higher-income households tend to own larger DER portfolios (e.g., multiple EVs or larger batteries)~\citep{flexibility2}. This heterogeneity introduces bias, whereby more flexible consumers are able to provide more energy and potentially receive greater benefits. 
This raises a central design question: \textit{what constitutes fairness in VPP operations, and how should VPP operations be designed to incorporate fair participation and benefit allocation across heterogeneous consumers?}

VPPs are already deployed at scale. In the United States, approximately $30$~GW of VPP capacity--operated by utilities and private companies--can serve about $3.75\%$ of peak demand at an estimated cost of $\$43$/kW-year, which is $37.7\%$ and $56.6\%$ lower than that of utility-scale batteries and gas peaker plants~\citep{DOE_VPP}. Prominent examples include Tesla’s VPP program, which aggregates residential Powerwall batteries to provide grid support in California and Texas. Consumers enroll by setting a minimum state-of-charge (SoC) threshold, below which their battery will not be discharged, and are compensated at $\$2$--$\$5$ per kWh discharged--significantly above retail electricity rates~\citep{Tesla_VPP}. Another emerging company is Base Power, which recently raised \$1 billion~\citep{BasePower}. Instead of selling residential batteries, Base Power leases batteries to households and operates them as a VPP, controlling aggregated storage while maintaining the SoC above $20\%$ to preserve backup capability for individuals~\citep{BasePower1}. VPP models are also expanding to EV and smart home devices through platforms such as WeaveGrid and EnergyHub~\citep{WeaveGrid, EnergyHub}.

With the rapid expansion of VPP deployment, fairness concerns have become increasingly salient in both policy and program practice.
The U.S. Department of Energy has emphasized fairness in VPP development, calling for broader access and fair mechanisms for enrollment and compensation~\citep{DOE_VPP, DOE_VPP2}. For example, PG\&E has implemented a VPP initiative requiring that at least $60$\% of participants come from disadvantaged or low-income communities~\citep{PGE}. EnergyHub reports that aligning program design with user behavioral trends and equity goals can increase enrollment among EV drivers~\citep{Energyhub_report}. Similarly, Tesla ensures that all participating consumers receive some compensation~\citep{Tesla_VPP}, and Base Power maintains each battery’s SoC above $20\%$~\citep{BasePower1}, reflecting attention to participant protection and fairness. 
Furthermore, socioeconomic studies also suggest that fairness considerations in daily operations can enhance consumer engagement, thereby supporting the scalability and reliability of VPP programs~\citep{ito2018moral, social2}. While these initiatives highlight the importance of fairness, existing regulatory and programmatic structures provide limited guidance on how specific incentive rules and fairness requirements affect participation, consumer welfare, and system performance.

\subsection{Summary of Main Contributions and Implications}
We develop a fairness-aware VPP pricing framework to study how incentive prices shape participation and benefit allocation when consumers are heterogeneous in their flexibility. We model the aggregator–consumer interaction as a Stackelberg game in which the aggregator sets incentive prices to maximize profit while anticipating consumers’ optimal energy responses derived from a utility maximization model. Consumer utility is defined as the payment received minus the cost of providing energy, and the amount of energy provided is limited by capacity--the maximum amount a consumer can provide to the VPP. We formalize three fairness criteria that can be operationalized through VPP pricing: energy fairness (promotes equitable energy provision across consumers), price fairness (promotes price consistency across consumers), and utility fairness (promotes balance in consumer utility). We evaluate the operational and welfare implications of these criteria using three complementary performance measures: (i) consumer Nash welfare (CNW), which captures efficiency and distributional balance across consumers, (ii) total consumer utility, and (iii) social welfare, which captures overall system well-being. We vary the level of fairness via a dimensionless parameter $\alpha \in [0,1]$, where $\alpha = 0$ corresponds to the status quo without fairness considerations and $\alpha = 1$ represents perfect enforcement under a given fairness criterion.

Using a stylized model with two consumer types differentiated by their ability to provide energy in response to incentive prices, i.e., flexibility, we provide a comprehensive analytical spectrum from $\alpha=0$ to $\alpha=1$ that characterizes how fairness constraints reshape optimal incentive prices and energy provided by consumers, and how performance changes as the fairness level changes. First, in the absence of any fairness considerations, a profit-only pricing framework favors consumers with greater response ability (higher flexibility), and yields decreasing CNW as the flexibility gap between consumers grows (Theorem~\ref{prop:profit_only}). Second, we prove that it is practically infeasible to enforce all three fairness criteria simultaneously, as doing so drives the aggregator’s profit to zero (Theorem~\ref{thm:impossibility}). Third, the choice of the fairness criterion matters. Different fairness criteria induce qualitatively different operating outcomes, and increasing the level of fairness $\alpha$ does not necessarily improve CNW, total consumer utility, or social welfare. We show that each fairness criterion induces a set of operating regimes, where a ``regime'' is a region of $\alpha$ over which the system performance measure (e.g., CNW) maintains a constant direction of change--either increasing, remaining constant, or decreasing. As $\alpha$ varies, the system may transition across different operating regimes.
\begin{itemize}
    \item Energy fairness exhibits four regimes where CNW and total consumer utility may increase or decrease, and the regimes that improve all system performance are attained only at moderate fairness levels, which must transition to a different regime with a higher $\alpha$ (Theorem~\ref{thm:energy_fairness}).
    \item Price fairness yields three regimes and does not benefit less flexible consumers in any regime. When $\alpha$ is large, less flexible consumers may even be excluded. Prior to this threshold, CNW may either increase or decrease, while social welfare continues to increase (Theorem~\ref{thm:price_fairness}).
    \item Utility fairness induces four possible regimes, which cause no harm to less flexible consumers and provide no benefit to more flexible consumers. While CNW may increase or decrease, social welfare does not increase under any regime (Theorem~\ref{thm:utility_fairness}).
\end{itemize}

Lastly, we conduct a case study using data from the \textit{iFlex} field experiment in Norway under a tiered pricing design. We show that our theoretical qualitative regime patterns can be extended to multi-consumer settings, which can also be effectively reduced to an equivalent two-consumer representation. When price response is calibrated using real-world energy consumer data, we first observe higher flexibility consumers corresponding to higher-income households. Furthermore, when the aggregator collects $80\%$ of the total baseline demand of all consumers, separately enforcing perfect demand, price, and utility fairness (i.e., $\alpha=1$) leads to increases in total consumer utility of $11.91\%$, $29.42\%$, and $50.47\%$, respectively, at the cost of profit reductions of $0.69\%$, $2.32\%$, and $6.45\%$.

For regulators, utilities, and aggregators seeking to operate VPP programs in a fair and economically viable manner, our findings imply the following takeaways.
\begin{enumerate}
    \item \textbf{Fairness must be incorporated into mechanism design}. Under profit-maximizing incentives, VPP operations tend to favor more flexible consumers, while less flexible consumers are less likely to participate, leading to unfairness regarding participation and benefit allocation. This highlights the need to account for fairness at the design stage rather than as an afterthought.
    \item \textbf{The choice of fairness metric is itself a policy decision}. Different fairness metrics define a different notion of what it means to be fair and therefore redistribute benefits in different ways. Policymakers should choose the fairness criterion carefully: (i) when prioritizing higher total consumer utility and social welfare, both energy fairness and price fairness can be appropriate with a moderate level of $\alpha$; (ii) from the consumer perspective, equalizing prices and energy provisions is not necessarily desirable, as doing so may reduce CNW and place a greater burden on less flexible consumers; and (iii) when the primary objective is to protect less flexible consumers from utility losses, utility fairness serves as an effective safeguard.
    \item \textbf{More enforcement is not always better}. Increasing the fairness level does not monotonically improve CNW, total consumer utility, or social welfare. Regulators should avoid one-size-fits-all mandates and instead calibrate both the fairness criterion and $\alpha$ to local system conditions, consumer heterogeneity, and policy objectives.
    \item \textbf{Fairness is a tool to compensate flat-rate payments}. Fairness-aware VPP incentives can function as a targeted subsidy, channeling program payments to less flexible consumers who face uniform flat-rate prices but unequal access to DER resources.
\end{enumerate}

\subsection{Related Literature}
To contextualize our work, we organize the related literature into three areas--fairness in pricing, power systems, and VPP operations. 

\paragraph{Fairness in Pricing.}
A substantial body of research establishes key frameworks and metrics for analyzing fairness in pricing outside the VPP context, providing important insights for our analysis. These studies examine how various fairness constraints--related to price, demand, consumer utility, and access--shape profit, consumer utility, and social welfare in diverse settings, including monopolistic markets \citep{price_discrimination}, vehicle-sharing systems \citep{vehicle_sharing}, and personalized pricing models \citep{kallus2021fairness}.

Some research studies focus on implementing pricing fairness based on specific consumer characteristics or market conditions. These works analyze diverse criteria, ranging from individualized fairness, which requires similar prices for consumers with similar characteristics \citep{das2022individual}, and group fairness, which limits price disparities across consumer groups \citep{chen2023personalized}. Further work addresses fairness in dynamic settings, including constraints across consumer groups and time periods \citep{cohen2025dynamic}, and the enforcement of utility fairness using contextual bandit algorithms \citep{chen2025utility}.

Our work shares the goal of examining fairness effects on system performance and consumer outcomes with the aforementioned studies, yet these insights are unexplored in VPPs comprehensively. The key distinction is that, unlike standard settings in which consumers only pay, VPP consumers provide energy and are compensated. This reciprocal interaction fundamentally alters the pricing structure and fairness considerations, necessitating a focused analysis.

\paragraph{Fairness in VPP operations.}
Fairness in VPP design is formalized through specific metrics and account for consumer heterogeneity. Existing fairness criterion for DERs relies on a Rawlsian perspective, enforcing fairness by maximizing the minimum consumer utility via social welfare maximization~\citep{aggregator_consider}. Heterogeneity across consumers is further quantified using inequality measures, such as Gini-based or variance-based indices, which guide VPP decision-making~\citep{define_a}. Beyond affecting individual participation, incorporating fairness at the consumer level alters intra-aggregator interactions, leading to quasi-concave games and fundamentally changing market outcomes~\citep{each_aggregator}. A study further extends the VPP framework to energy communities, where consumers can directly trade energy within the group, showing that embedding fairness directly at the governance level--through community agreements and market rules--can promote equitable outcomes~\citep{include_fairness}.

Recent work adopts a long-term perspective, demonstrating that embedding fairness as a weighted objective enables VPPs to balance efficiency, equity, and network constraints, thereby improving consumer retention and sustaining participation~\citep{fair_aggregation}. From a control perspective, fairness is further interpreted as requiring long-run prices or incentives to be independent of initial conditions--a property not satisfied by standard controllers and instead requiring carefully designed feedback mechanisms~\citep{VPP_fairness}.

Prior work embeds fairness into VPP models using a single fairness criterion imposed at a fixed level, either via hard constraints or weighted welfare objectives, yielding point solutions that illustrate fairness–efficiency tradeoffs under a chosen notion. Our work differs from and generalizes the existing literature in three ways. First, we formulate the aggregator–consumer interaction as a Stackelberg game with closed-form consumer responses. Second, we study multiple fairness definitions within a unified framework. Third, and most importantly, we characterize the entire spectrum of fairness enforcement by varying the fairness level continuously from profit-only pricing to full enforcement, identifying all feasible operating regimes and their transitions. This reveals non-monotonic welfare effects and regime changes that are invisible in fixed-level or single-metric formulations.

\paragraph{Fairness in power systems.}
In the power system context, fairness challenges arise not only from consumer heterogeneity but also from power grid structure, infrastructure deployment, and policy priorities. Because system operation is inherently shaped by network topology, certain nodes play a disproportionate role in maintaining voltage stability and supply–demand balance. In radial distribution networks, photovoltaic (PV) installations located toward the end of feeders are more likely to induce voltage rise; consequently, economically efficient modeling tends to curtail these PVs more frequently than those closer to substations to the main voltage level, leading to location-based fairness concerns. To mitigate such disparities, recent studies incorporate fairness into PV curtailment and voltage control by modifying voltage sensitivity matrices~\citep{fairness_in}, imposing curtailment-equality constraints~\citep{add_fairness}, or introducing fairness-weighted objectives~\citep{curtailment}, revealing an inherent trade-off between operational efficiency and consumer fairness and suggesting potential solutions by careful objective design or targeted incentive mechanisms~\citep{PV_for, similar_to}. Beyond voltage control, network structure also creates differences in outage duration and restoration priority across locations. Recent work embeds fairness into resilience indices alongside system performance~\citep{equally_dis} and employs local search over feasible post-fault network configurations to balance efficiency and fairness in restoration~\citep{min_outage}, while longer-term analyses further demonstrate that fairness-aware restoration strategies can improve both power recovery and economic resilience~\citep{restore_should}.

Power system infrastructure investment costs are ultimately reflected in higher electricity prices, which, if not carefully designed, can impose disproportionate burdens on low-income consumers~\citep{infrastructure_investment}. Similar fairness concerns arise in technology adoption, where large-scale AI computing can exacerbate regional inequities by creating uneven demand growth and inducing heterogeneous decarbonization pressures across locations~\citep{technology_AI}. Potential pathways to mitigating these disparities include the deployment of clean energy technologies, such as heat pumps, which have been shown to reduce energy equity gaps across income groups~\citep{technology_development}, as well as the design of tariff structures that explicitly account for consumer heterogeneity, including differences between solar and non-solar consumers~\citep{utility_set}.

\section{Framework and Preliminaries}
This section introduces the modeling and evaluation framework to study fairness throughout the paper.
Section~\ref{section21} introduces the basic VPP pricing model and consumer price response structure, formulated as a Stackelberg game between consumers and a VPP aggregator. Section~\ref {section22} defines the fairness metrics considered in this paper, followed by Section~\ref {section23}, which presents the performance measures used to evaluate system outcomes.

\subsection{Consumer and Aggregator Models}\label{section21}
We consider a single-period setting where consumers respond to incentive prices offered by a VPP aggregator to provide energy. By aggregating this provided energy, the aggregator can participate in upper-level markets--such as utility programs (e.g., Con Edison’s demand response initiatives~\citep{ConEd}) or wholesale electricity markets--to offer grid services and generate revenue~\citep{DOE_VPP}. We assume that the aggregator has full visibility into each consumer’s information, including DER status and price response behavior. This assumption allows us to isolate the impact of fairness considerations from uncertainties associated with prediction errors. In practice, DERs such as EVs and home batteries are typically monitored directly by aggregators~\citep{Energyhub_report}, while other response behaviors may be estimated using predictive models.

Let $D_i$ denote the provided energy of consumer $i\in [N]:=\{1,2,\ldots,N\}$ in response to an incentive price $p_i$ per unit energy. The provided energy comes from discharging home batteries or behavioral adjustments (e.g., changes in thermostat setpoints), which are accomplished at a cost that captures opportunity costs or discomfort from these actions. We denote the cost of providing energy $D_i$ by $C_i(D_i)$, thereby each consumer chooses $D_i$ to maximize individual utility 
\begin{equation}
    U_i: = p_iD_i - C_i(D_i),\label{utility_define}
\end{equation}
We assume $C_i(D_i)$ is strictly convex and non-decreasing, reflecting the increasing marginal cost of providing additional energy.

The \textbf{consumer’s utility maximization problem} is 
\begin{equation}
\begin{aligned}
     d_i(p_i) := \max_{D_i}&\quad U_i (p_i,D_i)\\
     \text{s.t. }&\quad0 \le D_i  \le \bar{D}_i.
\end{aligned}
\label{consumer}
\end{equation}
where $d_i(p_i)$ is the consumer’s price response function, $\bar D_i >0$ is the capacity of consumer~$i$, representing the maximum energy that consumer $i$ can provide. The capacity depends on the current operating conditions of consumer-owned DERs, such as home battery SoC or thermostat setpoint.

Then, for a given price $p_i$, the optimal solution $D_i^\ast$ to \eqref{consumer} defines $d_i(p_i)$. Under standard convexity assumptions, the optimal response can be written as
\begin{equation}
d_i(p_i) = 
\begin{cases}
  0, & \text{if } p_i \le C_i'(0)\\
  \left(C_i'\right)^{-1}(p_i), &\text{else if } 
  C_i'(0) <p_i < C_i'(\bar D_i),\\
  \bar{D}_i, & \text{otherwise}.
\end{cases}
\label{eq:response}
\end{equation}

The aggregator pays each consumer $p_i D_i$ and derives a gross benefit from the aggregated energy, denoted by a concave and non-decreasing function $V\big(\sum_{i\in [N]}D_i\big)$, reflecting diminishing marginal returns in the upper-level market. The aggregator has a maximum aggregated energy amount $D_s >0$, which is determined by the upper-level market's operating conditions. For instance, in a peak demand shaving event, the system operator first issues a target peak demand reduction amount, then the aggregator achieves this target by incentivizing consumers~\citep{coupon}.

The aggregator seeks to maximize its profit by setting the incentive prices $\{p_i\}$. The \textbf{aggregator's profit maximization} problem is given by
\begin{subequations}\label{aggregator_problem}
\begin{align}
    \max_{p_i\ge0}&\quad\Pi := V\Big(\sum_{i\in [N]}D_i\Big) - \sum_{i\in [N]}p_iD_i\label{aggregator_profit}\\
    \text{s.t. }&\quad D_i = d_i(p_i),~\forall i\in [N],\nonumber\\
    &\quad\sum_{i\in [N]}D_i \le D_{\mathrm{s}}.\nonumber
\end{align}
\end{subequations}
Given the prices set by the aggregator, consumers respond optimally according to \eqref{consumer}. This leader–follower structure constitutes a Stackelberg game, with the aggregator as the leader and consumers as followers. Since the aggregator’s objective is concave in prices and each consumer’s utility maximization problem is strictly concave in $D_i$, a unique Stackelberg equilibrium exists~\citep{SE}. The closed-form response \eqref{eq:response} allows us to analytically characterize optimal pricing and study how fairness constraints reshape equilibrium outcomes.

We summarize the main model assumptions as follows.
\begin{assumption}[\emph{Model assumptions}]~
\begin{itemize}
    \item The consumer's cost function $C_i(D_i)$ is strictly convex, non-decreasing, continuously differentiable, and $C_i(0)=0$.
    \item The aggregator's revenue function $V(\cdot)$ is concave, non-decreasing, and $V(0)=0$.
    \item The aggregator has full knowledge of each consumer's price-response function $d_i(p_i)$.
\end{itemize}
\end{assumption}

\subsection{Fairness Metrics}\label{section22}
We define three fairness metrics in VPP operations: energy, price, and utility.

\paragraph{(1) Energy fairness} evaluates how equitably consumers provide energy to the VPP program, measured by the absolute difference in the ratios of energy provided to capacity across consumers, $|{D_i}/{\bar{D}_i}-{D_j}/{\bar{D}_j}|$. In resource allocation contexts, particularly in energy systems, such a criterion is essential for guaranteeing equal access to grid participation. For instance, in EV charging systems, energy fairness dictates that all users should have equitable access to charging opportunities~\citep{demand_ev}. Similarly, during emergency load shedding events, system operators enforce equity outcomes by maintaining minimum energy access levels across consumers~\citep{demand_Shedding}. In the VPP context, where consumers are also energy providers, energy fairness resembles fair task allocation principles~\citep{demand_task}, ensuring that provided energy is equitably distributed. 

\paragraph{(2) Price fairness} requires that the incentive price per unit of energy offered to each consumer is similar. It is mathematically expressed as the pairwise difference $|p_i-p_j|$. This metric supports the principle that a uniform pricing structure represents a fair allocation--a perspective widely accepted in marketing and economic theory~\citep{price1}.
Although energy consumers have traditionally been viewed as passive recipients of electricity, VPPs transform them into active energy providers. Under this paradigm, price fairness aligns with the labor economics principle of ``equal pay for equal work''~\citep{price2}, ensuring that equivalent units of provided energy receive the same compensation, irrespective of the provider’s type or circumstance.

\paragraph{(3) Utility fairness} aims to ensure that all consumers attain a similar level of satisfaction--measured in terms of utility--regardless of their circumstances. This notion can be expressed mathematically by minimizing the difference in utility, e.g., $|U_i - U_j|$. In other words, utility fairness guarantees that participation yields benefits comparable across individuals, ensuring that no participant is significantly worse off in utility. Defining fairness as equality of utility has been regarded as normatively attractive.
\citet{kolm1997justice}, for example, notes that ``\emph{If the fundamental preference ordering can be represented by an ordinal utility function, this justice becomes equality of the utilities of the different persons}''.

Incorporating the fairness metrics defined above, we augment the aggregator profit maximization problem \eqref{aggregator_problem} by introducing each of them as constraints. To capture the trade-off between profit and fairness, we introduce a parameter $\alpha \in [0,1]$. Specifically, $\alpha = 0$ corresponds to a setting that maximizes profit without any consideration of fairness, whereas $\alpha = 1$ corresponds to a fully fairness-oriented setting, under which the fairness metrics are strictly equalized across all consumers. Let $M_i(p_i)$ denote a generic fairness metric as a function of price for consumer $i$.
\begin{itemize}
    \item Price fairness, $M_i(p_{i}) = p_i$,
    \item energy fairness, $M_i(p_i)=d_i(p_i)$,
    \item Utility fairness, $M_i(p_{i}) = U_i(p_{i})$.
\end{itemize}

We can then formally define the \textbf{fair aggregator profit maximization} problem as follows:
\begin{equation*}
    \begin{aligned}
    \max_{p_i}&~\Pi := V\Big(\sum_{i\in [N]}D_i\Big) - \sum_{i\in [N]}p_iD_i\\
    \text{s.t. }&~ D_i = d_i(p_i),~\forall i\in [N],\\
    &~\sum_{i\in [N]}D_i \le D_{\mathrm{s}},\\
    &~\lvert M_i(p_{i}) - M_{j}(p_{j}) \rvert \le (1-\alpha)\Delta,~\forall i,j \in [N],
    \end{aligned}
\end{equation*}
where $\Delta$ denotes the maximum disparity under the chosen fairness metric, i.e., $\Delta:=\max_{k,l\in [N]} \lvert M_k(p^\ast_{k}) - M_{l}(p^\ast_{l}) \rvert$. Here, $p_i^\ast$ denotes the optimal incentive prices obtained from the aggregator profit-only maximization problem~\eqref{aggregator_problem}, which serves as a baseline for evaluating fairness.

\subsection{Performance Measures}\label{section23}
To assess the impact of adding fairness criteria on the profit maximization VPP operations, we introduce the following three performance measures.

\paragraph{(1) Consumer Nash welfare} (CNW) is adapted from Nash social welfare, a classical economic metric widely used in resource allocation and auction design to capture equity among participants~\citep{NSW}. While Nash social welfare is defined as the product of individual utilities (equivalently, the sum of their logarithms), we focus exclusively on consumer equity and exclude the aggregator, whose profit has a fundamentally different economic structure. Accordingly, CNW is defined as $W_\mathrm{CNW} = \sum_{i\in [N]}\log \left(U_i \right)$, where $U_i>0$ is consumer $i$'s utility as defined in~\eqref{utility_define}. 
Higher CNW reflects more equitable utility allocation and is associated with greater consumer participation and retention~\citep{NSW_ref, fair_aggregation, VPP_fairness}, while extreme inequality drives CNW toward negative infinity, penalizing allocations that concentrate among a few consumers~\citep{NSW_resource}.

\paragraph{(2) Total consumer utility} is the sum of all individual utilities, i.e., $U = \sum_{i\in[N]}U_i$, which evaluates the VPP program performance from the consumer perspective, complementing the original problem, which focuses on maximizing the aggregator’s profit. While CNW effectively balances efficiency and equity, it places greater emphasis on equity across consumers and is therefore less sensitive to the magnitude of improvements. Total consumer utility complements this limitation by capturing the overall scale of consumer utility.

\paragraph{(3) Social welfare} is defined as the sum of aggregator profit (\ref{aggregator_profit}) and total consumer utility, i.e., $W_\mathrm{SW} = \Pi + U$. It captures the overall economic well-being of the VPP system and serves as an important metric for regulatory approval and institutional support~\citep{breyer2009regulation}. By jointly accounting for aggregator profit and consumer utility, social welfare reveals whether gains in consumer utility outweigh potential reductions in aggregator profit, thereby indicating whether the system-level outcome is socially beneficial.

Interpreting the performance measure from multiple perspectives is essential for understanding the impact of the fairness criteria. For instance, even if the aggregator's profit decreases, social welfare may increase, indicating higher total consumer utility ($U$) and better consumer outcomes. Improving CNW can enhance perceived fairness among consumers and, in turn, the program’s reputation. Moreover, supporting less flexible consumers (higher $U_1$) aligns with corporate social responsibility goals and may build long-term trust. These factors can ultimately contribute to the program’s profitability.

\section{Theoretical Analysis of a Stylized Model}
We begin by analyzing a stylized model with two consumers, i.e., $N = 2$. 
Each consumer $i \in [2]$ has a capacity $\bar{D}_i$, and without loss of generality, we assume $\bar{D}_1 < \bar{D}_2$. We assume that each consumer's cost function takes the form
\begin{equation}
C_i(D) = \tfrac{1}{2} a D^2 + (b - a \bar{D}_i) D,
\label{eq:cost_function}
\end{equation}
where $a > 0$ and $\frac{b}{a} > \max_i \bar{D}_i$. This specification ensures that the cost function is strictly convex ($C_i''(D) = a > 0$) and monotonically increasing ($C_i'(D) > 0$) on $[0, \bar{D}_i]$. 
Thus, larger provided energy $D$ becomes increasingly costly. Moreover, in this stylized setting, consumers with higher capacity $\bar{D}_i$ face lower marginal costs for the same response level. We treat this as a modeling assumption to capture one important class of consumers--those for whom larger available capacity makes a given response less burdensome, i.e., if $\bar{D}_1 < \bar{D}_2$, then $C_1(D) > C_2(D)$ for all $D \ge 0$. This is also supported by the economic literature, where consumers with larger capacity tend to have greater ability to adjust usage, and they are typically high-income consumers~\citep{ito2014consumers}. With this setting, consumer $2$ has a strong ability to provide energy in response to the incentive price, indicating greater flexibility. We also exclude a trivial case with $\bar{D}_1 = \bar{D}_2$, since the consumers are homogeneous and fairness constraints become unnecessary.

The cost function~\eqref{eq:cost_function} can be derived from the cost associated with reducing current energy consumption by $D$. Specifically, let $f(D)$ denote the comfort (or utility) from energy consumption level $D$. Following standard formulations in the price response literature~\citep{samadi2012advanced,yang2022prosumer}, comfort is assumed to be quadratic, strictly concave, and increasing:
$$
f(D) = -\tfrac{1}{2} a D^2 + b D,
$$
where $a > 0$ and $b > a \max_i \bar{D}_i$. Substituting this functional form into the definition of discomfort, i.e., the comfort loss from reducing energy consumption by $D$, yields
$$
C_i(D) = f(\bar{D}_i) - f(\bar{D}_i - D) = \tfrac{1}{2} a D^2 + (b - a \bar{D}_i) D.
$$

This formulation parallels the cost curves of conventional generators in electricity markets. In this analogy, the quadratic term $\tfrac{1}{2}a D_i^2$ reflects the increasing marginal cost associated with higher provided energy. The parameter $b-a\bar{D}_i$ resembles a start-up cost, representing the minimum incentive required for participation--typically dictated by the generator’s physical characteristics~\citep{wood2013power}.

Each consumer solves
$$
\max_{0 \le D_i \le \bar{D}_i} \; p_i D_i - C_i(D_i),
$$
The first-order condition $p_i = C_i'(D_i)$ yields
\begin{equation}
p_i = a D_i + (b - a \bar{D}_i)
\quad \Rightarrow \quad
D_i^{\ast} = \min\!\left(\bar D_i,\max\!\left(0,~\frac{p_i-b}{a}+\bar D_i\right)\right)
\label{linear_response}
\end{equation}
Note that to elicit participation from consumer~$i$, i.e., $D_i^\ast > 0$, the incentive price $p_i$ must satisfy $p_i > b - a\bar D_i$.

We assume that the aggregator's profit function is linear, defined as the electricity price from the upper-level market multiplied by the aggregated energy. Accordingly, we define
$$V(D_1 + D_2) = \pi (D_1 + D_2),$$
where $\pi$ denotes the upper-level market price. We assume $\pi > b - a \bar{D}_1$, ensuring that profitable aggregation remains feasible under consumer participation constraints. Additionally, we assume that $D_{\mathrm{s}} < \bar{D}_1 + \bar{D}_2$, meaning that the aggregator does not wish to aggregate an amount exceeding the total available capacity on the consumers’ side. The case $D_{\mathrm{s}} = \bar{D}_1 + \bar{D}_2$ is excluded, because when $D_1^\ast = \bar{D}_1$ and $D_2^\ast = \bar{D}_2$, the provided energy already achieves both energy fairness $\left(\sfrac{D_1^\ast}{\bar{D}_1} = \sfrac{D_2^\ast}{\bar{D}_2} = 1\right)$ and price fairness (since $p_1^\ast = p_2^\ast = b$).

\subsection{Implications of Profit-Only Optimization}
We first analyze the optimal solution to the profit maximization problem without fairness constraints, as defined in \eqref{aggregator_problem}, which serves as a baseline for evaluating the implications of each fairness criterion. The closed-form optimal solution is presented in Lemma~\ref{lemma:opt_sol}. All proofs are provided in the Appendix.
\begin{lemma}[Profit-Only Optimal Solution]
Assume $a>0$ and $\pi - b + a\bar D_i > 0$ for all $i \in [2]$. Define
\begin{equation*}
    D_1^{\dagger} = \min\!\left( \frac{\pi - b}{2a} + \frac{\bar{D}_1}{2},\ \bar{D}_1 \right)\quad\text{and}\quad D_2^{\dagger}=\frac{\pi - b}{2a} + \frac{\bar{D}_2}{2}.
\end{equation*}
Then, the optimal solution to the profit maximization problem~\eqref{aggregator_problem} is given by
\begin{enumerate}
    \item[(1)] If $D_1^{\dagger} + D_2^{\dagger} \le D_{\mathrm{s}}$, then $D_i^{*} = D_i^{\dagger}$ for all $i \in [2]$.
    \item[(2)] Otherwise,
    \begin{equation*}
        D_1^{\ast} = 
        \min\left(\bar D_1,\max\left(0,\frac{D_{\mathrm{s}}}{2} + \frac{\bar{D}_1 - \bar{D}_2}{4}\right)\right),
        \qquad
        D_2^{\ast} = D_{\mathrm{s}} - D_1^{\ast}.
    \end{equation*}
\end{enumerate}
\label{lemma:opt_sol}
\end{lemma}

Notice that there are two possible cases in Lemma \ref{lemma:opt_sol}. Case (1) is when the aggregated energy is not binding to $D_\mathrm{s}$, and case (2) is when the aggregated energy is binding. In case (1), the provided energy is limited by consumers rather than the aggregated energy from the upper-level market quota. This corresponds to a situation of a scarce consumer pool, in which attracting and retaining participants becomes important, thereby motivating the inclusion of fairness considerations that may facilitate participation. On the other hand, in case (2), the provided energy is constrained by the maximum aggregated energy limit. In this case, the aggregator must decide which consumers to collect energy from to avoid excluding consumers from participation, which naturally raises fairness concerns in the pricing schemes.

The following Theorem~\ref{prop:profit_only} motivates the need for fairness criteria by characterizing the consequences of profit-only optimization. It shows that without fairness considerations, the aggregator tends to favor consumers with higher capacity, resulting in an imbalance in the allocation of benefits.
\begin{theorem}[Effectiveness of Profit-Only Solution]
\label{prop:profit_only}
    The optimal solution to the profit-only maximization problem~\eqref{aggregator_problem} always satisfies $D_1^\ast < D_2^\ast$. Furthermore, under the model assumption $\bar D_1+\bar D_2 > D_\mathrm{s}$, when the aggregated energy constraint is binding, i.e., $D_1^\ast + D_2^\ast = D_\mathrm{s}$, except for the boundary case where $D_1^\ast = \bar{D}_1$, an increase in the disparity between consumer capacities $\bar{D}_2-\bar{D}_1$ leads to a decrease in the CNW but an increase in the total consumer utility and social welfare.
\end{theorem}

Theorem~\ref{prop:profit_only} shows that the optimal solution to the profit-only problem consistently favors consumer~$2$ with a higher capacity $\bar{D}$. This shows, in practice, that aggregators are likely to concentrate energy collection among a few more flexible consumers, who are typically consistent participants. This may leave limited participation opportunities and profit for less flexible consumers and raise severe fairness concerns.

Moreover, when the maximum aggregated energy $D_\mathrm{s}$ is limited, consumer heterogeneity has a pronounced impact on participation opportunities and system performance. The greater the disparity between the two consumers, as reflected in $\bar{D}_2- \bar{D}_1$, the lower the CNW, despite an increase in total consumer utility. This implies that the gains in utility are concentrated among a few more flexible consumers, while the resulting drop in CNW reflects worsening fairness. In contrast, when the aggregated energy constraint is not binding, each consumer's provided energy is determined only by their characteristics (parameters $\bar{D}$), and their outcomes are independent, making the disparity less consequential. This observation is intuitive, as fairness concerns tend to become salient under resource scarcity.

\subsection{Impossibility of Perfect Fairness Across Criteria}
We now analyze the conditions under which the aggregator seeks to satisfy all fairness criteria simultaneously. The ideal scenario is to achieve $1$-fairness ($\alpha = 1$) simultaneously in terms of energy, price, and utility. However, in Theorem~\ref{thm:impossibility}, we show that perfect fairness is attainable only in limited and impractical scenarios. Even then, it yields zero profit, rendering the solution economically meaningless.

\begin{theorem}[\emph{Impossibility of Perfect Fairness}]
Perfect fairness across energy, price, and utility criteria is achieved only when $p_1 = p_2 \le b-a\bar{D}_2$, which results in zero profit.
\label{thm:impossibility}
\end{theorem}
\begin{proof}{Proof.}
We first observe that \emph{1-price fairness} is achieved if-and-only-if there exists a uniform price $p$ such that $p_1 = p_2 = p$. Under a uniform price $p$, \emph{1-energy fairness} is trivially satisfied for $p \le b-a\bar D_2$, since $d_1(p)=d_2(p)=0$. 
For $b-a\bar D_2 < p \le b-a\bar D_1$, we have 
$d_2(p)-d_1(p)=\frac{p-b}{a}+\bar D_2>0$. 
Finally, for $p > b-a\bar D_1$, the difference remains strictly positive, given by 
$d_2(p)-d_1(p)=\bar D_2-\bar D_1>0$. Therefore, \emph{1-energy fairness} is achieved only when $p \le b-a\bar D_2$. When $p \le b-a\bar{D}_2$, both consumers provide zero energy, hence \emph{1-utility fairness} is achieved with $U_1 = U_2 = 0$. Therefore, all fairness criteria can be achieved simultaneously when $p\le b-a\bar{D}_2$, resulting in $D_1(1)=D_2(1)=0$ and zero profit.\hfill\Halmos
\end{proof}

As shown in Theorem~\ref{thm:impossibility}, achieving all fairness criteria simultaneously is theoretically infeasible and economically impractical. Therefore, in the following sections, we examine each fairness criterion separately.

We analyze the impact of fairness by comparing the solution of the fairness-constrained problem to the solution obtained from the profit-only optimization problem~\eqref{aggregator_problem}. Specifically, we examine how the decision variables shift as fairness constraints are introduced, and derive the system performance with all feasible fairness level~$\alpha$ for each criterion. Let $\left(p_1(\alpha), p_2(\alpha)\right)$ denote the optimal prices, and $\left(D_1(\alpha), D_2(\alpha)\right)$ the corresponding optimal provided energy, for a given fairness level~$\alpha$. These solutions enable us to quantify how incorporating each fairness metric affects the CNW, total consumer utility, and social welfare as $\alpha$ varies. Note that incorporating any fairness criterion cannot increase the aggregator’s profit, as fairness constraints reduce the feasible region.

\subsection{Energy Fairness}
Energy fairness affects the system performance measure, depending on the system parameters $(\pi,a,b,\bar{D}_1,\bar{D}_2, D_\mathrm{s})$. The following theorem systematically describes the full spectrum of performance measures across four regimes. 
\begin{theorem}
    [Energy fairness Under $N=2$] \label{thm:energy_fairness}
   Under the energy fairness criterion, four distinct regimes may emerge as $\alpha$ varies from $0$ to $1$, characterized by the directions of change in consumer utility $(U_i)$, CNW $(W_\mathrm{CNW})$, total consumer utility $(U)$, and social welfare $(W_\mathrm{SW})$. Transitions occur only along the arrows shown in the diagram. Depending on the parameters $(\pi,a,b,\bar{D}_1,\bar{D}_2,D_\mathrm{s})$, the system may start in any of the four regimes and may end in Regimes $2$, $3$, or $4$.

\vspace{3mm}
\begin{center}
\begin{tikzpicture}[
    node distance=10mm and 15mm,
    box/.style={
        draw, rounded corners, thick,
        text width=60mm,           
        align=center,
        inner xsep=6pt, inner ysep=3pt
    },
    arr/.style={-{Stealth[length=3.5mm]}, thick, shorten >=-1pt, shorten <=-1pt}
]
\node[box] (R1) {\textbf{Regime 1}\\
$U_1 -,\ U_2\uparrow,\ W_\mathrm{CNW}\uparrow,\ U\uparrow,\ W_\mathrm{SW}\uparrow$};

\node[box, right=of R1] (R2) {\textbf{Regime 2}\\
$U_1\downarrow,\ U_2\uparrow,\ W_\mathrm{CNW}\downarrow,\ U\uparrow,\ W_\mathrm{SW}\uparrow$};

\node[box, below=of R1] (R3) {\textbf{Regime 3}\\
$U_1\downarrow,\ U_2\uparrow,\ W_\mathrm{CNW}\downarrow,\ U\downarrow,\ W_\mathrm{SW}\downarrow,\ $};

\node[box, right=of R3] (R4) {\textbf{Regime 4}\\
$U_1\uparrow,\ U_2\downarrow,\ W_\mathrm{CNW}\uparrow,\ U\downarrow,\ W_\mathrm{SW}\downarrow,\ $};

\draw[arr] (R1) -- (R2);
\draw[arr] (R1) -- (R3);
\end{tikzpicture}
\end{center}
Here, $\uparrow$ (resp. $\downarrow$) indicates an increase (resp. decrease) with respect to $\alpha$. \text{--} denotes remaining constant at a positive value.
\end{theorem}

Theorem~\ref{thm:energy_fairness} reveals the full system spectrum and transition relationships when $\alpha$ varies under different parameter settings, which can be categorized into four possible regimes. Figure~\ref{fig:demand_eq} shows an example of the regimes and their transitions. In Regimes~$1$, $2$, and $3$, the outcomes favor the more flexible consumer (consumer $2$) while disadvantaging--or at least not benefiting--the less flexible consumer (consumer $1$) in terms of utility. This occurs because the more flexible consumer possesses a higher capacity, $\bar{D}$, which necessitates a greater provided energy to balance the energy ratio required to satisfy the energy fairness condition.

Among all the regimes, Regime~$1$ exhibits the best system performance across all measures; however, it does not persist--with a higher $\alpha$, the system inevitably transitions to Regime~$2$ or $3$, as shown in Figure~\ref{fig:demand_eq}. Among them, Regime $3$ is the worst regime, as it benefits only the more flexible consumer’s utility while worsening all other performance measures. Therefore, when parameter settings place the system in Regime~$1$, an appropriate choice of $\alpha$ is necessary to prevent entering Regime~$3$. Notably, in Regime~$1$, the less flexible consumer provides all the energy yet receives no benefit from the energy fairness criterion. This raises concerns about whether these regimes are truly desirable, which motivates the exploration of other fairness metrics. In Regime~$2$, the more flexible consumer still gains utility, but unlike in Regime~$3$, this gain dominates the utility loss of the less flexible consumer and the aggregator's profit reduction, resulting in an overall increase in total consumer utility and social welfare. As Regimes~$1$ and~$2$ require $D_1^* = \bar D_1$, they are less likely to arise as starting regimes. Regime~$4$ mirrors Regime~$2$ in the opposite direction: the less flexible consumer’s utility increases, but this gain cannot outweigh the utility loss of the more flexible consumer, because the latter's utility is more strongly affected when the fairness level changes. Since $D_1^\ast < D_2^\ast$ from Theorem~\ref{prop:profit_only}, CNW increases in Regimes~$1$ and $4$, where the less flexible consumer’s utility is non-decreasing.

\begin{figure}[htbp]
\FIGURE{
    \includegraphics[width=\textwidth]{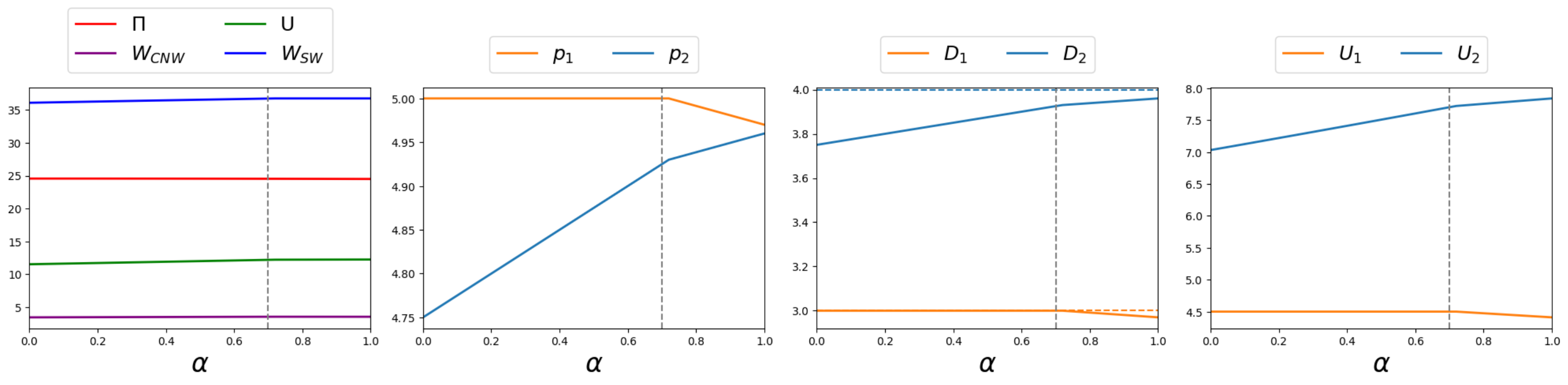}
}
{Energy fairness. \label{fig:demand_eq}
\vspace{.3cm}}
{The gray dashed lines distinguish the regimes (Regime~$1$ and $2$), shown sequentially from left to right. Parameters are set to $a = 1$, $b = 5$, $\pi = 8.5$, $\bar D_1 = 3$, $\bar D_2 = 4$, and $D_{\mathrm{s}} = 6.93$.
}
\end{figure}

\subsection{Price Fairness}
Similar to energy fairness, the implications of price fairness also critically depend on the system parameters. Theorem~\ref{thm:price_fairness} systematically characterizes the full spectrum and the implications of the fairness criterion across three distinct regimes. Transitions between regimes occur under specific conditions.

\begin{theorem}[Price Fairness Under $N=2$]
Under the price fairness criterion, three distinct regimes may emerge as $\alpha$ varies from $0$ to $1$, characterized by the directions of change in consumer utility $(U_i)$, CNW $(W_{\mathrm{CNW}})$, total consumer utility $(U)$, and social welfare $(W_{\mathrm{SW}})$. Transitions occur only along the arrows shown in the diagram. Depending on the parameters $(\pi,a,b,\bar{D}_1,\bar{D}_2,D_\mathrm{s})$, the initial and terminal regimes may each be any of the three regimes.
\begin{center}
\begin{tikzpicture}[
    node distance=15mm and 7mm,
    box/.style={
        draw, rounded corners, thick,
        text width=42mm,           
        align=center,
        inner xsep=6pt, inner ysep=3pt
    },
    arr/.style={-{Stealth[length=3.5mm]}, thick, shorten >=-1pt, shorten <=-1pt}
]
\node[box] (R1) {\textbf{Regime 1}\\
$U_1-,\ U_2\uparrow,\ W_{\mathrm{CNW}}\uparrow ,$ \\ $U\uparrow,\ W_{\mathrm{SW}}\uparrow$};

\node[box, right=of R1] (R2) {\textbf{Regime 2}\\
$U_1\downarrow,\ U_2\uparrow,\ W_{\mathrm{CNW}}\downarrow,$\\ $U\uparrow,\ W_{\mathrm{SW}}\uparrow$};

\node[box, right=of R2, text width=54mm] (R3) {\textbf{Regime 3}\\
$U_1=0,\ U_2-,\ W_{\mathrm{CNW}}=-\infty,$\\ $U-,\ W_{\mathrm{SW}}-$};

\draw[arr] (R1) -- (R2);
\draw[arr] (R2) -- (R3);
\end{tikzpicture}
\end{center}
Here, $\uparrow$ (resp. $\downarrow$) indicates an increase (resp. decrease) with respect to $\alpha$, \text{--} denotes remaining constant at a positive value, and $0$ (resp. $-\infty$) indicates remaining constant at zero (resp. negative infinity).
\label{thm:price_fairness}
\end{theorem}

Theorem~\ref{thm:price_fairness} characterizes the outcomes of each regime, and Figure~\ref{fig:price_eq} illustrates the corresponding cases. Overall, price fairness does not harm the more flexible consumer (consumer~$2$). Rather, it either benefits or leaves them unaffected. In contrast, it consistently disadvantages the less flexible consumer (consumer~$1$). As the system transitions across regimes when $\alpha$ gradually increases, the outcomes for consumer~$1$ progressively worsen or remain unchanged.

More specifically, in Regime~$1$, price fairness benefits only consumer~$2$ (the more flexible consumer) by increasing their incentive price. In Regime~$2$, as the price fairness constraint becomes tighter due to a higher $\alpha$, consumer~$2$’s price--and consequently their provided energy and utility--rises further. However, this improvement comes at the expense of consumer~$1$, whose price, provided energy, and utility all decline. Finally, in Regime~$3$, the system effectively excludes consumer~$1$ by further reducing their price and discouraging participation, resulting in zero utility gain for consumer~$1$. In conclusion, these results indicate that enforcing price fairness systematically disadvantages the less flexible consumer, thereby failing to achieve genuine fairness.

\begin{figure}[htbp]
\FIGURE{
    \includegraphics[width=\textwidth]{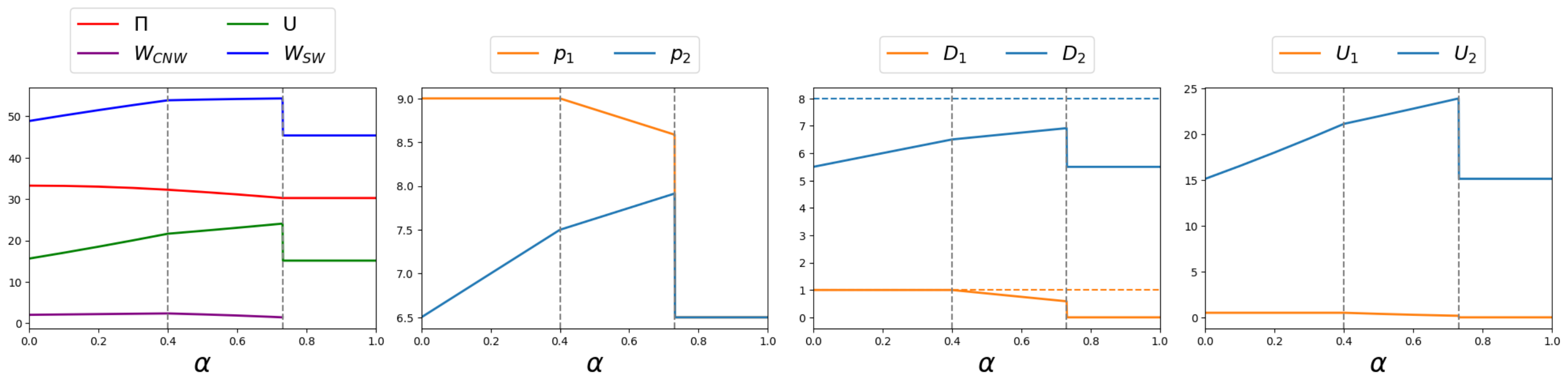}
}
{Price fairness \label{fig:price_eq}
\vspace{.3cm}}
{The gray dashed lines distinguish the three regimes (Regime~$1$, $2$, and $3$), shown sequentially from left to right. 
If $W_{\mathrm{CNW}}$ does not appear in the plot, it indicates that it attains $-\infty$. 
Parameters are set to $a = 1$, $b = 9$, $\pi = 12$, $\bar D_1 = 1$, $\bar D_2 = 8$, and $D_{\mathrm{s}} = 8$.
}
\end{figure}

\subsection{Utility Fairness}
Lastly, Theorem~\ref{thm:utility_fairness} characterizes the possible regimes under utility fairness across $\alpha$ and describes how these regimes transition across $\alpha$.

\begin{theorem}[Utility Fairness Under $N=2$]
Under the utility fairness criterion, four distinct regimes may emerge as $\alpha$ varies from $0$ to $1$, characterized by the directions of change in consumer utility $(U_i)$, CNW $(W_{\mathrm{CNW}})$, total consumer utility $(U)$, and social welfare $(W_{\mathrm{SW}})$. Transitions occur only along the arrows shown in the diagram. Depending on the parameters $(\pi,a,b,\bar{D}_1,\bar{D}_2,D_\mathrm{s})$, the system may start in Regime $1$, $3$, or $4$, and may end in Regime $1$, $2$, or $4$.
~\\
\begin{center}
\begin{tikzpicture}[
    node distance=10mm and 15mm,
    box/.style={
        draw, rounded corners, thick,
        text width=60mm,           
        align=center,
        inner xsep=6pt, inner ysep=3pt
    },
    arr/.style={-{Stealth[length=3.5mm]}, thick, shorten >=-1pt, shorten <=-1pt}
]
\node[box] (R1) {\textbf{Regime 1}\\
$U_1\uparrow,\ U_2\downarrow,\ W_{\mathrm{CNW}}\uparrow,\ U\downarrow,\ W_{\mathrm{SW}}\downarrow$};

\node[box, right=of R1] (R2) {\textbf{Regime 2}\\
$U_1\uparrow,\ U_2\downarrow,\ W_{\mathrm{CNW}}\uparrow,\  U\uparrow,\ W_{\mathrm{SW}}\downarrow$};

\node[box, below=of R1] (R3) {\textbf{Regime 3}\\
$U_1-,\ U_2\downarrow,\ W_{\mathrm{CNW}}\downarrow,\  U\downarrow,\ W_{\mathrm{SW}}\downarrow$};

\node[box, right=of R3] (R4) {\textbf{Regime 4}\\
$U_1\uparrow,\ U_2-,\ W_{\mathrm{CNW}}\uparrow,\ U\uparrow,\  W_{\mathrm{SW}}-$};

\draw[arr] (R1) -- (R2);
\draw[arr] (R1) -- (R3);
\draw[arr] 
  ([xshift=-1mm,yshift=0mm]R1.south east)
  -- ([xshift=+1mm,yshift=0mm]R4.north west);
\draw[arr] 
  ([xshift=1mm,yshift=0mm]R2.south west)
  -- ([xshift=-1mm,yshift=0mm]R3.north east);
\draw[arr] (R3) -- (R4);
\draw[arr] (R2) -- (R4);
\end{tikzpicture}
\end{center}
Here, $\uparrow$ (resp. $\downarrow$) indicates an increase (resp. decrease) with respect to $\alpha$ and \text{--} denotes remaining constant at a positive value.
\label{thm:utility_fairness}
\end{theorem}

Theorem~\ref{thm:utility_fairness} shows that there are four possible regimes and clearly characterizes the transitions as $\alpha$ varies, and Figure~\ref{fig:utility_eq} illustrates an example for the corresponding regimes. In Theorem~\ref{thm:utility_fairness}, all cases either benefit (or at least do not harm) more flexible consumer (consumer~$1$) and worsen (or at least do not harm) less flexible consumer (consumer~$2$) in terms of utility. This occurs because consumer~$2$ initially has strictly higher utility (as shown in Figure~\ref{fig:utility_eq}), and the utility fairness constraint effectively transfers utility from the higher-utility consumer to the lower-utility consumer. This is consistent with the normative principle that, when some loss of efficiency is unavoidable, fairness considerations prioritize improving the welfare of the less flexible consumer.

In Regime~$1$, the aggregator increases the utility of consumer~$2$ while decreasing that of consumer~$1$, leading consumer~$1$ to provide more energy and consumer~$2$ to provide less energy. However, the utility loss of consumer~$2$ dominates the utility gain of consumer~$1$, reducing total consumer utility. In Regime~$2$, the same directional pattern persists, but the gain to consumer~$1$ exceeds the loss to consumer~$2$, so both total consumer utility and CNW increase. In Regime~$3$, consumer~$1$ hits its capacity, so the aggregator keeps consumer~$1$'s utility unchanged and reduces consumer~$2$'s utility to satisfy utility fairness, which decreases both total consumer utility and CNW. By contrast, in Regime~$4$, even though consumer~$1$ remains at its capacity boundary, the aggregator increases consumer~$1$'s utility, inducing an increase in total consumer utility and CNW.
Note that Regime~$3$ and Regime~$4$ may appear beneficial for consumer~$1$, but they require consumer~$1$ to provide all of its capacity. This is similar to the energy fairness criterion. However, unlike energy fairness--under which consumer~$1$ may increase provided energy without gaining utility--utility fairness can translate this increased provided energy into an actual utility improvement for consumer~$1$.

\begin{figure}[htbp]
\FIGURE{
    \includegraphics[width=\textwidth]{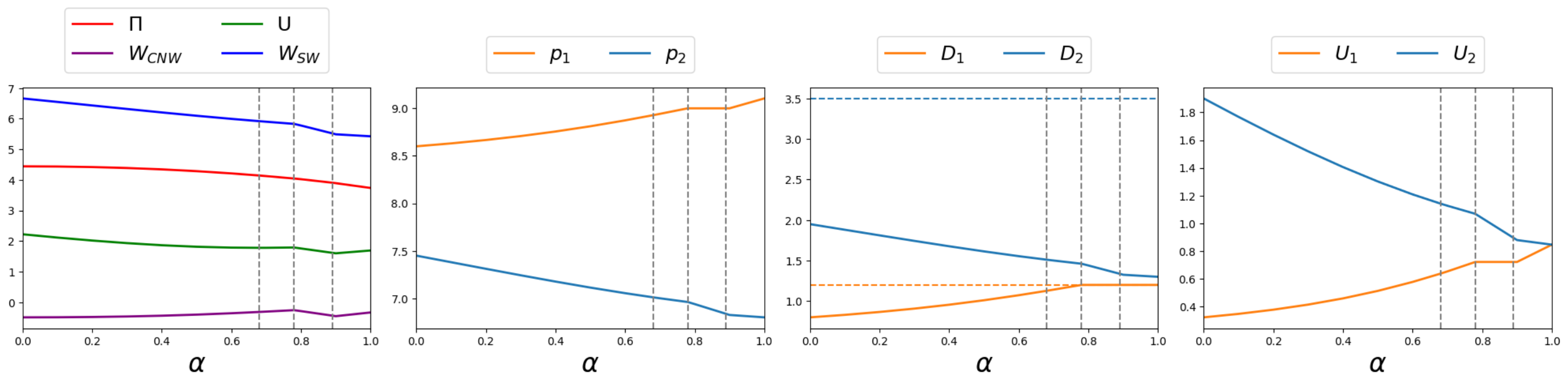}
}
{Utility Fairness. \label{fig:utility_eq}
\vspace{.3cm}}
{The gray dashed lines distinguish the four regimes (Regime~$1$, $2$, $3$, and $4$), shown sequentially from left to right. Parameters are set to $a = 1$, $b = 9$, $\pi = 9.4$, $\bar D_1 = 1.2$, $\bar D_2 = 3.5$, and $D_{\mathrm{s}} =4.5$.
}
\end{figure}

\section{Numerical Analysis with Multiple Consumers}
We extend our analysis to environments with more than two consumers. The primary objective of this section is to investigate whether additional phenomena arise in the three-consumer setting that do not appear when only two consumers are present. We conduct this numerical analysis using both an off-the-shelf optimization solver (Pyomo) and grid search, which provides cross-validation of the results. Further implementation details are provided in Appendix~\ref{appendix:experimental_details}.

\paragraph{Energy fairness.}
Figure~\ref{fig:demand_eq3} shows the results under the energy fairness constraint. The system exhibits the same regime transition from Regime~$1$ to Regime~$2$ when $\alpha$ increases as in the two-consumer case. Specifically, in Regime~$1$ (i.e., for $\alpha < 0.68$), only the most flexible consumer (consumer~$3$ with the largest capacity) increases its provided energy, while the provided energy of the other consumers remains unchanged. Beyond the threshold $\alpha = 0.68$, the system transitions to Regime~$2$, in which consumer~$3$ continues to increase provided energy, accompanied by provided energy reductions from the remaining consumers. This qualitative behavior mirrors the two-consumer setting, with the only difference being that multiple consumers now jointly reduce the provided energy. 

As a result, the three-consumer system can be interpreted as an effective two-consumer system by aggregating consumers~$1$ and~$2$, which share the same changing direction for provided energy and collectively offset the increasing provided energy from consumer~$3$. This equivalence implies that although more consumers may introduce more Regimes due to different individual energy consumption patterns, the energy fairness results derived in the two-consumer setting can be extended to multi-consumer systems with the same system performance change. Individual consumers may differ in the amount of energy they provide, but consumers whose provided energy changes in the same direction can be appropriately grouped, and the aggregate change in provided energy at the group level follows the same pattern as in the two-consumer case.

\begin{figure}[htbp]
\FIGURE{
    \includegraphics[width=\textwidth]{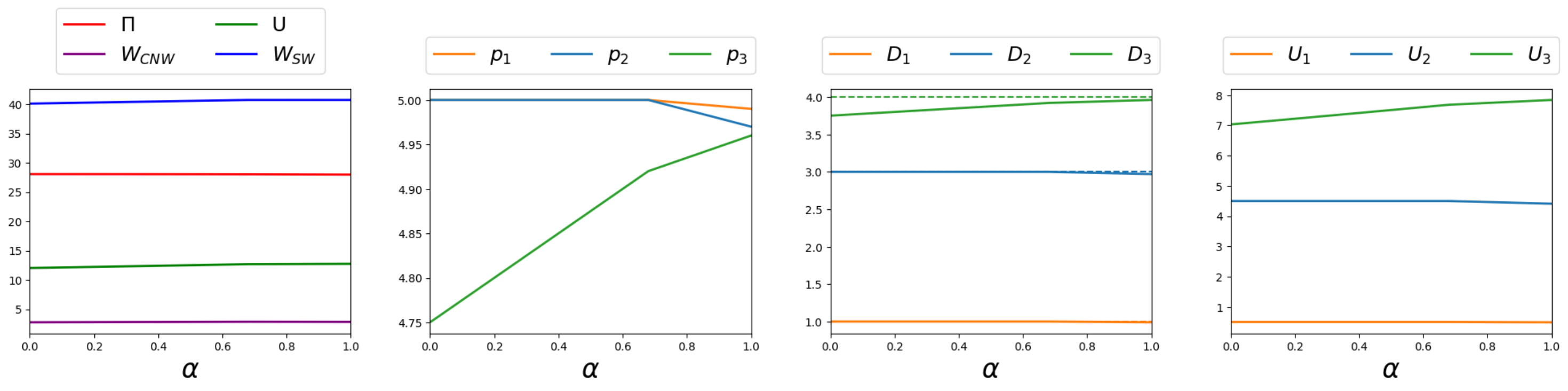}
}
{Energy fairness \label{fig:demand_eq3}
\vspace{.3cm}}
{Parameters are set to $a = 1$, $b = 5$, $\pi = 8.5$, $\bar D_1 = 1$, $\bar D_2 = 3$, $\bar D_3 = 4$, and $D_{\mathrm{s}} = 7.92$.}
\end{figure}

\paragraph{Price fairness.}
Figure~\ref{fig:price_eq3} presents the results under the price fairness constraint. A noteworthy observation emerges when we focus on the least flexible consumer (consumer~$1$ with the smallest capacity) and the most flexible consumer (consumer~$3$ with the largest capacity). Up to $\alpha = 0.7$, the system exhibits regime patterns analogous to Regimes~$1$ and~$2$ in the two-consumer setting. Specifically, $U_1$ remains constant until $\alpha = 0.1$ (Regime~$1$), after which it begins to decrease (Regime~$2$). Around $\alpha = 0.7$, consumer~$1$'s utility reduces to $0$, mirroring the behavior observed in the two-consumer case in Regime~$3$.

For $\alpha > 0.7$, the problem can be interpreted as a reduced two-consumer system consisting of consumers~$2$ and~$3$. This leads to a regime analogous to Regime~$2$ in the two-consumer setting, where the less flexible consumer (consumer~$2$) experiences a decline in utility, while the more flexible consumer (consumer~$3$) gains utility. In summary, although a larger number of consumers introduces additional regimes, the qualitative implication of price fairness remains unchanged. More flexible consumers do not incur any loss, whereas less flexible consumers do not obtain any benefit. Thus, the three-consumer case reinforces the conclusion already established in the two-consumer setting.

\begin{figure}[htbp]
\FIGURE{
    \includegraphics[width=\textwidth]{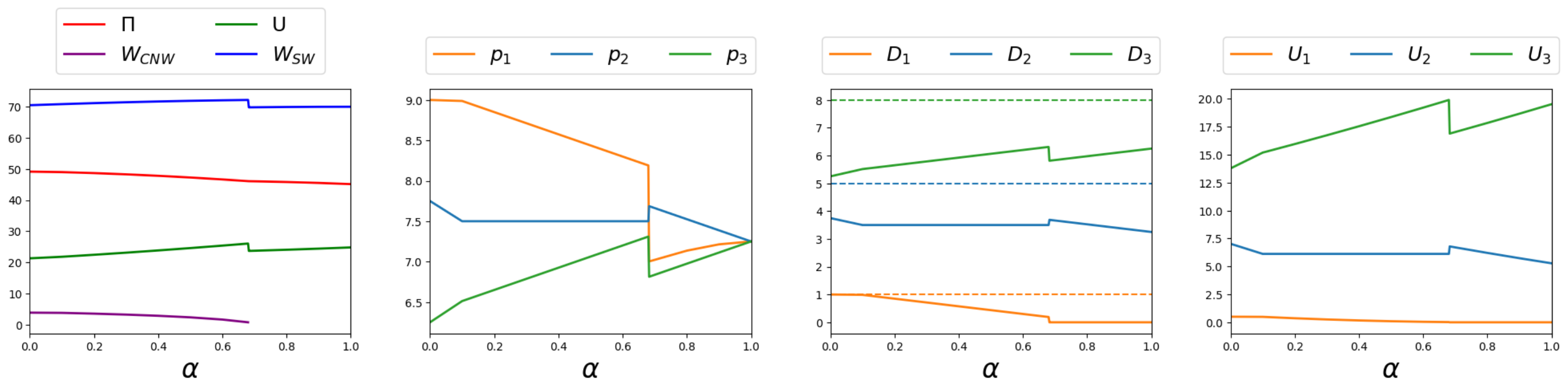}
}
{Price fairness \label{fig:price_eq3}
\vspace{.3cm}}
{Parameters are set to $a = 1$, $b = 9$, $\pi = 12$, $\bar D_1 = 1$, $\bar D_2 = 5$, $\bar D_3 = 8$, and $D_{\mathrm{s}} = 10$.
}
\end{figure}

\paragraph{Utility fairness.}
Figure~\ref{fig:utility_eq3} presents the results under the utility-fairness constraint. When we focus on the least flexible consumer (consumer~$1$ with the lowest capacity) and the most flexible consumer (consumer~$3$ with the highest capacity), we observe interesting patterns. Up to the point at which the utility of consumer~$1$ converges to that of consumer~$2$ (around $\alpha=0.9$), the regime transitions still mirror those in the two-consumer case. Specifically, Regime~$1$ persists until $\alpha = 0.5$, followed by Regime~$3$ in the interval $\alpha \in [0.5, 0.7)$, and finally Regime~$4$ thereafter up to $\alpha=0.9$. 

The difference between the three-consumer case and the two-consumer case emerges only beyond $\alpha = 0.9$. In this region, the utility of consumer~$1$ becomes equal to that of consumer~$2$. Beyond this point, their utilities increase together, while the utility of consumer~$3$ declines. If consumers~$1$ and~$2$ are regarded as a single aggregated consumer, the resulting pattern resembles Regime~$2$ in the two-consumer case. More precisely, in this regime, the aggregated utility of the system increases. The utilities of consumers~$1$ and~$2$ rise, while that of consumer~$3$ declines. Therefore, although the three-consumer setting exhibits additional regimes not present in the two-consumer case, the patterns remain similar. In particular, under utility fairness, less flexible consumers benefit, whereas more flexible consumers experience utility losses.

\begin{figure}[htbp]
\FIGURE{
    \includegraphics[width=\textwidth]{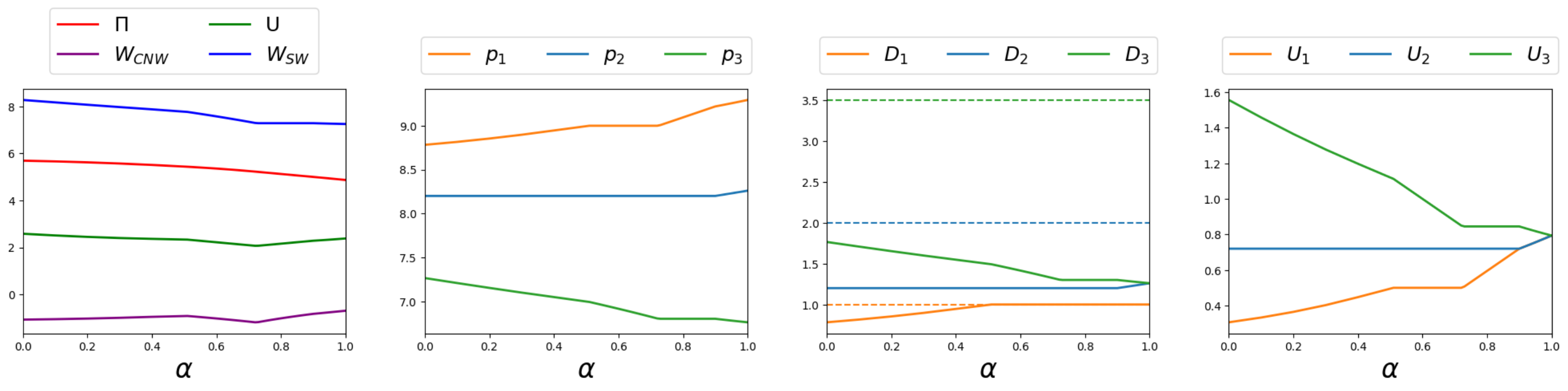}
}
{Utility fairness \label{fig:utility_eq3}
\vspace{.3cm}}
{Parameters are set to $a = 1$, $b = 9$, $\pi = 9.4$, $\bar D_1 = 1$, $\bar D_2 = 2$, $\bar D_3 = 3.5$, and $D_{\mathrm{s}} = 6.4$.
}
\end{figure}

\section{Case Study}
\label{sec:case_study}
In this section, we illustrate how our fairness framework can be applied in a real-world setting, using data from the \emph{iFlex} field experiment\footnote{\url{https://zenodo.org/records/8248802}} conducted by the Norwegian transmission system operator \emph{Statnett}. The primary objective of the \emph{iFlex} project was to quantify households’ (consumers') price sensitivity and implicit flexibility, capturing how consumers adjust electricity consumption when exposed to short-term incentive prices. To this end, participating households were exposed to experimentally designed hourly incentive price signals on selected days during two winter periods ($2019$–$2020$ and $2020$–$2021$). These price signals ranged from $2$ to $30$ NOK/kWh and followed pre-specified intra-day profiles, reflecting different incentive designs in the experiment. This dataset also includes household survey information, allowing us to examine the demographic profiles of each household.
A more detailed description of the data is provided in \citet{hofmann2023rich}. 

In this case study, we adopt a practical modification to the pricing policy. Rather than implementing a fully personalized pricing scheme, we group the $N=1{,}233$ households into three clusters and consider a tiered pricing scheme. Under this scheme, distinct prices are assigned across clusters. This approach is consistent with tier-based (or segment) pricing schemes commonly adopted in practice.\footnote{For instance, \emph{OhmConnect} employs a performance-based status system for residential participants. See its official documentation: \url{https://www.ohmconnect.com/help/en_us/what-is-an-ohmconnect-status-level-rknMwYsNu}.}. We then examine how different fairness criteria can be implemented and evaluated within this pricing framework.

We first estimate the parameter of the consumer model and provide detailed explanations of the pricing framework (Section~\ref{sec:demand_model}) with the preprocessing steps provided in Appendix~\ref{appendix:case_study_preprocessing}. We then evaluate the consequences of the fairness criteria proposed in this paper and demonstrate how their empirical implications align with the theoretical insights developed in earlier sections (Section~\ref{sec:results}).


\subsection{Parameter Estimation}
\label{sec:demand_model}
In this section, we first estimate the linear response parameters $a$ and $b$ in \eqref{linear_response} using experimental price–consumption data, and then estimate the capacity. Our analysis focuses on households that participated exclusively in Phase~$2$ of the \emph{iFlex} field experiment and were randomly assigned to the price treatment group. We further restrict the sample to participants who completed the post-experiment survey, which provides household characteristics. The detailed sample selection procedure is described in Appendix~\ref{appendix:case_study_preprocessing}.

\paragraph{Estimation of response parameters.} Note that in our stylized model, the price faced by individual consumer $i$ who provides energy $D_i$ is given by
$p_i = a D_i + b - a \bar{D}_i$.
Because the response behavior may vary across hours, we estimate these parameters separately for each hour and identify the hour with the most accurate response estimation. Importantly, this step does not require prior estimation of capacity $\bar D_i$, instead, we estimate capacity after selecting the hour, which is subsequently used for clustering and further analysis. In the experimental data, the provided energy is not directly observed. Instead, we observe consumers' electricity consumption, denoted by $Q_i$ for consumer $i$. We interpret the provided energy as the deviation from capacity, $\bar{D}_i - Q_i$, due to incentive prices. Substituting this expression into the price-consumption relationship at a given hour yields
\begin{equation}
\label{eq:OLS}
p_i = a(\bar{D}_i - Q_i) + b - a \bar{D}_i = -a Q_i + b,
\end{equation}
which eliminates the need to estimate capacity $\bar D_i$ when estimating the response function parameters, and allows the coefficients $a$ and $b$ to be identified directly from the observed pairs $(Q_i, p_i)$.

We estimate the coefficients $a$ and $b$ by regressing observed electricity consumption on prices during experimental days, as specified in~\eqref{eq:OLS}. Table~\ref{tab:hourly_ab} presents the estimated coefficients in which the estimate of $a$ is statistically significant ($p<0.05$).

\begin{table}[htbp]
\centering
\caption{Estimated price-consumption coefficients by hour}
\label{tab:hourly_ab}

\makebox[\textwidth][c]{%
\resizebox{\textwidth}{!}{%
\begin{tabular}{c|cccccccc}
\toprule
Hour & 8 & 9 & 12 & 13 & 14 & 15 & 19 & 20 \\
\midrule
$a$ (p-value)
& 0.0308 (0.0207)
& 0.0278 (0.0414)
& 0.0383 (0.00757)
& 0.0408 (0.00646)
& 0.0349 (0.0205)
& 0.0359 (0.0177)
& 0.0420 (0.00223)
& 0.0387 (0.00459) \\
$b$
& 4.9059
& 4.9864
& 4.6970
& 4.5686
& 4.4988
& 4.5287
& 5.2495
& 5.1243 \\
\bottomrule
\end{tabular}%
}}
\end{table}

We focus on hour $13$ (i.e., $12{:}00$--$13{:}00$), as it lies within the longest consecutive time window during which the estimated coefficient $a$ is statistically significant, and within this window, hour $13$ exhibits one of the largest positive estimates of $a$, indicating a meaningful price-consumption relationship during this period.

\paragraph{Estimation of capacity.} The capacity, $\bar{D}_i$, is estimated as the average electricity consumption of household $i$ during non-experiment days (i.e., days without incentive prices) for the same hour of the day (hour $13$), following \citet{hofmann2021households}. We then group individuals based on their capacity. Using the elbow method~\citep{kodinariya2013review}, we observe that the reduction in within-cluster variation slows significantly after three clusters. We therefore partition households into three clusters.
We label Clusters~$1$, $2$, and~$3$ in increasing order of their mean capacity. Table~\ref{tab:cluster_summary} summarizes the clustering results based on capacity.

We also analyze the characteristics of each cluster by examining the reported household income from the survey.
Figure~\ref{fig:income_cluster} shows the distribution of households across income categories within each cluster, where the shares sum to one for each cluster. Consistent with literature~\citep{ito2014consumers}, clusters with larger capacity $\bar{D}$ tend to include a higher share of high-income households, suggesting a positive association between capacity and income levels. Thus, tiered pricing can have effects across income levels, which is where fairness considerations become important.
As our theory shows, in the absence of fairness constraints, households in Cluster $3$--typically more high-income households--may provide more energy under profit-maximizing VPP strategies, enabling them to capture a disproportionately larger share of participation opportunities.

\begin{table}[htbp]
\centering
\begin{minipage}{0.4\textwidth}
\centering
\captionof{table}{Cluster statistics}
\label{tab:cluster_summary}

\setlength{\tabcolsep}{8pt}
\begin{tabular}{ccc}
\toprule
Cluster & Mean $\bar{D}$ & $\#$ of households \\
\midrule
1 & 0.907 & 505 \\
2 & 2.692 & 497 \\
3 & 4.991 & 231 \\
\bottomrule
\end{tabular}
\end{minipage}
\hfill
\begin{minipage}{0.55\textwidth}
\centering
\captionof{figure}{Income distribution by cluster}
\label{fig:income_cluster}

\includegraphics[width=\linewidth]{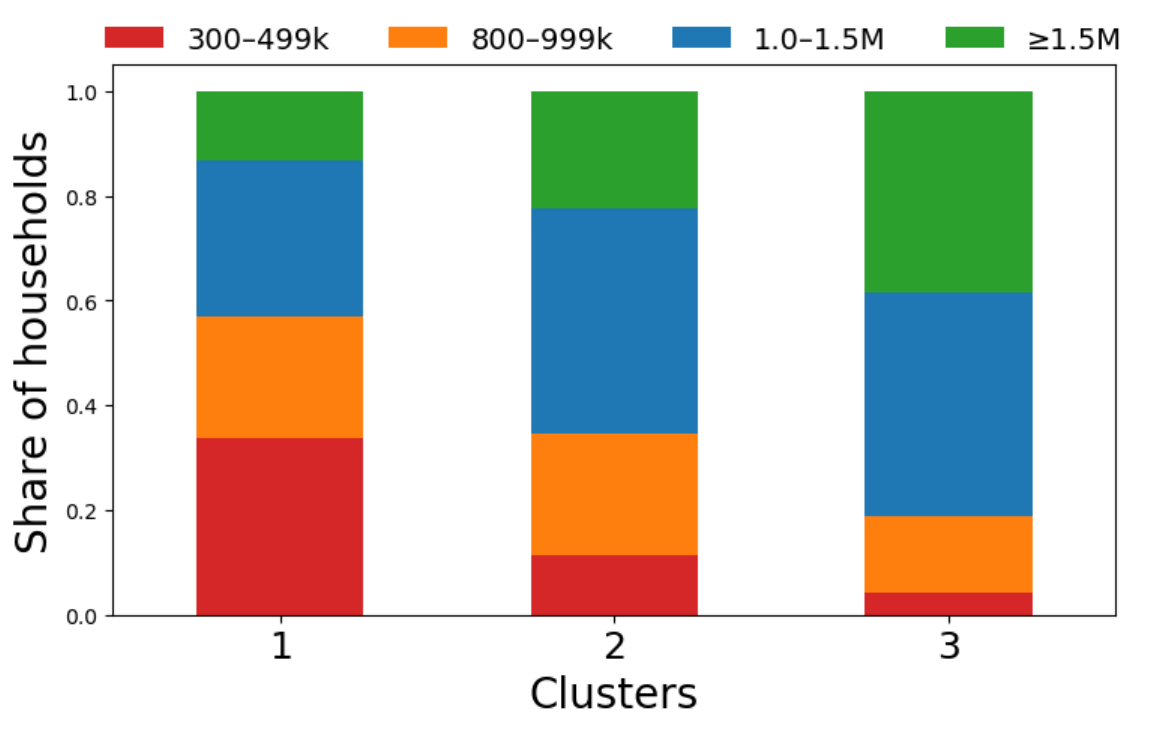}
\end{minipage}

\end{table}




\subsection{Results}
\label{sec:results}
Denote the number of households in cluster $g\in[3]$ as $n_g$ and the capacity of each household in cluster $g\in[3]$ as $\bar{D}_g$. As the pricing scheme is imposed at the cluster level, the fairness criterion is also defined at the cluster level. Accordingly, we use the mean capacity of each cluster to represent the capacity of individual consumers within that cluster. This construction induces homogeneity within each cluster, i.e., consumers in the same cluster are exposed to the same price signal and exhibit the same provided energy.

Suppose the aggregator collects energy $D_g$ from each individual household in cluster $g$. Then the total aggregated energy is $\sum_{g\in[3]} n_g D_g$, with linear aggregator profit function $V\left(\sum_{g\in[3]} n_g D_g\right)=\pi\sum_{g\in[3]} n_g D_g$.

Ignoring fairness considerations, the profit-only optimization problem can be written as
\begin{equation}
\begin{aligned}
\max_{\{D_g\}_{g\in[3]}} \quad
& \sum_{g\in[3]} n_g \left(\pi D_g - \bigl(a D_g + b - a\bar{D}_g\bigr) D_g \right) \\
\text{s.t.}\quad
& \sum_{g\in[3]} n_g D_g \le D_s, \\
& 0 \le D_g \le \bar{D}_g, \quad \forall g\in[3].
\end{aligned}
\label{eq:case_study_opt}
\end{equation}
Similar to the stylized model, the optimal price offered to households in cluster $g$ is $p_g^\ast=aD_g^\ast+b-a\bar{D}_g$ for all $g\in[3]$, where $D_g^\ast$ denotes the optimal solution to~\eqref{eq:case_study_opt}. Here, $U_g$ denotes the individual utility of group $g$. The total consumer utility is $U=\sum_{g\in[3]} n_g U_g$, and the CNW is $W_{\text{CNW}}=\sum_{g\in[3]} n_g \log(U_g)$.

When incorporating a fairness criterion, since households within the same cluster are homogeneous, comparing the relevant outcome across households is equivalent to considering comparisons across clusters, which yields the constraint
$$|M_g(p_g)- M_{g'}(p_{g'})\max_{s,s'}|\le (1-\alpha)|M_s(p_s^\ast)-M_{s'}(p_{s'}^\ast)|~\text{for all}~g,g\in [3]~\text{with}~g\ne g',$$
where $M(\cdot)$ denotes a fairness measure, such as energy, price, and utility.

Because wholesale market prices fluctuate substantially and a VPP collects energy only when market conditions are profitable, we set $\pi = 5$. Under our setting, the condition $\pi > b- a \max_i \bar D_i$ guarantees that the aggregator earns a positive profit. This setting regarding $\pi$ also implies a threshold for the VPP to participate in the wholesale market in practice. For the maximum aggregated energy $D_\mathrm{s}$, we consider several values. Larger values of $D_{\mathrm{s}}$ correspond to peak periods with extreme supply shortages, while smaller values represent mild situations. In this section, we present results for $D_{\mathrm{s}}=0.8\sum_{g\in[3]} n_g \bar{D}_g$, which corresponds to extreme supply shortages of peak-period grid operation. Results for a mild shortages situations, $D_{\mathrm{s}}=0.3\sum_{g\in[3]} n_g \bar{D}_g$, are reported in Appendix~\ref{appendix:additional_case_study}.

Figure~\ref{fig:demand_case_study} presents the outcomes under energy fairness. The observed patterns are consistent with our theoretical analysis when the three clusters are aggregated into two representative groups. In particular, more flexible consumers (cluster~$3$) increase their provided energy, while less flexible consumers (clusters~$1$ and~$2$) reduce provided energy to offset this increase. Although either cluster~$2$ or cluster~$3$ may increase or decrease provided energy as $\alpha$ varies, their aggregated energy provision remains positive. The regime pattern is analogous to Regime~$1$ and Regime~$2$ as described in Theorem~\ref{thm:energy_fairness}. This confirms that the regime characterization derived in the two-cluster setting extends to the multi-cluster case with realistic response behavior. The resulting reallocation leads to a reduction in CNW, reflecting increased dispersion in individual utilities. The utility gains of cluster~$3$ exceed the combined profit loss of the aggregator and clusters~$1$ and~$2$, thereby yielding an increase in social welfare. However, as shown in Figure~\ref{fig:income_cluster}, cluster~$3$ corresponds to high-income consumers, while clusters~$1$ and~$2$ comprise predominantly low-income consumers. Energy fairness therefore disproportionately benefits high-income consumers while burdening low-income consumers, which runs counter to the common policy objective of prioritizing benefits for low-income consumers. Numerically, comparing the $\alpha=0$ and $\alpha=1$ cases, a $0.69\%$ loss in profit is associated with an $11.91\%$ gain in total utility and a $0.50\%$ increase in social welfare.

\begin{figure}[htbp]
\FIGURE{
    \includegraphics[width=\textwidth]{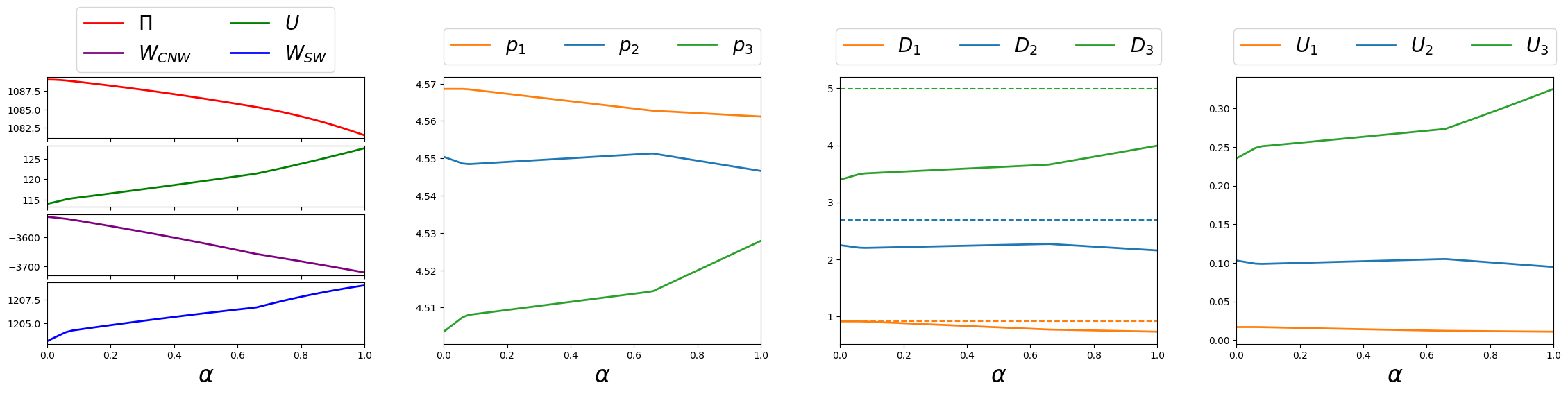}
}
{Energy fairness \label{fig:demand_case_study}
\vspace{.3cm}}
{}
\end{figure}

Figure~\ref{fig:price_case_study} illustrates the outcomes under price fairness. The observed patterns mirror the theoretical mechanisms identified earlier--more flexible consumers (cluster~$3$) experience no utility losses, while less flexible consumers (cluster~$1$) do not benefit in utility terms, which initially corresponds to Regime~$1$ in Theorem~\ref{thm:price_fairness}, except for CNW. This is likely affected by the number of consumers within each cluster. As $\alpha$ increases, less flexible consumers lose utility while more flexible consumers gain utility, indicating the transition from Regime~$1$ to Regime~$2$, as characterized in Theorem~\ref{thm:price_fairness}. The utility of cluster~$2$ is non-monotonic, with modest changes relative to the pronounced effects observed for clusters~$1$ and~$3$. 
Accepting a $2.32\%$ profit loss at $\alpha = 1$ relative to $\alpha = 0$ yields a $29.42\%$ increase in total consumer utility and increases overall social welfare by $0.68\%$. Nevertheless, these gains are unevenly distributed, as CNW decreases. Specifically, they arise alongside concentrated utility losses among less flexible consumers (cluster~$1$), while the benefits accrue primarily to more flexible consumers (high-income according to Figure~\ref{fig:income_cluster}). Similar to energy fairness, this distributional outcome challenges the conventional fairness rationale in energy policy, where fairness interventions are typically motivated by concerns for low-income or more vulnerable households. Moreover, the intermediate groups (cluster~$2$) are neither clear beneficiaries nor clear losers under price fairness policies, and the impact of such policies on these groups is highly context-dependent.

\begin{figure}[htbp]
\FIGURE{
    \includegraphics[width=\textwidth]{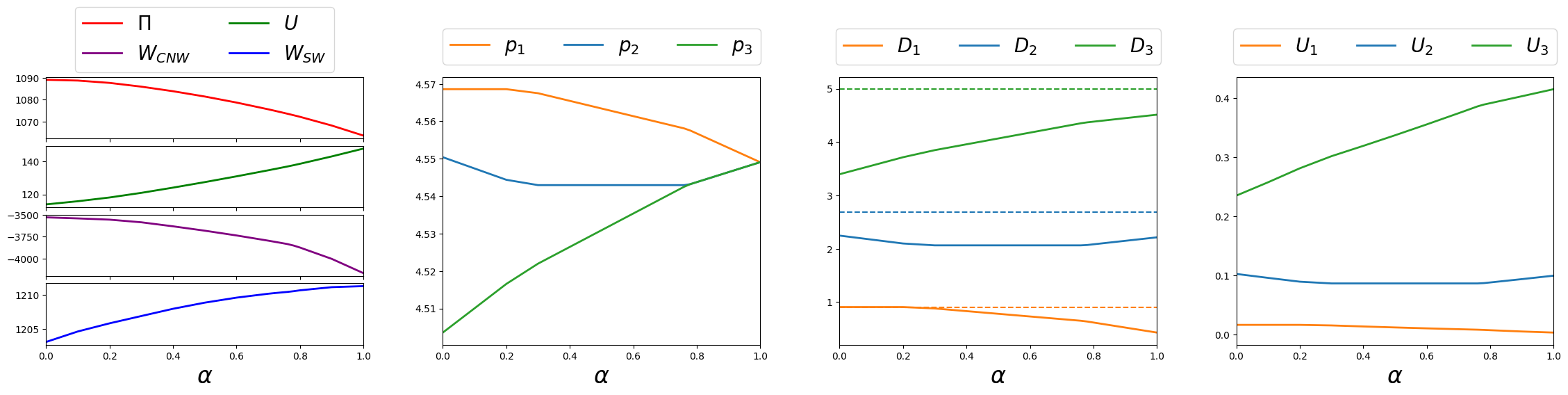}
}
{Price Fairness \label{fig:price_case_study}
\vspace{.3cm}}
{}
\end{figure}

Lastly, Figure~\ref{fig:utility_case_study} reports the outcomes under utility fairness. When the fairness constraint is first introduced, more flexible consumers (cluster~$3$) experience a substantial reduction in utility, while the less flexible consumers (cluster~$1$) remain unaffected, which mirrors Regime~$3$ in Theorem~\ref{thm:utility_fairness}, except for CNW.
With a higher $\alpha$, utility is progressively reallocated toward less flexible consumers (cluster~$1$), leading to great increases in CNW. Notably, social welfare remains approximately constant, indicating that gains in consumer utility are largely offset by losses in the aggregator's profit. In contrast to energy and price fairness, the redistributive loss under utility fairness falls primarily on more flexible consumers, while less flexible consumers emerge as the main beneficiaries, suggesting that the fairness gains accrue to low-income groups. However, this redistribution comes at the cost of the largest aggregator's profit loss among the three fairness criteria considered. As in the previous cases, the intermediate group (cluster~$2$) again plays a muted role, the magnitude of the change remains small relative to the pronounced effects observed for clusters~$1$ and~$3$, similar to the pattern observed under price fairness. When moving from the no-fairness case ($\alpha = 0$) to perfect fairness ($\alpha = 1$), profits decrease by $6.45\%$, accompanied by a $50.47\%$ increase in total utility, whereas social welfare falls by $1.06\%$. Despite this decline in social welfare, CNW increases.
\begin{figure}[htbp]
\FIGURE{
    \includegraphics[width=\textwidth]{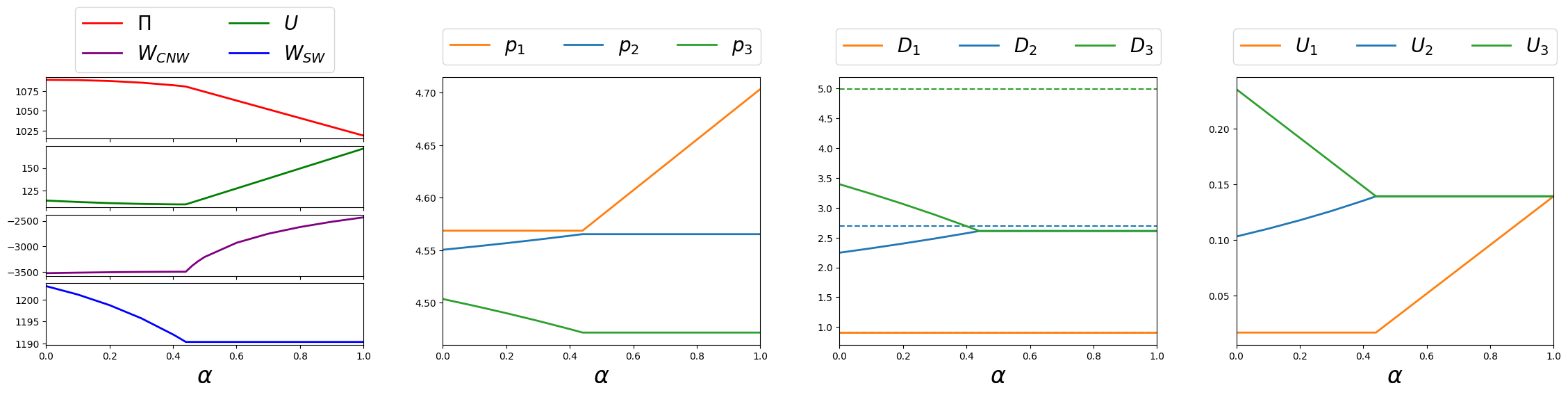}
}
{Utility Fairness \label{fig:utility_case_study}
\vspace{.3cm}}
{}
\end{figure}

\section{Conclusion}
In this paper, we analyze how fairness criteria and fairness levels in VPP incentive prices affect consumers and aggregators. We consider three fairness criteria along the dimensions of energy, price, and utility. We first show, under a stylized model with two heterogeneous consumers and a linear response model, that profit-only VPP pricing favors more flexible consumers, motivating the need for fairness considerations. We then demonstrate that fairness criteria must be considered separately, as they cannot be satisfied simultaneously in practice.

We characterize the entire spectrum of embedding fairness across all levels of $\alpha$, identifying all possible regimes and describing how these regimes evolve with $\alpha$ across performance measures, including CNW, total consumer utility, and social welfare for each fairness criterion. 
In particular, energy fairness exhibits four distinct operating regimes, some improving and others degrading performance measures, with favorable regimes attained at moderate
fairness levels and thus requiring careful choice of the fairness level. Price fairness, while potentially improving social welfare and total consumer utility, never benefits less flexible consumers and can exacerbate participation inequities relative to profit-only pricing. Utility fairness similarly exhibits multiple regimes with mixed system-level effects and consistently protects less flexible consumers without advantaging more flexible consumers. 
We then conduct a case study using real-world data from an experiment in Norway, adopting a tiered pricing scheme to examine the practical implications of implementing fairness considerations. Collectively, our results provide regulators and VPP operators with a principled map for selecting appropriate fairness criteria and fairness levels.

Lastly, our findings point to several promising directions for future research. One extension is to generalize consumers' cost functions beyond quadratic forms to examine how richer behavioral models affect fairness outcomes 
In addition, we can extend our framework to multi-period environments, where response behavior changes over time. Such extensions would enable the study of path-dependent and intertemporal fairness, and their implications for long-run aggregator profits.
More broadly, as VPPs scale and pricing algorithms become increasingly sophisticated--making feature-based and personalized pricing feasible--fundamental questions arise about which consumer features can be legitimately used for pricing, given legal constraints and concerns about participation and fairness.

\bibliographystyle{ormsv080}
\bibliography{references}

\newpage

\begin{APPENDIX}{}
\section{Proofs}
\begin{lemma}[Energy Equivalent Model]\label{lemma:equivalent_model}
    With $N = 2$ and the linear response function defined in \eqref{linear_response}, the profit-only optimization and energy fairness-constrained optimization, formulated with price decision variables $p_i$ can be equivalently reformulated in terms of the decision variables $D_i$ as follows:
    \begin{equation}
        \begin{aligned}
        \max_{D_1, D_2}&~ \pi (D_1+D_2) - (a D_1 + b-a\bar{D}_1) D_1 - (a D_2 + b-a\bar{D}_2) D_2\\
        \text{s.t. }&~ D_1 + D_2 \le D_\mathrm{s},\\
        &~ 0 \le D_i \le \bar{D}_i,\quad \forall i\in[2],\\
        &~\left| \frac{D_1}{\bar{D}_1} - \frac{D_2}{\bar{D}_2} \right|
        \le
        (1 - \alpha)
        \left| \frac{D_1^\ast}{\bar{D}_1} - \frac{D_2^\ast}{\bar{D}_2} \right|,
        \end{aligned}
    \label{eq:unconstrained_wrt_D}
    \end{equation}
    where $\alpha=0$ represents profit-only optimization.
\end{lemma}
\begin{proof}{Proof.}
    This equivalence follows directly from the linear structure of the response function in~\eqref{linear_response}.
    We first consider the profit-only problem. Without loss of generality, let $\hat p_1$ be any feasible price satisfying $\hat p_1 \le b - a \bar D_1$, under the linear response, the corresponding provided energy is zero. Consider the alternative feasible price $\tilde p_1 = b - a \bar D_1$, under which the induced energy for consumer~$1$ is zero. Then both prices yield the same profit level. Indeed,
    $$\Pi(\tilde p_1, p_2) - \Pi(\hat p_1, p_2)= \bigl[\pi\left(0+d_2(p_2)\right) - \tilde p_1 \cdot 0 - p_2 d_2(p_2)\bigr]
     - \bigl[\pi\left(0 + d_2(p_2)\right) - \hat p_1 \cdot 0 - p_2 d_2(p_2)\bigr]= 0.$$
    Hence, any feasible price $\tilde p_1 \le b - a \bar D_1$ is profit equivalent to $\hat p_1$.

    Similarly, for any feasible $\hat p_1 \geq b$, which induces provided energy $\bar D_1$, consider the alternative feasible price $\tilde p_1 = b$, under which the induced energy for consumer~$1$ equals $\bar D_1$, yields a greater or equal profit level, i.e., 
    $$\Pi(\tilde p_1, p_2)- \Pi(\hat p_1, p_2)= \bar{D}_1(\hat p_1-b) \geq 0.$$
    Therefore, it suffices to consider prices in the range $b - a \bar{D}_i  \leq p_i \leq  b$, over which the mapping between price and provided energy is injective through the linear relation $p_i = a D_i + b-a\bar{D}_i$. Substituting this expression into the original problem \eqref{aggregator_problem} yields the equivalent form in \eqref{eq:unconstrained_wrt_D}.

    Then, adding a fairness constraint whose metric is solely determined by the provided energy, not by price, does not affect the equivalence. Without loss of generality, consider a price $p_1 < b-a\bar{D}_1$ or $p_1 > b$, the value of provided energy remains unchanged as $0$ and $\bar{D}$, respectively, resulting in the same value for provided energy.\hfill\Halmos
\end{proof}

\medskip
\begin{proof}{Proof of Lemma~\ref{lemma:opt_sol}.}
The objective is a strictly concave quadratic in $(D_1,D_2)$ because its Hessian is $-2aI_2\prec0$ when $a>0$. The feasible set is a nonempty compact polytope. Hence, a unique maximizer exists, and the Karush--Kuhn--Tucker (KKT) conditions are necessary and sufficient.

Introduce multipliers $\lambda \ge 0$ for the aggregated energy constraint $D_1 + D_2 \le D_{\mathrm s}$,
$\mu_i \ge 0$ for $-D_i \le 0$, and $\nu_i \ge 0$ for $D_i - \bar D_i \le 0$, for $i\in[2]$. The Lagrangian is
\begin{equation*}
\mathcal L(D,\lambda,\mu,\nu)=
\pi(D_1{+}D_2)-\sum_{i=1}^2\Big(aD_i^2+(b-a\bar D_i)D_i\Big)
+\lambda(D_{\mathrm s}-D_1-D_2)+\sum_{i=1}^2\mu_i D_i+\sum_{i=1}^2\nu_i(\bar D_i-D_i).
\end{equation*}
Stationarity gives, for $i\in[2]$,
\begin{equation}
\label{eq:stationarity}
2aD_i = \pi-b+a\bar D_i - \lambda - \mu_i + \nu_i.
\end{equation}
The KKT conditions are given by
\begin{equation*}
\begin{aligned}
&\text{(i) Primal feasibility:} && D_i \in [0, \bar D_i], \quad D_1 + D_2 \le D_{\mathrm s},\\
&\text{(ii) Dual feasibility:}   && \lambda,\, \mu_i,\, \nu_i \ge 0, \quad i \in [2],\\
&\text{(iii) Complementary slackness:} && 
\lambda(D_{\mathrm s}-D_1-D_2) = 0,\quad
\mu_i D_i = 0,\quad
\nu_i(\bar D_i - D_i) = 0,\quad i \in [2].
\end{aligned}
\end{equation*}
We now consider two cases depending on whether the constraint, $D_1 + D_2 \le D_{\mathrm s}$, binds. 
When the constraint is not binding, complementary slackness implies $\lambda = 0$. Otherwise, $\lambda > 0$ and the constraint holds with equality so that $D_1 + D_2 = D_{\mathrm s}$.

\textbf{\boldmath Case $\lambda = 0$.}
With $\lambda = 0$, the stationarity condition~\eqref{eq:stationarity} reduces to
\begin{equation*}
2aD_i = \pi - b + a\bar D_i - \mu_i + \nu_i, \qquad i \in [2].
\end{equation*}
If neither the lower nor upper bound binds, then $\mu_i = \nu_i = 0$, and the first-order
condition yields the unconstrained maximizer
\begin{equation*}
D_i^{(0)} = \frac{\pi - b}{2a} + \frac{\bar D_i}{2}.
\end{equation*}

We now analyze the boundary cases in which the optimal solution $(D_1^\ast,D_2^\ast)$ lies on the boundary of the feasible set. The four possibilities are: (i) $D_1^\ast=0$, (ii) $D_2^\ast=0$, (iii) $D_1^\ast=\bar{D}_1$, and (iv) $D_2^\ast=\bar{D}_2$.

We first rule out cases (i) and (ii). For the sake of contradiction, suppose that $D_1^\ast = 0$ in the case $\lambda = 0$. Then the lower-bound constraint $D_1 \ge 0$ is active, so the associated KKT multiplier satisfies $\mu_1 \ge 0$ and complementary slackness implies $\mu_1 D_1^\ast = 0$. 
Moreover, since $D_1^\ast = 0 < \bar D_1$, the upper-bound constraint for $D_1$ is slack and thus $\nu_1 = 0$.

The stationarity condition \eqref{eq:stationarity} with $\lambda = 0$ gives
\begin{equation*}
2a D_1 = \pi - b + a \bar D_1 - \mu_1 + \nu_1,
\end{equation*}
and evaluating at $D_1 = 0$ and $\nu_1 = 0$ yields
\begin{equation*}
\mu_1 = \pi - b + a \bar D_1.
\end{equation*}
Using the explicit form of the profit function, we also compute
\begin{equation*}
\left.\frac{\partial \Pi(D_1,D_2^\ast)}{\partial D_1}\right|_{D_1 = 0}
= \pi - b + a \bar D_1
= \mu_1 > 0.
\end{equation*}

However, the KKT conditions for a maximization problem require that, at an optimal point, the directional derivative of the Lagrangian in any feasible direction be non-positive. 
Here, increasing $D_1$ from zero is a feasible direction, and the above calculation shows that the directional derivative in this direction is strictly positive. This contradicts the optimality of $D_1^\ast = 0$. Hence $D_1^\ast > 0$. 
By symmetry, the same argument rules out $D_2^\ast = 0$, so cases (i) and (ii) cannot occur.

We can also rule out case (iv), i.e., $D_2^\ast = \bar D_2$. For the sake of contradiction, suppose that $D_2^\ast = \bar D_2$. Then the upper-bound constraint is binding, and complementary slackness implies $\nu_2 \ge 0$. From the stationarity condition under $\lambda = 0$,
\begin{equation*}
2a D_2^\ast = \pi - b + a \bar D_2 - \mu_2 + \nu_2,
\end{equation*}
and using $D_2^\ast = \bar D_2$ and $\mu_2 = 0$ gives
\begin{equation*}
\nu_2 = (\pi - b) - 2a \bar D_2.
\end{equation*}
Thus $\nu_2 \ge 0$ requires $\frac{\pi - b}{2a} \ge \frac{\bar D_2}{2}$.

Under this condition, the unconstrained maximizer for $D_1$ also exceeds $\tfrac{\bar D_1}{2}$, so its projection saturates the upper bound and yields $D_1^\ast = \bar D_1$. Therefore,
\begin{equation*}
D_1^\ast + D_2^\ast
= \bar D_1 + \bar D_2 > D_{\mathrm s},
\end{equation*}
which violates the constraint $D_1 + D_2 \le D_{\mathrm s}$.  
Hence, case (iv) is impossible, and we must have $D_2^\ast < \bar D_2$.

Note that (iii) $D_1^\ast = \bar D_1$ cannot be ruled out, since $D_1^\ast > \bar D_1$ violates neither the KKT conditions nor the feasibility constraints. In particular, when the unconstrained maximizer $D_1^{(0)}$ exceeds $\bar D_1$, the upper bound naturally becomes binding. Therefore, define
\begin{equation*}
    D_1^{\dagger}= \min\left( \frac{\pi - b}{2a} + \frac{\bar{D}_1}{2},\;\bar{D}_1 \right)
    \quad\text{and}\quad
    D_2^{\dagger} = \frac{\pi - b}{2a} + \frac{\bar{D}_2}{2}.
\end{equation*}
If $D_1^{\dagger} + D_2^{\dagger} \le D_{\mathrm s}$, then this candidate is feasible. Hence $ D_i^\ast = D_i^\dagger$ for all $i \in [2]$.

\smallskip
\textbf{\boldmath Case $\lambda > 0$.}
When $\lambda > 0$, complementary slackness implies that $D_1^\ast + D_2^\ast = D_{\mathrm s}$. 
The stationarity condition~\eqref{eq:stationarity} then becomes
\begin{equation*}
2aD_i = \pi - b + a\bar D_i - \lambda - \mu_i + \nu_i, \qquad i \in [2].
\end{equation*}
For each $i$, the multipliers $\mu_i$ and $\nu_i$ correspond respectively to the lower 
and upper bound constraints $D_i \ge 0$ and $D_i \le \bar D_i$, and complementary slackness implies
\begin{equation*}
\mu_i D_i = 0 \quad \text{and} \quad \nu_i(\bar D_i - D_i) = 0.
\end{equation*}

If neither individual cap binds (i.e.,\ $0<D_i^\ast<\bar D_i$ for $i\in[2]$), then $\mu_i=\nu_i=0$ and \eqref{eq:stationarity} gives $2aD_i^\ast=\pi-b+a\bar D_i-\lambda$.
Summing over $i$ and imposing $D_1^\ast+D_2^\ast=D_{\mathrm s}$ yields
$\lambda=\pi-b+\tfrac{a}{2}(\bar D_1+\bar D_2)-aD_{\mathrm s}$ and hence
\begin{equation*}
D_1^{(0)}=\frac{D_{\mathrm s}}{2}+\frac{\bar D_1-\bar D_2}{4}\quad
\text{and}\quad
D_2^{(0)}=\frac{D_{\mathrm s}}{2}+\frac{\bar D_2-\bar D_1}{4}.
\end{equation*}

If $D_i^{(0)}$ lies below $0$, the lower bound $D_i = 0$ binds and $\mu_i > 0$; if it exceeds $\bar D_i$, the upper bound $D_i = \bar D_i$ binds and $\nu_i > 0$. Hence, in all cases, the optimal $D_i^\ast$ is obtained by projecting the unconstrained value $D_i^{(0)}$ onto the feasible interval $[0,\bar D_i]$, which gives
\begin{equation}
D_i^\ast
=\operatorname{proj}_{[0,\bar D_i]}\!\left(D_i^{(0)}
\right),
\qquad i \in [2],
\label{eq:clip}
\end{equation}
that satisfies $D_1^\ast + D_2^\ast = D_{\mathrm s}$.

We now analyze the boundary cases in which the optimal solution $(D_1^\ast,D_2^\ast)$ lies on the boundary of the feasible set. The four possibilities are: (i) $(0,D_{\mathrm{s}})$, (ii) $(\bar{D}_1,D_\mathrm{s}-\bar{D}_1)$, (iii) $(D_{\mathrm{s}},0)$, and (iv) $(D_{\mathrm{s}}-\bar{D}_2,\bar{D}_2)$.


We can exclude the case (iii) $D_1^\ast = D_{\mathrm{s}}$, because this would imply
$$D_1^{(0)}=\frac{D_{\mathrm s}}{2} + \frac{\bar D_1 - \bar D_2}{4} > D_\mathrm{s}
\;\Rightarrow\;
\bar D_1 - \bar D_2 > 2D_\mathrm{s} > 0,$$
which contradicts the assumption that $\bar D_2 > \bar D_1$.

Similarly, we can exclude the case (iv) $D_2^\ast = \bar{D}_2$. To achieve this condition,
$$
D_2^{(0)}=\frac{D_{\mathrm s}}{2} + \frac{\bar D_2 - \bar D_1}{4} \ge \bar{D}_2
\;\Longleftrightarrow\;
D_{\mathrm s} \ge \frac{\bar{D}_1 + 3\bar{D}_2}{2},
$$
but the right-hand side, $\tfrac{\bar{D}_1 + 3\bar{D}_2}{2}$, is larger than $\bar{D}_1 + \bar{D}_2$ since $\bar{D}_2 > \bar{D}_1$. 
This contradicts the assumption that $D_{\mathrm{s}} < \bar{D}_1 + \bar{D}_2$.

Therefore, all cases can be summarized as
\begin{equation}
D_1^\ast
=\operatorname{proj}_{[0,\bar D_1]}
\left(D_1^{(0)}\right)\quad\text{and}\quad
D_2^\ast=D_{\mathrm s}-D_1^\ast. \label{eq:profit_only_case2}
\end{equation}
Note that there are three main cases: each optimal decision either lies in the interior region, or $D_1^\ast$ is located on the boundary. This completes the characterization of the optimal solution. \hfill\Halmos
\end{proof}
\medskip

\begin{proof}{Proof of Theorem~\ref{prop:profit_only}.}
To prove this Theorem, we first establish the following lemma, which provides an equivalent representation of the consumer Nash welfare (CNW) under the energy equivalent model.
\begin{lemma}[\emph{Energy-CNW}]\label{Energy_CNW_proposition}
    Under the energy equivalent model defined in Lemma~\ref{lemma:equivalent_model}, the effect of energy changes on the energy-based consumer Nash welfare (DCNW), defined as $W_\mathrm{D} = \sum_{i\in [N]}\log (D_i)$, is equivalent to their effect on the CNW. Specifically,
    $$
    \frac{\partial W_\mathrm{D}}{\partial D_i} = \frac{\partial W_\mathrm{CNW}}{\partial D_i}, \forall i\in [N], \label{DCNW}
    $$
\end{lemma}
\begin{proof}{Proof of Lemma \ref{Energy_CNW_proposition}.}
    Substituting the price-energy relation $p_i = a D_i + b-a\bar{D}_i$ into the utility function yields $U_i = \frac{1}{2}aD_i^2$. By the definition of CNW,
        \begin{equation*}
            \begin{aligned}
                 W_\mathrm{CNW} &= \sum_{i\in [N]}\log\left(\frac{1}{2}aD_i^2\right) 
            = \sum_{i\in [N]}\log\left( \frac{1}{2}a\right) + 2\sum_{i\in [N]}\log(D_i) = \sum_{i\in [N]}\log\left( \frac{1}{2}a\right) + 2W_\mathrm{D}.
            \end{aligned}
        \end{equation*}
           
    Thus, CNW and DCNW differ only by an additive constant  $\sum_{i\in [N]}\log\left(\frac{1}{2}a\right)$ independent of the decision variables $D_i$. Thus, the provided energy change has the same impact on these two definitions.  \hfill\Halmos
\end{proof}

This Lemma establishes that DCNW can be used interchangeably with the CNW metric when analyzing the impact of energy changes, both in the profit-only problem and in the energy fairness-constrained problem. Because DCNW admits a simpler functional form in terms of $D_i$, it is more convenient for optimization and sensitivity analysis.

Lemma~\ref{lemma:opt_sol} characterizes the optimal solutions of the profit-only problem and distinguishes two cases. In the first case, where the aggregated energy constraint $D_1 + D_2 \le D_{\mathrm s}$ is nonbinding, $D_i^\ast = D_i^{\dagger}$. As $\bar{D}_1 < \bar{D}_2$, the optimal solution always satisfies $D_1^\ast < D_2^\ast$. Let $c:=\frac{\pi-b}{2a}$, we consider two scenarios:\\
(i) When $c \geq 0$:
\begin{itemize}
    \item If $\bar{D}_2 \leq 2c$, then $c+\frac{\bar{D}_i}{2} \geq \bar{D}_i$ and $D_i^\dagger = \bar{D}_i$, thus $D_2^\dagger-D_1^\dagger = \bar{D}_2 -\bar{D}_1>0$.
    \item If $\bar{D}_1 \leq 2c \leq \bar{D}_2$, then $D_1^\dagger = \bar{D}_1$, while $D_2^\dagger = c+ \frac{\bar{D}_2}{2}> 2c \geq \bar{D}_1$, thus $D_2^\dagger > D_1^\dagger$.
    \item If $\bar{D}_1 > 2c$, then both solutions are interior, thus $D_2^\dagger-D_1^\dagger = \frac{\bar{D}_2 -\bar{D}_1}{2}>0$.
\end{itemize}
(ii) When $c < 0$:
\begin{itemize}
    \item If $\bar{D}_2 \leq -2c$, then $c+\frac{\bar{D}_i}{2} \leq 0$ for $i\in [2]$, thus $D_1^\dagger = D_2^\dagger =0$; however, this contradicts the assumption $\pi > b -a\bar{D}_1$.
    \item If $\bar{D}_1 \leq -2c \leq \bar{D}_2$, then $D_1^\dagger = 0$ and $D_2^\dagger = c+ \frac{\bar{D}_2}{2} > 0$, thus $D_2^\dagger > D_1^\dagger$.
    \item If $\bar{D}_1 \geq -2c$, both solutions are interior, which again implies $D_2^\dagger > D_1^\dagger$.
\end{itemize}

Otherwise, when the aggregated energy constraint is binding, the optimal solution is given by~\eqref{eq:profit_only_case2}, which can be expressed as
$$
D_1^\ast = \operatorname{proj}_{\,[0,\bar{D}_1]}\left(D_1^{(0)}\right) = \min\left(\bar{D}_1,\max\left(0,D_1^{(0)}\right)\right),
$$
where $D_1^{(0)}=\tfrac{D_s}{2} + \tfrac{\bar D_1 -\bar D_2}{4}$. As $D_2^\ast = D_\mathrm{s} - D_1^\ast$, proving $D_2^\ast > D_1^\ast$ is equivalent to prove $D_1^\ast < \frac{D_\mathrm{s}}{2}$. Since $\bar{D}_1 < \bar{D}_2$ by the assumption, we have $D_1^{(0)} < \tfrac{D_\mathrm{s}}{2}$. Moreover, $0 < \tfrac{D_\mathrm{s}}{2}$.
Thus, $\max\left(0,D_1^{(0)}\right) < \tfrac{D_\mathrm{s}}{2}$, and since $\min\left(\bar{D}_1, \frac{D_\mathrm{s}}{2}\right) \leq \tfrac{D_\mathrm{s}}{2}$, we have $D_1^\ast \leq \max\left(0,D_1^{(0)}\right)< \tfrac{D_\mathrm{s}}{2}$. 

We next analyze the DCNW, total consumer utility, and social welfare under the binding aggregated energy constraint. According to the function types, $W_\mathrm{D} =\sum_{i\in [2]} \log(D_i)$, $U_i = \frac{1}{2}aD_i^2$, and $W_\mathrm{SW} = \pi D_\mathrm{s} - C_1(D_1)- C_2(D_2)$. Let $D_1^\ast = D_\mathrm{s}-D_2^\ast$, the sensitivity of DCNW regarding $D_1^\ast$ is 
$$
\frac{\partial W_\mathrm{D}}{\partial D_1^\ast} = \frac{\partial \left(\log(D_1^\ast) + \log(D_\mathrm{s} - D_1^\ast) \right)}{\partial D_1^\ast} = \frac{D_\mathrm{s} - 2D_1^\ast}{D_1^\ast (D_\mathrm{s}-D_1^\ast)} >0,
$$
where $D_\mathrm{s} - 2D_1^\ast >0$ holds since $D_1^\ast < D_2^\ast$ and $D_1^\ast + D_2^\ast = D_\mathrm{s}$. The sensitivity of total consumer utility is 
$$\frac{\partial U}{\partial D_1^\ast} = \frac{\partial \left(\frac{1}{2}aD_1^{\ast2} + \frac{1}{2}a(D_\mathrm{s} - D_1^\ast)^2 \right)}{\partial D_1^\ast}=a(2D_1^\ast - D_\mathrm{s})<0,$$
and the sensitivity of social welfare is
$$
\frac{\partial W_\mathrm{SW}}{\partial D_1^\ast} = \frac{\partial \left(\pi D_\mathrm{s} -C_1(D_1^\ast) - C_2(D_s-D_1^\ast) \right)}{\partial D_1^\ast}
=a\left(D_\mathrm{s} - 2D_1^\ast + \bar D_1 - \bar D_2 \right),
$$
which could be positive or negative depending on the parameters. 

When the upper bound on $D_1^\ast$ is not binding, i.e., $0<D_1^\ast<\bar{D}_1$, the optimal solution is $D_1^\ast = D_1^{(0)}$. Therefore, by the chain rule, we have
\begin{equation*}
    \begin{aligned}
    \frac{\partial W_\mathrm{D}}{\partial (\bar{D}_2-\bar{D}_1)}=&\frac{\partial W_\mathrm{D}}{\partial D_1^\ast}\frac{\partial D_1^\ast}{\partial (\bar{D}_2-\bar{D}_1)}=-\frac{1}{4}\frac{\partial W_\mathrm{D}}{\partial D_1^\ast}<0,\\
    \frac{\partial U}{\partial (\bar{D}_2-\bar{D}_1)}=&\frac{\partial U}{\partial D_1^\ast}\frac{\partial D_1^\ast}{\partial (\bar{D}_2-\bar{D}_1)}=-\frac{1}{4}\frac{\partial U}{\partial D_1^\ast}>0,\\
    \frac{\partial W_\mathrm{SW}}{\partial (\bar{D}_2-\bar{D}_1)}=&\frac{\partial W_\mathrm{SW}}{\partial D_1^\ast}\frac{\partial D_1^\ast}{\partial (\bar{D}_2-\bar{D}_1)}+\frac{\partial W_\mathrm{SW}}{\partial \bar D_2} - \frac{\partial W_\mathrm{SW}}{\partial \bar D_1}
    =\frac{a}{2} (\bar D_2 - \bar D_1)>0.
    \end{aligned}
\end{equation*}
When the upper bound on the $D_1^\ast$ is binding, $D_1^\ast= \bar{D}_1$, further increases in $\bar{D}_2$ do not change optimal solutions, and both DCNW and total consumer utility remain unchanged, but social welfare increases. In contrast, increasing $\bar{D}_1$ reduces parameter difference $\bar{D}_2-\bar{D}_1$, corresponding to an increase in DCNW and a decrease in utility and social welfare. Note that even $\bar D_1 + \bar D_2$ decrease, it should satisfy $\bar D_1 + \bar D_2 > D_\mathrm{s}$. \hfill\Halmos


\end{proof}

\medskip
\begin{proof}{Proof of Theorem~\ref{thm:energy_fairness}.}
    The optimization problem under the energy fairness criterion can be represented based on Lemma~\ref{lemma:equivalent_model}. Let $D_i^\ast$ (resp. $D_i(\alpha)$) denote the no-fairness (resp. $\alpha$-energy fairness) optimal energy of consumer~$i$, and define the initial energy ratio gap $\Delta_D:=\left|\frac{D_1^\ast}{\bar{D}_1}-\frac{D_2^\ast}{\bar{D}_2}\right|$. 
    Because each $D_i$ is restricted to the interval $[0,\bar{D}_i]$, the analysis proceeds by distinguishing whether each optimal provided energy $D_1^\ast$ lies in the interior $(0,\bar D_i)$ or on the boundary, noting that Lemma~\ref{lemma:opt_sol} excludes the cases $D_2^\ast\in\{0, \bar{D}_2\}$.

    Because some of the regimes are connected sequentially (Regime~$1 \to 2$ or Regime~$1 \to 3$), we begin by analyzing the latter regime (Regime~$3$), which also brings the conditions for Regime~$4$ accordingly, and then proceed backward to Regimes~$1$ and $2$.
    
    \noindent\textbf{\boldmath Regimes~$3$ and $4$: $D_1^\ast < \bar D_1$}.\\
    Since $\pi > b -a\bar D_1$ by assumption, $D_1^\ast = c + \frac{\bar D_1}{2} > 0$, where $c =\frac{\pi - b}{2a}$. Suppose $0 < D_i^\ast <\bar{D}_i$ for any $i\in[2]$,
    let $\lambda \geq 0$ denote the Lagrangian multiplier on the aggregated energy constraint $D_1 + D_2 \leq D_\mathrm{s}$. As in Lemma~\ref{lemma:opt_sol}, we have two cases: slack ($\lambda=0$) and binding ($\lambda >0$).

    \textbf{\boldmath Case $\lambda = 0$.}
    Suppose $D_1^\ast + D_2^\ast < D_\mathrm{s}$. According to Lemma~\ref{lemma:opt_sol}, the unconstrained maximizer satisfies $D_i^\ast = c + \frac{\bar D_i}{2}$, and the initial energy ratio gap is 
    \begin{equation*}
    \Delta_D 
    := \left|\frac{{D}_1^\ast}{\bar{D}_1} - \frac{{D}_2^\ast}{\bar{D}_2}\right| 
    = \left|c\right|\left(\frac{1}{\bar{D}_1} - \frac{1}{\bar{D}_2}\right).
    \end{equation*}

    The energy fairness constraint is
    $$
    \left| \frac{D_1}{\bar D_1} - \frac{D_2}{\bar D_2} \right|
    \leq (1-\alpha) \Delta_D,$$
    can be written equivalently as 
    \begin{equation*}
    |D_1\bar D_2 - D_2 \bar D_1| - (1-\alpha)(\bar D_2 - \bar D_1)\frac{|\pi-b|}{2a}  \le 0.
    \end{equation*}
    Let $T_1(\alpha):= (1-\alpha)(\bar D_2 - \bar D_1)\frac{|\pi - b|}{2a}$. Because the box constraints ($D_i\in [0,\bar D_i]$) and aggregated energy constraint ($D_1 + D_2 \le D_\mathrm{s}$) are slack, the Lagrangian with multiplier $\eta \geq 0$ for the energy fairness constraint is 
    \begin{equation*}
    \mathcal{L}
    = \pi (D_1 + D_2)
    - \sum_{i=1}^2 \left( a D_i^2 + (b - a \bar D_i) D_i \right)
    + \eta \left( T_1(\alpha) - |D_1\bar D_2 - D_2 \bar D_1| \right). 
    \end{equation*}
    
    Then, the first-order stationarity conditions yield the following. When $D_1\bar D_2 - D_2 \bar D_1 > 0$,
    \begin{equation}\label{less0optimalD}
    \begin{aligned}
        0 = \pi - 2aD_1(\alpha) - (b-a\bar D_1) - \eta \bar D_2  
            \iff D_1 (\alpha)= c+ \frac{\bar D_1}{2} - \frac{\eta}{2a} \bar D_2,\\
        0 = \pi - 2aD_2(\alpha) - (b-a\bar D_2) + \eta \bar D_1
            \iff D_2 (\alpha)= c+ \frac{\bar D_2}{2} + \frac{\eta}{2a} \bar D_1,
    \end{aligned}
    \end{equation}
    otherwise, when $D_1\bar D_2 - D_2 \bar D_1 < 0$,
    \begin{equation}
        D_1 (\alpha) = c+ \frac{\bar D_1}{2} + \frac{\eta}{2a} \bar D_2
        \quad\text{and}\quad
         D_2 (\alpha) = c+ \frac{\bar D_2}{2} - \frac{\eta}{2a} \bar D_1. \label{greater0optimalD}
    \end{equation}

    Note that the energy fairness constraint is always binding. This is because the unconstrained optimum lies on the boundary
    $|\frac{D_1}{\bar D_1} - \frac{D_2}{\bar D_2}| \le (1-\alpha)\Delta_D$ at $\alpha=0$, since $|\frac{D_1^\ast}{\bar D_1} - \frac{D_2^\ast}{\bar D_2}| = \Delta_D$. As $\alpha$ increases, the feasible region shrinks, and the unconstrained optimum lies strictly outside it for all $\alpha>0$. Under strict concavity of the objective, any slack in the constraint would allow a profitable move toward the unconstrained optimum, contradicting optimality. Therefore, the constraint binds and can be written as
    $$
    |D_1(\alpha)\bar D_2 - D_2(\alpha)\bar D_1|
    = (1-\alpha)|D_1^\ast \bar D_2 - D_2^\ast \bar D_1|.
    $$
    Moreover, since the fairness constraint implies $|D_1(\alpha)\bar D_2 - D_2(\alpha)\bar D_1|>0$ for all $\alpha \in [0,1)$ whenever $|D_1^\ast \bar D_2 - D_2^\ast \bar D_1|>0$. Mathematically,
    $$
    \operatorname{sign}(D_1(\alpha)\bar D_2 - D_2(\alpha)\bar D_1)
    =
    \operatorname{sign}(D_1^\ast \bar D_2 - D_2^\ast \bar D_1),
    \quad \forall \alpha \in [0,1].
    $$
    Therefore, we consider two cases: (i) $D_1^\ast \bar D_2 - D_2^\ast \bar D_1 > 0$ and (ii) $D_1^\ast \bar D_2 - D_2^\ast \bar D_1 < 0$. When $D_1^\ast \bar D_2 - D_2^\ast \bar D_1 = 0$, the energy fairness constraint is already satisfied, and we omit this degenerate case.

    \textbf{Scenario (1).} $D_1^\ast \bar D_2 - D_2^\ast \bar D_1 > 0$ implies that $T_1(\alpha)=(1-\alpha)\left(\bar D_2 - \bar D_1\right)c$, i.e., $c>0$, and $D_1(\alpha)\bar D_2 - D_2(\alpha)\bar D_1>0$ for all $\alpha\in[0,1)$.
    The energy fairness constraint is given by
    \begin{equation*}
    \begin{aligned}
        D_1\bar D_2(\alpha) - D_2(\alpha) \bar D_1 = T_1 (\alpha)
        \iff& c\left(\bar D_2 - \bar D_1\right) - \frac{\eta(\alpha)}{2a} Q = (1-\alpha)\left(\bar D_2 - \bar D_1\right)c\\
        \iff& \frac{\eta(\alpha)}{2a} Q = \alpha \left(\bar D_2 - \bar D_1\right)c\\
        \iff& \eta (\alpha) = \frac{2a\alpha c \left(\bar D_2 - \bar D_1\right)}{Q},
    \end{aligned}
    \end{equation*}
    where $Q := \bar D_1^2 + \bar D_2^2$.
    Plug $\eta(\alpha)$ back into \eqref{less0optimalD}, then
    \begin{equation}
        D_1(\alpha) = c+\frac{\bar D_1}{2} - \alpha c \left(\bar D_2 - \bar D_1\right)\frac{\bar D_2}{Q}
        \quad\text{and}\quad
        D_2(\alpha) = c+ \frac{\bar D_2}{2} + \alpha c \left(\bar D_2 - \bar D_1\right)\frac{\bar D_1}{Q}.
        \label{demand_solution_alpha}
    \end{equation}

    Under this scenario, differentiating $D_1(\alpha)$ and $D_2(\alpha)$ with respect to $\alpha$ gives $D_1'(\alpha) = -c (\bar D_2 - \bar D_1)\frac{\bar D_2}{Q}$ and $D_2'(\alpha) = c (\bar D_2 - \bar D_1)\frac{\bar D_1}{Q}$. Thus, as $\alpha$ increases, $D_1(\alpha)$ decreases and $D_2(\alpha)$ increases. Since $D_1'(\alpha) + D_2'(\alpha) = -\frac{c(\bar D_2- \bar D_1)^2}{Q} < 0$ and $D_1^\ast + D_2^\ast < D_\mathrm{s}$, we must have $D_1(\alpha) + D_2(\alpha) < D_\mathrm{s}$ for all $\alpha$. Let $\tilde{\alpha}_1$ denote the first value of $\alpha$ at which either $D_1(\alpha)=0$ or $D_2(\alpha)=\bar D_2$. More precisely, $\tilde{\alpha}_1:=\min(\alpha_1,\alpha_2,1)$, where
    \begin{equation*}
    \begin{aligned}
    \alpha_1&:=\inf\{\alpha|D_1(\alpha)=0\} = \frac{\left(c+\frac{\bar D_1}{2} \right)Q}{c\left(\bar D_2 - \bar D_1\right)\bar D_2}, \\
    \alpha_2&:=\inf\{\alpha|D_2(\alpha)=\bar{D}_2\} = \frac{\left(\frac{\bar D_2}{2} -c\right)Q}{c\left(\bar D_2 - \bar D_1\right)\bar D_1}.
    \end{aligned}
    \end{equation*}
    The necessary and sufficient conditions for $\alpha_1<1$ and $\alpha_2<1$ are
    \begin{equation*}
    \begin{aligned}
    \alpha_1<1 
    \iff \left(c+\frac{\bar D_1}{2}\right)Q < c\left(\bar D_2-\bar D_1\right) \bar D_2
    &\iff c\bar D_1^2+c\bar D_2^2 + \frac{\bar D_1}{2}Q < c\bar D_2^2 -c\bar D_1 \bar D_2\iff c\bar D_1^2+ \frac{\bar D_1}{2}Q < -c\bar D_1 \bar D_2,\\[4pt]
    \alpha_2<1 
    \iff \left(\frac{\bar D_2}{2} -c\right)Q < c\left(\bar D_2 - \bar D_1\right)\bar D_1
    &\iff \frac{ \bar D_2}{2} \left(\bar D_1^2 + \bar D_2^2\right) < c\bar D_1 \bar D_2 + c\bar D_2^2\\
    &\iff \bar D_1 \bar D_2 \left(c - \frac{\bar D_1}{2}\right) + \bar D_2^2\left(c - \frac{\bar D_2}{2}\right) > 0.
    \end{aligned}
    \end{equation*}
    The first inequality, $\alpha_1<1$, is infeasible since $c\bar D_1^2 + \frac{\bar D_1}{2}Q > 0$. For $\alpha_2 <1$, the condition is also infeasible since 
    \begin{equation*}
    D_i^\ast = c + \frac{\bar D_i}{2} < \bar D_i 
    \iff c - \frac{\bar D_i}{2} < 0, \quad i\in [2].
    \end{equation*}
    Hence, $\tilde{\alpha}_1 = 1$.
        
    From Lemma~\ref{Energy_CNW_proposition}, energy-based consumer Nash welfare (DCNW) can be used interchangeably with consumer Nash welfare (CNW) under the energy fairness constraint. We thus use the change of DCNW $W_\mathrm{D} (\alpha)= \log(D_1(\alpha)) + \log (D_2(\alpha))$ to present CNW change,
    \begin{equation*}
    \begin{aligned}
         W_\mathrm{D}'(\alpha) = \frac{D_1'(\alpha)}{D_1(\alpha)} + \frac{D_2'(\alpha)}{D_2(\alpha)}
        &= \frac{c\left(\bar D_2 - \bar D_1\right)\frac{\bar D_1}{Q} D_1(\alpha) 
        - c\left(\bar D_2 - \bar D_1\right) \frac{\bar D_2}{Q} D_2(\alpha) }
        {D_1(\alpha)D_2(\alpha) } \\
        &= \frac{c\left(\bar D_2 - \bar D_1\right) }{Q D_1(\alpha)D_2(\alpha) } 
        \left(\bar D_1 D_1(\alpha) - \bar D_2 D_2(\alpha)\right) < 0.
    \end{aligned}    
    \end{equation*}
    The last inequality holds because $D_1^\ast < D_2^\ast$ from Proposition~\ref{prop:profit_only}, and $D_1(\alpha)$ decreases while $D_2(\alpha)$ increases in this scenario, implying that $D_1(\alpha) < D_2(\alpha)$ for any $\alpha$. Thus, CNW also decreases as $\alpha$ increases. 
    
    The utility of each consumer is given by $U_i(\alpha) = \frac{1}{2}aD_i(\alpha)^2$ for any $i\in [2]$. Differentiating with respect to $\alpha$ yields
    \begin{equation*}
    \begin{aligned}
         U_1'(\alpha) = aD_1(\alpha)D_1'(\alpha) = -aD_1(\alpha)c\left(\bar D_2 - \bar D_1\right)\frac{\bar D_2}{Q} < 0, \\
         U_2'(\alpha) = aD_2(\alpha)D_2'(\alpha) = aD_2(\alpha)c\left(\bar D_2 - \bar D_1\right)\frac{\bar D_1}{Q} > 0. 
    \end{aligned}
    \end{equation*}
    The total consumer utility satisfies
    $$
    U'(\alpha) = aD_2(\alpha)c(\bar D_2 - \bar D_1)\frac{\bar D_1}{Q}
    -aD_1(\alpha)c(\bar D_2 - \bar D_1)\frac{\bar D_2}{Q}  
    = \frac{ac(\bar D_2 - \bar D_1)}{Q} 
    \left( D_2 (\alpha)\bar D_1 - D_1(\alpha)\bar D_2 \right) < 0,$$
    where the inequality follows from the condition implied by \textbf{Scenario~(1)},
    $D_1(\alpha)\bar D_2-D_2(\alpha)\bar D_1 >0$.
    
    As the total consumer utility decreases as $\alpha$ increases, and the aggregator profit also decreases when the fairness constraint is imposed, social welfare must decrease as well.
    
    These characterize the Regime $3$ for any $\alpha \le \tilde{\alpha}_1$,
    \[
    \boxed{
    U_1~\text{decreases},\quad
    U_2~\text{increases},\quad
    U~\text{decreases},\quad
    W_\mathrm{SW}~\text{decreases},~\text{and}~
    W_\mathrm{CNW}~\text{decreases}.\quad(\text{Regime}~$3$)
    }
    \]
    Since $\tilde{\alpha}_1 = 1$, there is no transition from this regime to another.
    
    \textbf{Scenario (2).} $D_1^\ast \bar D_2 - D_2^\ast \bar D_1 < 0$ implies that $T_1(\alpha)=-(1-\alpha)\left(\bar D_2 - \bar D_1\right)c$, i.e., $c<0$, and $D_2(\alpha)\bar D_1-D_1(\alpha)\bar D_2 <0$ for all $\alpha\in[0,1)$.
    \begin{equation}
    \begin{aligned}
        D_2(\alpha) \bar D_1 - D_1(\alpha)\bar D_2 = -T_1 (\alpha)
        \iff& -c(\bar D_2 - \bar D_1) - \frac{\eta(\alpha)}{2a} Q = -(1-\alpha)\left(\bar D_2 - \bar D_1\right)c\\
        \iff& \frac{\eta(\alpha)}{2a} Q = - \alpha \left(\bar D_2 - \bar D_1\right)c\\
        \iff& \eta (\alpha) = \frac{2a\alpha c (\bar D_2 - \bar D_1)}{Q} >0.\label{eta_alpha}
    \end{aligned}
    \end{equation}
    Plug $\eta(\alpha)$ back into \eqref{greater0optimalD}, $D_1(\alpha)$ and $D_2(\alpha)$ are the same as~\eqref{demand_solution_alpha}, with the only difference being that $c<0$. Within this regime, differentiating with respect to $\alpha$ gives 
    $$D_1'(\alpha) = -c (\bar D_2 - \bar D_1)\frac{\bar D_2}{Q} > 0\quad\text{and}\quad D_2'(\alpha) = c (\bar D_2 - \bar D_1)\frac{\bar D_1}{Q} < 0.$$
    Thus, as $\alpha$ increases, $D_1(\alpha)$ increases and $D_2(\alpha)$ decreases.
    
    Note that $c<0$ implies $D_1'(\alpha) + D_2'(\alpha) > 0$. From Proposition~\ref{prop:profit_only}, $D_1^\ast < D_2^\ast$, $D_1(\alpha)$ could be limited by either $D_1(\alpha)=\bar D_1$, $\frac{D_1(\alpha)}{\bar D_1} = \frac{D_2(\alpha)}{\bar D_2}$, or $D_1(\alpha) + D_2(\alpha) = D_\mathrm{s}$. The second case, $\frac{D_1(\alpha)}{\bar D_1} = \frac{D_2(\alpha)}{\bar D_2}$, occurs only when $\alpha = 1$, since energy fairness constraint is binding with $\frac{D_1(\alpha)}{\bar D_1} - \frac{D_2(\alpha)}{\bar D_2}= (1-\alpha)\Delta_D$, which equals zero only when $\alpha = 1$. 
    Note that we do not need to consider the case $D_2(\alpha)=0$. The minimum feasible value of $D_2(\alpha)$ is attained at $\alpha = 1$, and even in this case we have $D_2(1) = D_1(1)\tfrac{\bar D_2}{\bar D_1} > 0$, since $D_1(\alpha)$ is increasing from $D_1^\ast\ge0$.

    The threshold at which the constraint switches from slack to binding is
    $\tilde{\alpha}_2:=\min(\alpha_3,\alpha_4,1)$, which is determined by 
    \begin{equation*}
    \begin{aligned}
        \alpha_3 &:=\inf\{\alpha|D_1(\alpha)=\bar D_1\} 
        = \frac{\left(\frac{\bar D_1}{2} -c\right)Q}{-c\left(\bar D_2 - \bar D_1\right)\bar D_2},\\
        \alpha_4 &: =\inf\{\alpha|D_1(\alpha) + D_2(\alpha) = D_\mathrm{s}\} 
        = 1 + \frac{Q \left(D_1^\ast +D_2^\ast - D_\mathrm{s}\right) - c\left(\bar D_2 -\bar D_1\right)^2}{c\left(\bar D_2 - \bar D_1\right)^2}.
    \end{aligned}
    \end{equation*}
    The necessary and sufficient conditions for $\alpha_3 <1$ and $\alpha_4 < 1$ are
    \begin{equation*}
    \begin{aligned}
        \alpha_3<1 
        &\iff \left(\frac{\bar D_1}{2} -c\right)Q < -c\left(\bar D_2 - \bar D_1\right)\bar D_2
        \iff \frac{ \bar D_1}{2} \left(\bar D_1^2 + \bar D_2^2\right) < c\bar D_1^2 + c\bar D_1 \bar D_2\\
        &\iff \bar D_1 \bar D_2 \left(c - \frac{\bar D_2}{2}\right) + \bar D_1^2\left(c - \frac{\bar D_1}{2}\right) > 0,\\
        \alpha_4 < 1 
        &\iff Q \left(D_1^\ast +D_2^\ast - D_\mathrm{s}\right) - c\left(\bar D_2 -\bar D_1\right)^2 > 0.
    \end{aligned}
    \end{equation*}
    The inequality $\alpha_3<1$ is infeasible as $c<0$, while the inequality $\alpha_4<1$ could be held as $D_1^\ast +D_2^\ast - D_\mathrm{s}< 0$ and $c< 0$. Hence, $\tilde{\alpha}_2 = \min (\alpha_4, 1)$. 
    
    Regarding the performance measure, the DCNW change satisfies 
    \begin{equation*}
    \begin{aligned}
         W_\mathrm{D}'(\alpha) 
        &= \frac{c\left(\bar D_2 - \bar D_1\right) }{Q D_1(\alpha)D_2(\alpha) } 
        \left(\bar D_1 D_1(\alpha) - \bar D_2 D_2(\alpha)\right) > 0.
    \end{aligned}    
    \end{equation*}
    
    The change in the utility of each consumer is given by 
    \begin{equation*}
    \begin{aligned}
         U_1'(\alpha) = aD_1(\alpha)D_1'(\alpha) = -aD_1(\alpha)c\left(\bar D_2 - \bar D_1\right)\frac{\bar D_2}{Q} > 0, \\
         U_2'(\alpha) = aD_2(\alpha)D_2'(\alpha) = aD_2(\alpha)c\left(\bar D_2 - \bar D_1\right)\frac{\bar D_1}{Q} < 0. 
    \end{aligned}
    \end{equation*}
    The change of total consumer utility satisfies
    $$
    U'(\alpha) 
    = \frac{ac(\bar D_2 - \bar D_1)}{Q} 
    \left( D_2 (\alpha)\bar D_1 - D_1(\alpha)\bar D_2 \right) < 0,$$
    where the last inequality is determined by     
    $$ D_2 (\alpha)\bar D_1 - D_1(\alpha)\bar D_2 
    = c (\bar D_1 - \bar D_2) + \alpha c (\bar D_2 -\bar D_1) 
    = c(1-\alpha) (\bar D_1 -\bar D_2) > 0.
    $$
    
    Because total consumer utility is decreasing in $\alpha$, and the aggregator profit is decreasing under the fairness constraint, social welfare is also decreasing.
    
    These characterize the Regime $4$ for any $\alpha < \tilde{\alpha}_2$,
    \[
    \boxed{
    U_1~\text{increases},\quad
    U_2~\text{decreases},\quad
    U~\text{decreases},\quad
    W_\mathrm{SW}~\text{decreases},~\text{and}~
    W_\mathrm{CNW}~\text{increases}.\quad(\text{Regime}~$4$)
    }
    \]
    
    For $\alpha > \tilde{\alpha}_2$, the $D_1(\alpha) + D_2(\alpha) = D_\mathrm{s}$ and the system remains in Regime $4$, but the dynamics of $D_1(\alpha)$ and $D_2(\alpha)$ change. This is determined by the analysis in the next \textbf{\boldmath Case $\lambda > 0$}.

    \noindent\textbf{\boldmath Case $\lambda > 0$.} 
    Suppose $D_1^\ast + D_2^\ast = D_\mathrm{s}$ and both $D_i^\ast$ are interior solutions. Then, according to Lemma~\ref{lemma:opt_sol}, the unconstrained optimum is
    \begin{equation*}
    D_1^\ast = \frac{D_{\mathrm s}}{2} + \frac{\bar D_1 - \bar D_2}{4}
    \quad \text{and} \quad
    D_2^\ast = \frac{D_{\mathrm s}}{2} + \frac{\bar D_2 - \bar D_1}{4} = D_\mathrm{s}-D_1^\ast.
    \end{equation*}
    The initial energy ratio gap is
    \begin{equation*}
    \Delta_D 
    := \left|\frac{{D}_1^\ast}{\bar{D}_1} - \frac{{D}_2^\ast}{\bar{D}_2}\right| 
    = \left|\frac{D_\mathrm{s} (\bar D_2 - \bar D_1)}{2\bar D_2 \bar D_1} 
    + \frac{\bar D_1^2 - \bar D_2^2}{4\bar D_2 \bar D_1}\right| 
    = \frac{(\bar D_2 - \bar D_1)}{4\bar D_2 \bar D_1}\left|\bar D_1 + \bar D_2 - 2D_\mathrm{s}\right|.
    \end{equation*}

    The energy fairness condition can be expressed as 
    \begin{equation*}
    |D_1\bar D_2 - D_2 \bar D_1| - \frac{(1-\alpha)(\bar D_2 - \bar D_1)}{4} |\bar D_1 + \bar D_2 - 2D_\mathrm{s}| \le 0.
    \end{equation*}
    As in \textbf{\boldmath Case $\lambda = 0$}, the energy fairness constraint is binding. Let $T_2(\alpha) := \frac{(1-\alpha) (\bar D_2 - \bar D_1)}{4}
    |\bar D_1 + \bar D_2 - 2D_\mathrm{s}| \ge 0$. $\left(D_1(\alpha),D_2(\alpha)\right)$ can be obtained with the following conditions,
    \begin{equation*}
        D_1 + D_2 =D_\mathrm{s}
        \quad \text{and}\quad 
        |D_1\bar D_2 - D_2 \bar D_1| = T_2(\alpha) .
    \end{equation*}
     When $D_1 \bar D_2 - D_2 \bar D_1 > 0$, the optimal solution is 
    \begin{equation}
        D_1(\alpha) = \frac{\bar D_1 D_\mathrm{s} + T_2(\alpha)}{\bar D_1 + \bar D_2}
        \quad\text{and}\quad
        D_2(\alpha) = \frac{\bar D_2 D_\mathrm{s} - T_2(\alpha)}{\bar D_1 + \bar D_2}. \label{optimal_Sgeq0}
    \end{equation}
    When $D_1 \bar D_2 - D_2 \bar D_1 < 0$, the optimal solution becomes 
    \begin{equation}
        D_1(\alpha) = \frac{\bar D_1 D_\mathrm{s} - T_2(\alpha)}{\bar D_1 + \bar D_2}
        \quad\text{and}\quad
        D_2(\alpha) = \frac{\bar D_2 D_\mathrm{s} + T_2(\alpha)}{\bar D_1 + \bar D_2}. \label{optimal_Sleq0}
    \end{equation}

    There are two scenarios determined by the absolute value of $\bar D_1 + \bar D_2 - 2D_\mathrm{s}$.
    
    \textbf{Scenario (1).} $\bar D_1 + \bar D_2 < 2D_\mathrm{s}$, which indicates $T_2(\alpha) = \frac{(1-\alpha)(\bar D_2 - \bar D_1)}{4}(2D_\mathrm{s} - \bar D_1 - \bar D_2)$. Also, the $D_1 \bar D_2 - D_2 \bar D_1 > 0$ holds because the unconstrained optimum satisfies $D_1^{\ast} \bar D_2 - D_2^{\ast} \bar D_1 = \frac{\bar D_2 - \bar D_1 }{4} (2D_\mathrm{s} - \bar D_1 - \bar D_2) > 0$. From~\eqref{optimal_Sgeq0}, differentiating with respect to $\alpha$ gives 
    \begin{equation*}
        D_1'(\alpha) =- \frac{(\bar D_2 - \bar D_1) (2D_\mathrm{s} -\bar D_1 - \bar D_2)}{4(\bar D_1 + \bar D_2)} < 0
        \quad \text{and} \quad
        D_2'(\alpha) = \frac{(\bar D_2 - \bar D_1) (2D_\mathrm{s} -\bar D_1 - \bar D_2)}{4(\bar D_1 + \bar D_2)} > 0.
    \end{equation*}
    Here, $D_1'(\alpha) + D_2'(\alpha)= 0$, implying that $D_1(\alpha) + D_2(\alpha) = D_\mathrm{s}$ must hold as $\alpha$ changes. As $\alpha$ increases, the $D_1(\alpha)$ decreases and $D_2(\alpha)$ increases. The threshold $\tilde{\alpha}_3:=\min(\alpha_5,\alpha_6,1)$, at which the constraint switches from slack to binding, is determined by 
     \begin{equation*}
    \begin{aligned}
    \alpha_5&:=\inf\{\alpha|D_1(\alpha)=0\} 
    = 1 + \frac{4\bar D_1 D_\mathrm{s}}
    {(\bar D_2 - \bar D_1)(2D_\mathrm{s} - \bar D_1 -\bar D_2)}>1, \\
    \alpha_6&:=\inf\{\alpha|D_2(\alpha)=\bar{D}_2\} 
    = 1 + \frac{4 \bar D_2 (\bar D_1 +\bar D_2-D_\mathrm{s})}{(\bar D_2 - \bar D_1)(2D_\mathrm{s} - \bar D_1 -\bar D_2)} >1.
    \end{aligned}
    \end{equation*}
    The first inequality follows from the condition of \textbf{Scenario~(1)}, $\bar D_1 + \bar D_2 < 2 D_{\mathrm s}$. The second inequality follows from our modeling setting that $D_{\mathrm s} < \bar D_1 + \bar D_2$. Moreover, under the case $\lambda > 0$, the equality $D_1(\alpha) + D_2(\alpha) = D_{\mathrm s}$ holds for all $\alpha$, and therefore does not affect the threshold. Consequently, $\tilde{\alpha}_3 = 1$.
    
    The DCNW change satisfies
    \begin{equation*}
        \begin{aligned}
            W_\mathrm{D}'(\alpha) = -\frac{(\bar D_2 - \bar D_1) (2D_\mathrm{s} - \bar D_1 -\bar D_2)}{4}
            \frac{(\bar D_2 -\bar D_1) D_\mathrm{s} - 2T_2(\alpha)}
            {(\bar D_1 D_\mathrm{s} + T_2(\alpha)) (\bar D_2 D_\mathrm{s} - T_2(\alpha)) } < 0.
        \end{aligned}
    \end{equation*}
    To verify the sign, note that
    $$
    (\bar D_2 -\bar D_1) D_\mathrm{s} - 2T_2(\alpha) = (\bar D_2 -\bar D_1)
    \left(D_\mathrm{s} - \frac{1-\alpha}{2} (2D_\mathrm{s} - \bar D_1- \bar D_2)\right).$$
    Since $2D_\mathrm{s} - \bar D_1- \bar D_2>0$, the term in parentheses is minimized at $\alpha = 0$, where it equals $\sfrac{(\bar D_1 + \bar D_2)}{2}$. Hence, 
    \begin{equation}
        (\bar D_2 -\bar D_1) D_\mathrm{s} - 2T_2(\alpha) > 0,\label{inter_scenario1}
    \end{equation}
    for all $\alpha \in [0,1]$. Moreover, $\bar D_2 D_{\mathrm s}-T_2(\alpha)>0$ since $D_2(\alpha)>0$.
    
    The change in the utility of each consumer is given by
    \begin{equation*}
        \begin{aligned}
            U_1'(\alpha) = aD_1 (\alpha) D_1'(\alpha) = -\frac{a(\bar D_2 - \bar D_1) (2D_\mathrm{s} - \bar D_1 -\bar D_2)}
            {4 (\bar D_1 + \bar D_2)^2} 
            (\bar D_1 D_\mathrm{s} + T_2(\alpha)) < 0,\\
            U_2'(\alpha) = aD_2 (\alpha) D_2'(\alpha) = \frac{a(\bar D_2 - \bar D_1) (2D_\mathrm{s} - \bar D_1 -\bar D_2)}
            {4 (\bar D_1 + \bar D_2)^2} 
            (\bar D_2 D_\mathrm{s} - T_2(\alpha)) > 0,\\
        \end{aligned}
    \end{equation*}
    and the change of total consumer utility satisfies the following due to~\eqref{inter_scenario1}
    $$
    U'(\alpha) = - \frac{a(\bar D_2 - \bar D_1) (2D_\mathrm{s} - \bar D_1 -\bar D_2)}
            {4 (\bar D_1 + \bar D_2)^2} \left((\bar D_2 -\bar D_1) D_\mathrm{s} - 2T_2(\alpha)\right) < 0.
    $$
    Since the total consumer utility decreases as $\alpha$ increases, and the aggregator profit also decreases when the fairness constraint is imposed, social welfare decreases as well. These imply the Regime $3$. As $\tilde{\alpha}_3 =1$, there is no transition from this regime to another regime.
    
    \textbf{Scenario (2).} $\bar D_1 + \bar D_2 \geq 2D_\mathrm{s}$, which implies $T_2(\alpha) = \frac{(1-\alpha)(\bar D_2 - \bar D_1)}{4}( \bar D_1 + \bar D_2 - 2D_\mathrm{s})$. Also, we have $D_1 \bar D_2 - D_2 \bar D_1 < 0$ because the unconstrained optimum is $D_1^{\ast} \bar D_2 - D_2^{\ast} \bar D_1 = \frac{\bar D_2 - \bar D_1}{4} (2D_\mathrm{s} - \bar D_1 - \bar D_2) < 0$. From~\eqref{optimal_Sleq0}, differentiating with respect to $\alpha$ gives 
    \begin{equation*}
        D_1'(\alpha) = \frac{(\bar D_2 - \bar D_1) (\bar D_1 + \bar D_2 - 2D_\mathrm{s})}{4(\bar D_1 + \bar D_2)} > 0,
        \quad \text{and} \quad
        D_2'(\alpha) = -\frac{(\bar D_2 - \bar D_1) (\bar D_1 + \bar D_2 - 2D_\mathrm{s})}{4(\bar D_1 + \bar D_2)} < 0,
    \end{equation*}
    with $D_1'(\alpha) + D_2'(\alpha)= 0$. As $\alpha$ increases, $D_1(\alpha)$ increases and $D_2(\alpha)$ decreases. Denote the threshold $\tilde{\alpha}_4$ as
    \begin{equation*}
    \begin{aligned}
    \tilde{\alpha}_4&:=\inf\{\alpha|D_1(\alpha)=\bar D_1\} 
    = 1 + \frac{4 \bar D_1 (\bar D_1 +\bar D_2 - D_\mathrm{s})}{(\bar D_2 - \bar D_1)( \bar D_1 + \bar D_2 - 2D_\mathrm{s} )} >1,
    \end{aligned}
    \end{equation*}
    which implies that $\tilde{\alpha}_4 = 1$.

    The DCNW change satisfies
    \begin{equation*}
        \begin{aligned}
            W_\mathrm{D}'(\alpha) = \frac{(\bar D_2 - \bar D_1) ( \bar D_1 + \bar D_2 - 2D_\mathrm{s})}{4}
            \frac{(\bar D_2 -\bar D_1) D_\mathrm{s} + 2T_2(\alpha)}
            {(\bar D_1 D_\mathrm{s} - T_2(\alpha)) (\bar D_2 D_\mathrm{s} + T_2(\alpha)) } > 0,
        \end{aligned}
    \end{equation*}
    The change in the utility of each consumer is given by
    \begin{equation*}
        \begin{aligned}
            U_1'(\alpha) = \frac{a(\bar D_2 - \bar D_1) ( \bar D_1 + \bar D_2 - 2D_\mathrm{s})}
            {4 (\bar D_1 + \bar D_2)^2} 
            (\bar D_1 D_\mathrm{s} - T_2(\alpha)) > 0,\\
            U_2'(\alpha) = -\frac{a(\bar D_2 - \bar D_1) (\bar D_1 + \bar D_2 - 2D_\mathrm{s})}
            {4 (\bar D_1 + \bar D_2)^2} 
            (\bar D_2 D_\mathrm{s} + T_2(\alpha)) < 0.
        \end{aligned}
    \end{equation*}
    The change of total consumer utility satisfies
    $$
    U'(\alpha) = - \frac{a(\bar D_2 - \bar D_1) (\bar D_1 + \bar D_2 - 2D_\mathrm{s})}
            {4 (\bar D_1 + \bar D_2)^2} \left((\bar D_2 -\bar D_1) D_\mathrm{s} + 2T_2(\alpha)\right) < 0
    $$
    Since the total consumer utility decreases as $\alpha$ increases, social welfare also decreases. These characterize the Regime $4$. As $\tilde{\alpha}_4 =1$, there is no transition from this regime to another regime.

    \noindent\textbf{\boldmath Regimes~$1$ and $2$: $D_1^\ast  =\bar D_1$}.\\
    We use the same Lagrangian multiplier $\lambda$ and $\eta$ as in \textbf{\boldmath Regimes~$3$ and $4$}, and introduce an additional Lagrangian multiplier $\nu$ for the upper-bound constraint $D_1 \leq \bar D_1$. We consider two cases: slack ($\lambda = 0$) and binding $(\lambda > 0)$.

    \textbf{\boldmath Case $\lambda > 0$.} 
    Suppose $D_1^\ast + D_2^\ast = D_\mathrm{s}$, we directly get the solution $D_1^\ast = \bar D_1$ and $D_2^\ast = D_\mathrm{s} - \bar D_1<\bar D_2$. The initial energy ratio gap is 
    $$\Delta_D : = \left|1 - \frac{D_\mathrm{s} - \bar D_1}{\bar D_2}\right| = \frac{\bar D_1 + \bar D_2 - D_\mathrm{s}}{\bar D_2}.$$ 
    The energy fairness condition can be written as 
    $$
    \left|D_1\bar D_2 - D_2 \bar D_1\right| \leq (1-\alpha) \bar D_1 \!\left(\bar D_1 + \bar D_2 - D_\mathrm{s}\right).
    $$
    Similar to the analysis in \textbf{\boldmath Regimes~$3$ and $4$}, the sign of $D_1\bar D_2 - D_2 \bar D_1 > 0$ is determined by the unconstrained optimum, $D_1^\ast\bar D_2 - D_2^\ast \bar D_1 = \bar D_1 (\bar D_1 + \bar D_2 - D_\mathrm{s}) > 0$. The Lagrangian can therefore be written as 
    \begin{equation*}
    \begin{aligned}
         \mathcal{L} =& \pi (D_1 + D_2) - \sum_{i=1}^2 \!\left( a D_i^2 + (b - a \bar D_i) D_i \right) + \eta \left((1-\alpha) \bar D_1 \!\left(\bar D_1 + \bar D_2 - D_\mathrm{s}\right) - \left(D_1\bar D_2 - D_2 \bar D_1\right)\right) \\
         &+ \lambda \left(D_s - D_1 -D_2\right) + \nu \left(\bar D_1 - D_1\right).
    \end{aligned} 
    \end{equation*}
    Under the binding conditions, $D_2(\alpha)$ cannot increase while $D_1(\alpha) = \bar D_1$, as doing so would violate the energy fairness constraint. This implies that there is only one interior segment, and consequently $\nu = 0$. The KKT conditions then yield the optimal solution 
    \begin{equation}
        D_1(\alpha) = \frac{\bar D_1 \left(\bar D_1 + \bar D_2 + \alpha \left(D_\mathrm{s}- \bar D_1 -\bar D_2\right) \right)}
        {\bar D_1 + \bar D_2}
        \quad\text{and} \quad
        D_2(\alpha) = D_\mathrm{s} - D_1(\alpha).\label{D_1alpha_equal}
    \end{equation}
    Differentiating with respect to $\alpha$ gives
    \begin{equation}
        D_1'(\alpha) = -\frac{\bar D_1 \left(\bar D_1 +\bar D_2-D_\mathrm{s}\right)} 
        {\bar D_1 + \bar D_2} < 0
        \quad\text{and} \quad
        D_2'(\alpha) = -D_1'(\alpha) > 0,\label{D_2alpha_equal}
    \end{equation}   
    and therefore $D_1'(\alpha) + D_2'(\alpha) = 0$. Thus, as $\alpha$ increases, the $D_1(\alpha)$ decreases while $D_2(\alpha)$ increases. The threshold $\tilde{\alpha}_5:=\min(\alpha_{7},\alpha_{8},1)$, which represents the first value of $\alpha$ at which a boundary condition of the box constraint is violated, where
    \begin{equation*}
    \begin{aligned}
    \alpha_{7}&:=\inf\{\alpha|D_1(\alpha)=0\} 
    = 1 + \frac{D_\mathrm{s}}  {\bar D_1 + \bar D_2 - D_\mathrm{s}} > 1, \\
    \alpha_{8}&:=\inf\{\alpha|D_2(\alpha)=\bar{D}_2\} 
    = 1+ \frac{\bar D_2}{\bar D_1} >1.
    \end{aligned}
    \end{equation*}
    Thus, $\tilde{\alpha}_5 = 1$.

    The change in DCNW satisfies
    \begin{equation*}
        \begin{aligned}
            W_\mathrm{D}'(\alpha) = -\frac{\bar D_1 \!\left(\bar D_1 -\bar D_2-D_\mathrm{s} \right)}{\bar D_1 + \bar D_2} 
            \left( \frac{1}{D_1(\alpha)} - \frac{1}{D_2(\alpha)} \right).
        \end{aligned}
    \end{equation*}
    Since $D_1^\ast < D_2^\ast$ by Proposition~\ref{prop:profit_only}, we have $W_\mathrm{D}'(0) < 0$. As $\alpha$ increases, $\tfrac{1}{D_1(\alpha)}$ increases and $\tfrac{1}{D_2(\alpha)}$ decreases, implying that $W_\mathrm{D}'(\alpha)$ increases. The necessary and sufficient conditions for $W_\mathrm{D}'(\alpha) = 0$ is
    \begin{equation*}
        \begin{aligned}
            D_2(\alpha) - D_1(\alpha) = 0 
            &\iff D_\mathrm{s} - 2D_1(\alpha) = 0\\
            &\iff \left(\bar D_1  +\bar D_2\right)D_\mathrm{s} - 2\bar D_1 \!\left(\bar D_1 + \bar D_2 + \alpha (D_\mathrm{s} - \bar D_1 - \bar D_2) \right)= 0 \\
            &\iff \alpha = \frac{\left(\bar D_1 + \bar D_2\right) \!\left(D_\mathrm{s} - 2\bar D_1\right)}
            {2\bar D_1 \left(D_\mathrm{s} - \bar D_1 - \bar D_2\right)} > 1.
        \end{aligned}
    \end{equation*}
    The last inequality holds because $D_\mathrm{s} - 2\bar D_1 > D_\mathrm{s} - \bar D_1 - \bar D_2$ and $\bar D_1 + \bar D_2 > 2\bar D_1$. Thus, $W_\mathrm{D}'(\alpha) < 0$ for all $\alpha \in [0,1]$, implying DCNW decreases monotonically with $\alpha$. Consequently, the CNW also decreases monotonically.
    
    The change in utility of each consumer is
    \begin{equation*}
        \begin{aligned}
            U_1'(\alpha) &= -a D_1(\alpha) \frac{\bar D_1 \!\left(\bar D_1 +\bar D_2-D_\mathrm{s}\right)}{\bar D_1 + \bar D_2} < 0,\\
            U_2'(\alpha) &= a D_2(\alpha) \frac{\bar D_1 \!\left(\bar D_1 +\bar D_2-D_\mathrm{s} \right)}{\bar D_1 + \bar D_2} > 0.\\
        \end{aligned}
    \end{equation*}
    The total consumer utility satisfies
    $$
    U'(\alpha) = \frac{a \bar D_1 }{\bar D_1 + \bar D_2} \left(\bar D_1 +\bar D_2-D_\mathrm{s} \right)\left( D_\mathrm{s}-2D_1(\alpha)\right).
    $$
    From Proposition~\ref{prop:profit_only}, $D_1^\ast < D_2^\ast$ with $D_1^\ast = \bar D_1$ and $D_2^\ast = D_{\mathrm s} - \bar D_1$ implies $2\bar D_1 < D_{\mathrm s}$, and hence $U'(0) > 0$. As $\alpha$ increases, $D_{\mathrm s} -2D_1(\alpha) > 0$ and is increasing. Together with $\bar D_1 + \bar D_2 -D_{\mathrm s} > 0$, this implies that $U'(\alpha) > 0$ for all $\alpha \in [0,1]$, and therefore total consumer utility monotonically increases.

    To determine the social welfare change, we first derive the aggregator profit change, 
    $$
    \Pi'(\alpha) = D_1'(\alpha) \left(\pi - 2a D_1(\alpha) - b + a\bar D_1 \right)
    + D_2'(\alpha) \left(\pi - 2a D_2(\alpha) - b + a\bar D_2 \right).
    $$
    According to~\eqref{D_2alpha_equal} and $D_1(\alpha) + D_2(\alpha) = D_\mathrm{s}$, $\Pi'(\alpha)$ can be simplified to 
    $$
    \Pi'(\alpha) = aD_1'(\alpha) \left( 2 D_\mathrm{s} - 4D_1(\alpha) + \bar D_1 - \bar D_2 \right).
    $$ 
    Then, the change in social welfare satisfies
    \begin{equation*}
        \begin{aligned}
            W_\mathrm{SW}'(\alpha) &= U'(\alpha) + \Pi'(\alpha)
            = aD_1'(\alpha)\left(2D_1(\alpha) - D_\mathrm{s} \right)
            + aD_1'(\alpha) \left( 2 D_\mathrm{s} - 4D_1(\alpha) + \bar D_1 - \bar D_2 \right)\\
            &=aD_1'(\alpha) \left( D_\mathrm{s} - 2D_1(\alpha) + \bar D_1 -\bar D_2 \right).
        \end{aligned}
    \end{equation*}
    As $D_1'(\alpha) < 0$, the sign of $W_\mathrm{SW}'(\alpha)$ is determined by the second term. According to~\eqref{D_1alpha_equal}, the second term is equivalent to
    $$
    D_\mathrm{s} - 2D_1(\alpha) + \bar D_1 -\bar D_2
    = \frac{(\bar D_1 + \bar D_2 -D_\mathrm{s})
    \left((2\alpha -1 ) \bar D_1 - \bar D_2 \right)}{\bar D_1 + \bar D_2} < 0,
    $$
    because $\bar D_1 + \bar D_2 -D_\mathrm{s}> 0$ and $(2\alpha -1 ) \bar D_1 - \bar D_2 < 0$ for all $\alpha \in [0,1]$. Thus, $W_\mathrm{SW}'(\alpha) > 0$ for all $\alpha \in [0,1]$.
    
    These characterize the Regime~$2$ for any $\alpha\le\tilde \alpha_5$.
     \[
    \boxed{
    U_1~\text{decreases},\quad
    U_2~\text{increases},\quad
    U~\text{increases},\quad
    W_\mathrm{SW}~\text{increases},~\text{and}~
    W_\mathrm{CNW}~\text{decreases}.\quad(\text{Regime}~$2$)
    }
    \]

    As $\tilde{\alpha}_5 =1$, there is no transition from this regime to another regime.
    
     \textbf{\boldmath Case $\lambda = 0$.}  Suppose $D_1^\ast + D_2^\ast < D_\mathrm{s}$. The initial energy ratio gap is 
     $$\Delta_D : = \left|1 - \frac{\pi - b}{2a\bar D_2} - \frac{1}{2}\right| = \left|\frac{a\bar D_2 - (\pi - b)}{2a\bar D_2}\right|.$$ 
     The energy-fairness condition is expressed as 
     $$
     \left|D_1\bar D_2 - D_2 \bar D_1\right| 
     \leq \frac{(1-\alpha)\bar D_1}{2a} \left|a\bar D_2 - (\pi-b)\right|.
     $$
     Since $D_2^\ast = \frac{\pi - b}{2a} + \frac{\bar D_2}{2} < \bar D_2$, it follows that
    \begin{equation}
    \pi - b < a \bar D_2.
    \label{eq:ordering_pi_b}
    \end{equation}
     Similar to \textbf{\boldmath Case $\lambda =0$} in Regimes~$3$ and $4$, $D_1\bar D_2 - D_2 \bar D_1 > 0$. The sign of $D_1\bar D_2 - D_2\bar D_1$ is determined by the unconstrained optimum $D_1^{\ast}\bar D_2 - D_2^{\ast} \bar D_1 = \frac{\bar D_1 (a\bar D_2 - (\pi - b))}{2a}> 0$. Let $T_3(\alpha) := \frac{(1-\alpha) \bar D_1}{2a}(a\bar D_2 - (\pi-b))$. The Lagrangian can be regarded as
    \begin{equation*}
        \mathcal{L}
    = \pi (D_1 + D_2)
    - \sum_{i=1}^2 \left( a D_i^2 + (b - a \bar D_i) D_i \right)
    + \eta \left( T_3(\alpha)- (D_1\bar D_2 - D_2 \bar D_1) \right) + \nu (\bar D_1 - D_1).
    \end{equation*}
    The first-order stationarity conditions yield the optimal solution
    \begin{equation}
        D_1(\alpha) = c+ \frac{\bar D_1}{2} - \frac{\eta}{2a}\bar D_2 - \frac{\nu}{2a}
        \quad \text{and} \quad
        D_2(\alpha) = c + \frac{\bar D_2}{2} + \frac{\eta}{2a} \bar D_1.
        \label{eq:frist_order_energy_fairness}
    \end{equation}
    Due to the upper bound constraint, there are two segments: a cap binding segment ($\nu >0$) and an interior segment ($\nu =0$).    
    
    \textbf{Cap binding segment.} Set $D_1(\alpha) = \bar D_1 $. The KKT conditions give the optimal solution
    \begin{equation}
        D_1(\alpha) = \bar D_1\quad \text{and}\quad
        D_2(\alpha) = \bar D_2 - \frac{(1-\alpha)(a\bar D_2 - (\pi - b)) }{2a}. \label{cap_bindingD_2alpha}
    \end{equation}
    Based on \eqref{eq:frist_order_energy_fairness} and \eqref{cap_bindingD_2alpha}, the corresponding $\nu(\alpha)$ is 
    \begin{equation*}
    \nu(\alpha) = \pi - b-a\bar D_1 - \frac{\alpha \bar D_2 (a \bar D_2 - (\pi - b))}{\bar D_1}.
    \end{equation*}
    Note that since $D_1^\ast = \bar D_1$, we have $c+ \frac{\bar D_1}{2} > \bar D_1$, which implies
    \begin{equation}
        \nu(0) =\pi-b-a\bar D_1> 0.
         \label{nualpha}
    \end{equation}
    
    Differentiating $D_1(\alpha)$ and $D_2(\alpha)$ with respect to $\alpha$ gives
    \begin{equation*}
        D_1'(\alpha) = 0
        \quad \text{and}\quad
        D_2'(\alpha) = \frac{a\bar D_2 - (\pi - b)}{2a} > 0.
    \end{equation*}
    Thus, as $\alpha$ increases, $D_1(\alpha)$ remains constant, $D_2(\alpha)$ increases, and $\nu(\alpha)$ decreases. Let $\tilde{\alpha}_6$ denote the first value of $\alpha$ at which $D_1(\alpha) + D_2(\alpha) = D_\mathrm{s}$ or $\nu(\alpha) = 0$. More precisely, $\tilde{\alpha}_6=\min(\alpha_9,\alpha_{10},1)$, where
    \begin{equation}
    \begin{aligned}
    \alpha_9&:=\inf\{\alpha | D_1(\alpha) + D_2(\alpha) = D_\mathrm{s}\} 
    = 1 - \frac{2 a(\bar D_1 + \bar D_2-D_\mathrm{s})}{a\bar D_2 - (\pi - b)} < 1,\\
    \alpha_{10}& := \inf\{\alpha |\nu(\alpha)=0\} 
    = \frac{\bar D_1(\pi - b - a\bar D_1)}{\bar D_2 (a\bar D_2 - (\pi - b))} > 0,
    \end{aligned} \label{alpha7}
    \end{equation}
    The condition, $\alpha_9<1$ is because $a\bar{D}_2-(\pi-b)>0$ by~\eqref{eq:ordering_pi_b} and $\bar D_1 + \bar D_2>D_\mathrm{s}$.


    Within the current regime, the DCNW change satisfies
    $$
    W_\mathrm{D}'(\alpha) = \frac{D_2'(\alpha)}{D_2(\alpha)} > 0.$$
    The change in the utility of each consumer is given by
    \begin{equation*}
        U_1'(\alpha) = 0
        \quad \text{and}\quad
        U_2'(\alpha) = aD_2(\alpha) D_2'(\alpha) > 0 
    \end{equation*}
    The change of total consumer utility is $U'(\alpha) = U_2'(\alpha) >0$.

    Then the change in social welfare satisfies
    \begin{equation*}
        \begin{aligned}
            W_\mathrm{SW}'(\alpha) &= U'(\alpha) + \Pi'(\alpha)
            = aD_2(\alpha) D_2'(\alpha)
            + D_2'(\alpha) \left(\pi - 2a D_2(\alpha) - b + a\bar D_2 \right)\\
            &=D_2'(\alpha) \left( \pi - aD_2(\alpha) - b + a\bar D_2  \right).
        \end{aligned}
    \end{equation*}
    According to~\eqref{cap_bindingD_2alpha}, we have 
    $$
    \pi - aD_2(\alpha) - b + a\bar D_2 
    = \frac{(1+ \alpha)(\pi - b) + (1-\alpha)a\bar D_2}{2} > 0,
    $$
    because $\pi - b > a\bar D_1 $ from~\eqref{nualpha}. As $D_2'(\alpha) > 0$, the $W_\mathrm{SW}'(\alpha) > 0$ must hold for all $\alpha \in [0,1]$.
    
    These characterize the Regime~$1$ for any $\alpha < \tilde{\alpha}_6$
    \[
    \boxed{
    U_1~\text{remains constant},\quad
    U_2~\text{increase},\quad
    U~\text{increases},\quad
    W_\mathrm{SW}~\text{increases},~\text{and}~
    W_\mathrm{CNW}~\text{increases}.\quad(\text{Regime}~$1$)
    }
    \]

    As $\tilde{\alpha}_6 < 1$, when $\alpha > \tilde{\alpha}_6$, the system dynamics pattern is determined by the sign of $\alpha_9 - \alpha_{10}$,
    \begin{equation*}
        \begin{aligned}
            \alpha_9 - \alpha_{10} = 1 - \frac{2a(\bar D_1 + \bar D_2 - D_\mathrm{s}) + \frac{\bar D_1 (\pi - b - a \bar D_1)}{\bar D_2}}{a\bar D_2 - (\pi - b)}.
        \end{aligned}
    \end{equation*}
    As $a\bar D_2 - (\pi - b) > 0$ since $D_2'(\alpha) > 0$, the sign of $\alpha_9 - \alpha_{10}$ is determined by
    \begin{equation*}
        \begin{aligned}
            \alpha_9 - \alpha_{10} \lessgtr 0
            \iff 2a (\bar D_1 + \bar D_2 - D_\mathrm{s} ) + \frac{\bar D_1 (\pi - b - a\bar D_1)}{\bar D_2 } \lessgtr a\bar D_2 - (\pi - b)
        \end{aligned}
    \end{equation*}
    The left-hand side increases as $D_{\mathrm{s}}$ decreases or $\pi-b$ increases, whereas the right-hand side increases as $\bar D_2$ increases or $\pi-b$ decreases. Since these effects may dominate in either direction depending on the parameter values, $\alpha_9-\alpha_{10}$ can be either positive or negative.
    
    If $\alpha_9 - \alpha_{10} < 0$, the threshold is determined by $\alpha_9$ and $D_1(\alpha) + D_2(\alpha) = D_\mathrm{s}$ at $\alpha=\tilde \alpha_9$. Therefore, the system dynamics follow \textbf{\boldmath Case $\lambda >0$}, which is characterized by Regime~$2$. Otherwise, when $\alpha_9 - \alpha_{10} > 0$, the threshold is determined by $\alpha_{10}$ and the condition $D_1(\alpha) + D_2(\alpha) < D_{\mathrm s}$. In this case, the system dynamics follow the \textbf{Interior segment}, which characterizes Regime~$3$, introduced below.
    
    \textbf{Interior segment.} 
    In this segment, $D_1(\alpha) < \bar D_1$, and hence $\nu = 0$. From the analysis of the \textbf{Cap binding segment}, the system transitions to the interior segment when $\alpha_9 > \alpha_{10}$. Since $\alpha_9 < 1$ by~\eqref{alpha7}, it must be that $\alpha_{10} < 1$, which is equivalent to
    \begin{equation}
        \bar D_2 \bigl(a\bar D_2 - (\pi - b)\bigr)
        + \bar D_1 \bigl(a\bar D_1 - (\pi - b)\bigr) > 0.
        \label{alpha8le1}
    \end{equation}

    The KKT conditions then give the optimal solution
    \begin{equation*}
    \begin{aligned}
         D_1(\alpha) &= c + \frac{\bar D_1}{2} - \frac{c\bar D_2 (\bar D_2 - \bar D_1)}{Q}
         +  \frac{(1-\alpha)\bar D_1 \bar D_2 (a\bar D_2 - (\pi - b))}{2aQ},\\
         D_2(\alpha) &= c + \frac{\bar D_2}{2} + \frac{c\bar D_1 (\bar D_2 - \bar D_1)}{Q}
         - \frac{(1-\alpha)\bar D_1^2 (a\bar D_2 - (\pi - b))}{2aQ},
    \end{aligned} 
    \end{equation*}
    Differentiating with respect to $\alpha$ gives
    \begin{equation*}
        D_1'(\alpha) = -\frac{\bar D_1 \bar D_2 (a\bar D_2 - (\pi - b))}{2aQ} < 0
        \quad \text{and}\quad
        D_2'(\alpha) = \frac{\bar D_1^2 (a\bar D_2 - (\pi - b))}{2aQ} > 0,
    \end{equation*}
    which implies that as $\alpha$ increases, $D_1(\alpha)$ decreases and $D_2(\alpha)$ increases. Let $\tilde{\alpha}_7$ denote the first value of $\alpha$ at which either $D_1(\alpha) = 0$ or $D_2(\alpha) = \bar D_2$. More precisely, $\tilde{\alpha}_7=\min(\alpha_{11},\alpha_{12},1)$, where
    \begin{equation*}
    \begin{aligned}
    \alpha_{11}&:=\inf\{\alpha|D_1(\alpha)=0\} = 1 + \frac{\bar D_1 (a\bar D_1 + \pi - b) + \bar D_2 (a\bar D_2 + \pi - b)}
    {\bar D_2 (a\bar D_2 - (\pi - b))} > 1, \\
    \alpha_{12}&:=\inf\{\alpha|D_2(\alpha)=\bar{D}_2\} = 1 
    + \frac{\bar D_2 \left(\bar D_1 (a\bar D_1 - (\pi - b)) + \bar D_2 (a\bar D_2 - (\pi - b)) \right)}
    {\bar D_1^2 (a\bar D_2 - (\pi - b))}>1.
    \end{aligned}
    \end{equation*}
    Here, $\alpha_{11} > 1$ holds because $a\bar D_i + \pi - b>0$ for any $i\in [2]$ by assumption, while the condition $\alpha_{12} > 1$ follows from~\eqref{alpha8le1}. Thus, $\tilde{\alpha}_7 = 1$.

    Within the current regime, the DCNW change satisfies
    \begin{equation*}
        W_\mathrm{D}'(\alpha) = \frac{a \bar D_2 - (\pi -b)}{2aQ}  \frac{\bar D_1 (\bar D_1 D_1(\alpha) - \bar D_2 D_2(\alpha))}{D_1(\alpha) D_2(\alpha)} < 0,
    \end{equation*}
    where the inequality follows from $a\bar D_2 - (\pi - b) > 0$ by \eqref{eq:ordering_pi_b}, $D_1(\alpha) < D_1^\ast < D_2^\ast < D_2(\alpha)$, and $\bar D_1 < \bar D_2$.
    
    Each consumer’s utility change is given by
    \begin{equation*}
        U_1'(\alpha) = aD_1(\alpha)D_1'(\alpha) < 0
        \quad \text{and}\quad
        U_2'(\alpha) = aD_2(\alpha) D_2'(\alpha) > 0, 
    \end{equation*}
    and the total consumer utility change is 
    $$U'(\alpha) = \frac{a \bar D_2 - (\pi -b)}{2Q} \bar D_1 \left( \bar D_1 D_2(\alpha) - \bar D_2 D_1(\alpha) \right) < 0,$$
    because $D_1(\alpha)\bar D_2 - D_2(\alpha) \bar D_1 > 0$ under the case $\lambda=0$.
    Since total consumer utility decreases as $\alpha$ increases, social welfare also decreases. These imply the Regime~$3$. As $\tilde{\alpha}_7 = 1$, there is no transition from this regime to another regime.\hfill \Halmos
\end{proof}
\medskip

\begin{proof}{Proof of Theorem~\ref{thm:price_fairness}.}
We analyze the two–consumer optimization problem under the price fairness criterion.
\begin{equation*}
\begin{aligned}
\max_{p_1, p_2}&~\pi (D_1+D_2) - p_1D_1 - p_2D_2\\
\text{s.t. }&~ D_i=\left(\frac{p_i-b}{a}+\bar{D}_i\right)\mathbb{I}\left\{b-a\bar{D}_i \le p_i \le b\right\} + \bar{D}\,\mathbb{I}\left\{p_i > b\right\},~\forall i \in [2],\\
&~D_i\in[0,\bar{D}_i],~\forall i\in[2],\\
&~ D_1 + D_2 \le D_\mathrm{s},\\
&~ |p_1 - p_2| \le (1 - \alpha)|p_1^\ast - p_2^\ast|.
\end{aligned}
\end{equation*}

Let $D_i^\ast$ (resp. $D_i(\alpha)$) and $p_i^\ast$ (resp. $p_i(\alpha)$) denote the no-fairness (resp. $\alpha$-price fairness) optimal provided energy and price of consumer~$i$, and define the initial price gap $\Delta_p:=|p_1^\ast-p_2^\ast|$. Because each $D_i$ is piecewise linear in $p_i$, the analysis proceeds by distinguishing whether each optimal provided energy lies in the interior $(0,\bar D_i)$ or $D_1^\ast$ lies on a boundary.

Because the regimes are connected sequentially (Regime~$1 \rightarrow 2 \rightarrow 3$), we begin by analyzing the final regime (Regime~$3$) and then proceed backward to Regime~$2$ and Regime~$1$, so that each preceding case builds upon the characterization of the subsequent, more constrained one.

\noindent\textbf{Regime~$\boldsymbol{3}$: Boundary condition when $\boldsymbol{D_1^\ast=0}$.}\\
Consider the boundary active set $D_1^\ast=0$. The one-dimensional maximization problem reduces to
\begin{equation*}
\max_{D_2 \in \left[0,D_{\mathrm{s}}\right]}\;
\pi D_2 - \big(a D_2^2 + (b - a\bar D_2) D_2\big).
\end{equation*}
The upper bound excludes $\bar D_2$, since feasibility of the optimal solution $(D_1^\ast,D_2^\ast)=(0,D_{\mathrm s})$ requires $D_2^\ast=D_{\mathrm s}\le \bar D_2$.
The unconstrained maximizer is 
\begin{equation*}
D_2^{\ast}=\min\!\left(\frac{\pi - b + a\bar D_2}{2a},\,D_{\mathrm{s}}\right)
\quad\text{and}\quad
p_2^{\ast} = aD_2^{\ast} + b - a\bar D_2.
\end{equation*}
The corresponding profit is
\begin{equation*}
\Pi^{\ast}
= \Pi(0,D_2^{\ast})
= \pi D_2^{\ast}
- \big(a(D_2^{\ast})^2 + (b - a\bar D_2)D_2^{\ast}\big).
\end{equation*}
For instance, if $D_2^{\ast}=\tfrac{\pi - b + a\bar D_2}{2a}$, then $\Pi^{\ast}=\tfrac{(\pi - b + a\bar D_2)^2}{4a}$.

In price space, the condition $D_1^\ast=0$ implies $p_1^\ast\le b-a\bar D_1$. 
If $p^{\ast}\le b-a\bar D_1$, we can set $p_1(\alpha)=p^{\ast}$ for all~$\alpha\in(0,1]$, which achieves perfect price fairness. To verify this,
\begin{equation*}
\begin{aligned}
b-a\bar D_1-p^{\ast}
&= b-a\bar D_1-(aD_2^{\ast}+b-a\bar D_2)\\
&= a(\bar D_2-\bar D_1)-aD_2^{\ast}\\
&= a(\bar D_2-\bar D_1)
     -a\!\left(\frac{\pi-b-\lambda}{2a}+\frac{\bar D_2}{2}\right)\\
&= a\!\left(\frac{\bar D_2}{2}-\frac{\bar D_1}{2}\right)
   -a\!\left(\frac{\pi-b-\lambda}{2a}+\frac{\bar D_1}{2}\right)\\
&\ge a\!\left(\frac{\bar D_2}{2}-\frac{\bar D_1}{2}\right)\ge0,
\end{aligned}
\end{equation*}
where the third equality substitutes $D_2^{\ast}$ from~\eqref{eq:clip} and the first inequality follows from the condition 
$\tfrac{\pi-b-\lambda}{2a}+\tfrac{\bar D_1}{2}\le0$ implied by~\eqref{eq:clip} and $\bar D_1=0$.

Therefore, we can set
\begin{equation*}
p_1(\alpha)=p^{\ast}=aD_2^{\ast}+b-a\bar D_2,
\end{equation*}
which ensures perfect price fairness. Since When $D_1(\alpha)=0$ and $p_2(\alpha)$ is fixed at its value at $\alpha=0$, neither provided energy nor utilities change as $\alpha$ varies. Hence,
\begin{equation*}
\boxed{
U_1,~U_2,~U,~W_{\mathrm{CNW}},~\text{and}~W_{\mathrm{SW}}~\text{remain constant.}
\quad(\text{Regime~3})
}
\end{equation*}

The expressions above are formulated in terms of the initial optimal solution $(D_1^\ast, D_2^\ast)$. 
If this regime follows a transition from Regime~$2$, the corresponding quantities can be expressed as $(D_1(\alpha), D_2(\alpha))$, where $\alpha$ denotes the point at which the regime change occurs. In this case, the qualitative behavior of the solution and the relative ordering of the threshold values remain unchanged.

\noindent\textbf{\boldmath Regime~$2$: $0 < D_i^\ast < \bar D_i$.}\\
In this case, let $\lambda \ge 0$ denote the Lagrangian multiplier on the aggregated energy constraint $D_1 + D_2 \le D_{\mathrm{s}}$. As in Lemma~\ref{lemma:opt_sol}, we analyze two cases: slack $(\lambda = 0)$ and binding $(\lambda > 0)$.

\textbf{\boldmath Case $\lambda = 0$.}
Suppose $D_1^\ast + D_2^\ast < D_\mathrm{s}$ and $0 < D_i^\ast < \bar D_i$ for all $i \in [2]$. 
In this interior regime, the unconstrained maximizer satisfies 
$D_i^\ast = \frac{\pi - b}{2a} + \frac{\bar D_i}{2}$ and the corresponding price is $p_i^\ast = a D_i^\ast + (b - a \bar D_i)$. Substituting the first expression into the second gives 
\begin{equation*}
p_i^\ast 
= a \left( \frac{\pi - b}{2a} + \frac{\bar D_i}{2} \right) + (b - a \bar D_i)
= \frac{\pi + b}{2} - \frac{a \bar D_i}{2}.
\end{equation*}
Hence, the initial price gap is 
\begin{equation*}
\Delta_p 
:= |p_1^\ast - p_2^\ast| 
= \frac{a}{2}\,|\bar D_2 - \bar D_1|
= \frac{a}{2}(\bar D_2 - \bar D_1) > 0.
\end{equation*}

The price fairness condition $|p_1 - p_2| \le (1 - \alpha)\Delta_p$ is equivalently expressed in provided energy space as 
\begin{equation*}
D_1(\alpha) - D_2(\alpha) - \frac{1 + \alpha}{2}(\bar D_1 - \bar D_2) \le 0.
\end{equation*}
Let $\tilde{\alpha}$ denote the first value of $\alpha$ at which the interior solution breaks down, i.e., when $D_1(\alpha)=0$. For $\alpha<\tilde{\alpha}$, Because both provided energy ($D_i\in[0,\bar{D}_i]$) and aggregated energy ($D_1 + D_2 \le D_\mathrm{s}$) constraints are slack, the Lagrangian with multiplier $\eta \ge 0$ for the price fairness constraint is
\begin{equation*}
\mathcal{L}
= \pi (D_1 + D_2)
- \sum_{i=1}^2 \!\left( a D_i^2 + (b - a \bar D_i) D_i \right)
+ \eta \left( - D_1 + D_2 + \frac{1 + \alpha}{2}(\bar D_1 - \bar D_2) \right).
\end{equation*}
Then, the first-order conditions yield
$D_1'(\alpha) = -\frac{\Delta_p}{2a}$ and $D_2'(\alpha) = +\frac{\Delta_p}{2a}$,
so
\begin{equation}
\frac{d}{d\alpha}\big(D_1(\alpha)+D_2(\alpha)\big)=0.
\label{eq:net_change_0}
\end{equation}
Hence $D_1(\alpha)+D_2(\alpha)$ is constant on this segment, if $D_1^\ast+D_2^\ast<D_{\mathrm s}$ at $\alpha=0$, the aggregated energy constraint remains slack for $\alpha<\tilde{\alpha}$.

At $\alpha = 0$, the unconstrained optimum reaches the boundary of the constraint $|p_1 - p_2| \le (1 - \alpha)\Delta_p$ since $|p_1^\ast - p_2^\ast| = \Delta_p$. 
For $\alpha > 0$, the feasible region $\{(p_1,p_2): |p_1 - p_2| \le (1 - \alpha)\Delta_p\}$ shrinks, while the unconstrained optimum would require $|p_1 - p_2| = \Delta_p$, which lies strictly outside the feasible set. Thus, any candidate optimum with slack, $|p_1 - p_2| < (1 - \alpha)\Delta_p$, must lie in the interior of the feasible set. But an interior point cannot be optimal because moving in the direction of the (unconstrained) maximizer strictly increases profit under the strict concavity of the objective function. This contradicts optimality, so the price fairness constraint must bind for all $\alpha > 0$. Consequently,
\begin{equation*}
D_1(\alpha) - D_2(\alpha) = \frac{1 + \alpha}{2}(\bar D_1 - \bar D_2)
\quad\text{and}\quad
\eta(\alpha) = \frac{a}{2}(\bar D_2 - \bar D_1)\alpha > 0.
\end{equation*}

Substituting this into the first-order conditions yields the closed-form solution
\begin{equation}
D_1(\alpha) 
= \frac{\pi - b}{2a} + \frac{\bar D_1}{2} - \frac{\Delta_p}{2a}\,\alpha,
\quad\text{and}\quad
D_2(\alpha) 
= \frac{\pi - b}{2a} + \frac{\bar D_2}{2} + \frac{\Delta_p}{2a}\,\alpha,
\label{eq:demand_changes_price_eq}
\end{equation}
and the multiplier increases linearly as $\eta(\alpha) = \Delta_p \alpha$.

We now take a closer look at $\tilde{\alpha}$, the first value of $\alpha$ at which the interior solution ceases to hold. Since $D_1(\alpha)$ decreases and $D_2(\alpha)$ increases while their sum remains constant, there are two possible breakdown points, either $D_1$ hits $0$ or $D_2$ reaches its upper bound $\bar D_2$. Accordingly, we define
\begin{equation*}
\begin{aligned}
\alpha_1&:=\inf\{\alpha|D_1(\alpha)=0\}=\frac{2a\,D_1^\ast}{\Delta_p}
=\frac{2\big((\pi-b)+a\bar D_1\big)}{a(\bar D_2-\bar D_1)},\\
\alpha_2&:=\inf\{D_2(\alpha)=\bar{D}_2\}=\frac{2a\,(\bar D_2-D_2^\ast)}{\Delta_p}
=\frac{2\big(a\bar D_2-(\pi-b)\big)}{a(\bar D_2-\bar D_1)}.
\end{aligned}
\end{equation*}

We must also account for Regime~$3$, which corresponds to price fairness being satisfied outside the interior price domain $p_i \in [b - a\bar D_i, b]$. The expression in \eqref{eq:demand_changes_price_eq} characterizes the price fairness condition only within this interior region and therefore provides merely one candidate solution. Regime~$3$, by contrast, yields an alternative candidate in which price fairness is attained at a different admissible boundary of the price space. Thus, to determine which regime prevails, all feasible price fairness candidates across the admissible price ranges must be compared.

Therefore, if $(D_1(\alpha), D_2(\alpha))$ yields the same profit as that generated in Regime~$3$, then the system transitions from the current regime to Regime~$3$. Hence $\alpha_3$ is defined by
$$
\alpha_3 := \{\alpha \mid \Pi(D_1(\alpha), D_2(\alpha)) = \Pi^{(3)}\},
$$
where $\Pi^{(3)} = \max_{D_2 \in [0,\,\min(\bar D_2,\,D_s)]} \Pi(0, D_2)$.

When $\lambda = 0$, Regime~$3$ attains its maximum at $D_2 = D_2^\ast$, so $\Pi^{(3)} = \Pi(0, D_2^\ast) = \pi D_2^\ast - p_2^\ast D_2^\ast$. Therefore, $\alpha_3$ satisfies
$$
\pi(D_1(\alpha_3) + D_2(\alpha_3)) - p_1(\alpha_3) D_1(\alpha_3) - p_2(\alpha_3) D_2(\alpha_3)
= \pi D_2^\ast - p_2^\ast D_2^\ast,
$$
and substituting the expressions for $D_i(\alpha)$, $D_i^\ast$, and $p_i^\ast$ yields
$$
\alpha_3 = \frac{\sqrt{2}\,\big((\pi - b) + a \bar D_1\big)}{a(\bar D_2 - \bar D_1)}=\frac{\alpha_1}{\sqrt{2}}<\alpha_1.
$$
Thus, the threshold $\alpha_1$ can be dismissed.

Likewise, we can rule out the threshold $\alpha_2$, because the condition $\alpha_2 < 1$ is equivalent to
\begin{equation*}
    \alpha_2<1 
    \iff 2\big(a\bar D_2-(\pi-b)\big) < a(\bar D_2-\bar D_1)
   \iff \frac{\pi-b}{2a} > \frac{\bar D_1+\bar D_2}{4}.
\end{equation*}
The latter inequality is infeasible when $\lambda = 0$, since it would require
$$D_1^\ast=\frac{\pi-b}{2a}+\frac{\bar{D}_1}{2}>\frac{3\bar{D}_1+\bar{D}_2}{4}>\bar{D}_1,$$
which contradicts the feasibility condition $D_1^\ast \le \bar D_1$.

On the other hand, the condition $\alpha_3<1$ is equivalent to
\begin{equation*}
\alpha_3 < 1
\iff D_1^\ast=\frac{\pi-b}{2a}+\frac{\bar{D}_1}{2}<\frac{\bar{D}_2-\bar{D}_1}{2\sqrt{2}},
\end{equation*}
which is feasible because the right-hand side is positive.

In conclusion, for $\alpha< \alpha_3$, differentiating with respect to $\alpha$ gives 
$D_1'(\alpha) = -\frac{\Delta_p}{2a}$ 
and $D_2'(\alpha) = +\frac{\Delta_p}{2a}$. 
Hence
$$U_1'(\alpha) = \frac{d}{d\alpha}\left(\frac{1}{2}aD_1(\alpha)^2\right) =a D_1(\alpha) D_1'(\alpha) = -\frac{\Delta_p}{2} D_1(\alpha) < 0\quad 
\text{and}\quad U_2'(\alpha) = a D_2(\alpha) D_2'(\alpha) = +\frac{\Delta_p}{2} D_2(\alpha) > 0,$$
where the prime denotes differentiation with respect to $\alpha$. Similarly,
\begin{equation*}
U'(\alpha) = U_1'(\alpha)+U'_2(\alpha)= \frac{\Delta_p}{2}\!\left( D_2(\alpha) - D_1(\alpha) \right)
\quad\text{and}\quad
W'_{\mathrm{CNW}}(\alpha) = \frac{U_1'(\alpha)}{U_1(\alpha)}+\frac{U_2'(\alpha)}{U_2(\alpha)}=\frac{\Delta_p}{a}\!\left( \frac{1}{D_2(\alpha)} - \frac{1}{D_1(\alpha)} \right).
\end{equation*}
Indeed,
\begin{equation}
D_2(\alpha)-D_1(\alpha)
= \frac{\bar D_2 - \bar D_1}{2} + \frac{\Delta_p}{a}\,\alpha
= \frac{\bar D_2 - \bar D_1}{2}\,(1+\alpha) > 0,
\label{eq:D2_minus_D1}
\end{equation}
because $\bar D_2 > \bar D_1$ and $\alpha \in [0,1]$. It implies that $U'(\alpha)>0$ and $W'_{\mathrm{CNW}}(\alpha)<0$.

Lastly, social welfare can be written as 
$$W_{\mathrm{SW}}=\sum_{i\in[2]}\left(\pi D_i-C_i(D_i)\right)=\sum_{i\in[2]}\left[\pi(D_1+D_2)-\frac{1}{2}aD_i^2-(b-a\bar{D}_i)D_i\right].$$ 
Differentiating $W_{\mathrm{SW}}$ with respect to $\alpha$ yields 
$$W'_{\mathrm{SW}}(\alpha)
=\sum_{i\in[2]} D'_i(\alpha)\left(\pi-b + a(\bar{D}_i - D_i(\alpha))\right)
= aD_2'(\alpha)\left(\bar{D}_2 - \bar{D}_1 + D_1(\alpha) - D_2(\alpha)\right)
= aD_2'(\alpha)(\bar{D}_2-\bar{D}_1)\frac{1-\alpha}{2} > 0,$$ 
where the second equality follows from $D_2'(\alpha)=D_1'(\alpha)$, the third equality follows from~\eqref{eq:D2_minus_D1}, and the inequality holds since $D_2'(\alpha)>0$ and $\bar{D}_2>\bar{D}_1$. Therefore, in this regime, social welfare is strictly increasing in $\alpha$. In summary,
\[
\boxed{
U_1~\text{decreases},\quad
U_2~\text{increases},\quad
U~\text{increases},~W_{\mathrm{CNW}}~\text{decreases}~\text{and}~
W_{\mathrm{SW}}~\text{increases}.\quad(\text{Regime}~$2$)
}
\]

For $\alpha \ge \alpha_3$, the system transitions into Regime~$3$.

\textbf{\boldmath Case $\lambda > 0$.}
Suppose $D_1^\ast + D_2^\ast = D_{\mathrm{s}}$ and both $D_i^\ast$ are interior solutions, i.e.,
\begin{equation*}
D_1^\ast = \frac{D_{\mathrm s}}{2} + \frac{\bar D_1 - \bar D_2}{4}
\quad \text{and} \quad
D_2^\ast = \frac{D_{\mathrm s}}{2} + \frac{\bar D_2 - \bar D_1}{4}.
\end{equation*}
This case can be regarded as \textbf{\boldmath Case~$\lambda=0$}, where the aggregated energy constraint $D_1 + D_2 \le D_{\mathrm{s}}$ does not affect the solution, 
since the adjustments in $D_1$ and $D_2$ with respect to $\alpha$ occur in opposite directions and with equal magnitude, as shown in~\eqref{eq:net_change_0}. Consequently, the total allocation $D_1 + D_2$ remains unchanged, implying that the binding condition $D_1 + D_2 = D_{\mathrm{s}}$ has no substantive effect on the dynamics in the current case. Hence, the present case is identical to Case~$\lambda=0$.

\noindent\textbf{\boldmath Regime~$1$: $D_1^\ast=\bar{D}_1$.}\\
\textbf{\boldmath Case $\lambda = 0$.}
Suppose $D_1^\ast + D_2^\ast < D_{\mathrm s}$, $D_1^\ast = \bar D_1$, and $0 < D_2^\ast < \bar D_2$.
This condition requires that the unconstrained interior optimizer for consumer~$1$ attains or exceeds its upper bound because 
$D_1^\ast = \min\left(\frac{\pi - b}{2a} + \frac{\bar D_1}{2},\, \bar D_1\right)$.
Hence,
\begin{equation}
\frac{\pi - b}{2a} + \frac{\bar D_1}{2} \ge \bar D_1
\quad\Longleftrightarrow\quad
\frac{\pi - b}{2a} \ge \frac{\bar D_1}{2}.
\label{eq:cond_1b}
\end{equation}

At $\alpha = 0$, the corresponding prices are
\begin{equation*}
p_1^\ast = b
\quad\text{and}\quad
p_2^\ast = \frac{\pi + b}{2} - \frac{a \bar D_2}{2}.
\end{equation*}
Therefore,
\begin{equation*}
p_1^\ast - p_2^\ast 
= b - \left(\frac{\pi + b}{2} - \frac{a \bar D_2}{2}\right)
= a\left(\frac{\bar D_2}{2}-\frac{\pi - b}{2a}\right) > 0.
\end{equation*}
Hence, the initial price gap is
\begin{equation*}
\Delta_p 
:= |p_1^\ast - p_2^\ast|
= p_1^\ast - p_2^\ast
= a\!\left(\frac{\bar D_2}{2}-\frac{\pi - b}{2a}\right) > 0.
\end{equation*}
When $D_1^\ast=\bar D_1$, two cases may arise. In the first case, $p_1(\alpha)\ge b$ and $p_1(\alpha)$ is strictly increasing in $\alpha$. In the second case, $p_1(\alpha)$ is non-increasing in $\alpha$. We next compare these two cases.

\textbf{\boldmath (i) $p_1(\alpha)$ is strictly increasing.} Suppose that $p_1(\alpha)$ is strictly increasing with $\alpha$, so that
$p_1(\alpha)=p_1^\ast+\Delta p(\alpha)$ with $\Delta p(\alpha)>0$. To satisfy $\alpha$-price fairness, prices must satisfy
$$p_1(\alpha)-p_2(\alpha)=(1-\alpha)\Delta_p,$$
which implies
$p_2(\alpha)=p_2^\ast+\Delta p(\alpha)+\alpha \Delta_p$. The corresponding provided energy are given by $D_1(\alpha)=\bar D_1$ and $D_2(\alpha)=D_2^\ast +\frac{\Delta p(\alpha)+\alpha \Delta_p}{a}$.

The resulting profit difference between the unconstrained case and the $\alpha$-fairness case is
\begin{equation*}
\begin{aligned}
\Pi(0)-\Pi(\alpha)
=\,&\pi\left(\bar D_1 + D_2^\ast\right)-p_1^\ast \bar D_1-p_2^\ast D_2^\ast -\pi\left(\bar D_1 + \left(D_2^\ast +\frac{\Delta p(\alpha)+\alpha \Delta_p}{a}\right)\right)\\
&+\left(p_1^\ast + \Delta p(\alpha)\right)\bar D_1 
+ \left(p_2^\ast+\Delta p(\alpha)+\alpha\Delta_p\right)
\left(D_2^\ast +\frac{\Delta p(\alpha)+\alpha \Delta_p}{a}\right)\\
=\,&\frac{1}{a}\,\Delta p(\alpha)^2
+\left(
\bar D_1+D_2^*
+\frac{p_2^*-\pi+2\alpha\Delta_p}{a}
\right)\Delta p(\alpha)
+\alpha\Delta_p D_2^*
+\frac{\alpha\Delta_p(\alpha\Delta_p+p_2^*-\pi)}{a}\\
=\,& \frac{1}{a}\,\Delta p(\alpha)^2
+\left(
\bar D_1
+\frac{2\alpha\Delta_p}{a}
\right)\Delta p(\alpha)
+\frac{\alpha^2\Delta_p^2}{a}.
\end{aligned}
\end{equation*}
where the third equality follows from $D_2^\ast = \frac{\pi-b}{2a}+\frac{\bar D_2}{2}$ and $p_2^\ast = aD_2^\ast + b- a\bar D_2$.

For a given $\alpha>0$, the minimizer with respect to $\Delta p(\alpha)$ is given by
$$\Delta p(\alpha)^\ast=-\alpha\Delta_p-\frac{a}{2}\bar D_1<0.$$
However, since $\Delta p(\alpha)$ is constrained to be positive, the optimal choice is obtained by
$$\Delta p(\alpha)^\ast =\epsilon,
$$
where $\epsilon>0$ is an arbitrarily small constant that enforces $\Delta p(\alpha)>0$. The corresponding profit difference is
\begin{equation}
\Pi(0)-\Pi(\alpha)\big|_{\Delta p(\alpha)^\ast}
= \frac{1}{a}\epsilon ^2
+\left(
\bar D_1
+\frac{2\alpha\Delta_p}{a}
\right)\epsilon 
+\frac{\alpha^2\Delta_p^2}{a}.
\label{eq:profit_diff_1}
\end{equation}


\textbf{\boldmath (ii) $p_1(\alpha)$ is non-increasing.} 
In this case, we may formulate the problem directly in terms of $(D_1,D_2)$ instead of $(p_1,p_2)$, because $p_i(\alpha)\le b$ until it reaches the boundary $b-a\bar D_i$, which corresponds to $D_i(\alpha)=0$. We begin from the Lagrangian
\begin{equation*}
\mathcal{L}
= \pi(D_1 + D_2)
- \sum_{i=1}^2 \big(a D_i^2 + (b - a\bar D_i) D_i \big)
+ \nu_1(\bar D_1 - D_1)
+ \eta\left(aD_1 - aD_2 -a\bar{D}_1+a\bar{D}_2 - (1-\alpha)\Delta_p\right),
\end{equation*}
where $\nu_1 \ge 0$ corresponds to the upper bound constraint $D_1 \le \bar D_1$, and $\eta \ge 0$ is the multiplier associated with the binding price fairness constraint. This Lagrangian characterization remains valid for all $\alpha$ prior to the point at which any boundary condition changes.

The first-order conditions are
\begin{subequations}
\begin{align}
    \frac{\partial \mathcal{L}}{\partial D_1} 
    &= \pi - (2a D_1 + b - a\bar D_1) - \nu_1 + a\eta = 0,\label{eq:first_order}\\[2pt]
    \frac{\partial \mathcal{L}}{\partial D_2} 
    &= \pi - (2a D_2 + b - a\bar D_2) - a\eta = 0.\label{eq:second_order}
\end{align}
\end{subequations}

When the upper bound is active ($D_1 = \bar D_1$), \eqref{eq:first_order} gives $\nu_1 = \pi - (b + a\bar D_1) + a\eta$. From \eqref{eq:second_order} we obtain
\begin{equation*}
\eta = \frac{\pi - (2a D_2 + b - a\bar D_2)}{a}.
\end{equation*}

Since the price fairness constraint binds, it follows that
\begin{equation*}
p_1(\alpha)-p_2(\alpha)=b-\left(aD_2(\alpha)+b-a\bar D_2\right) = (1 - \alpha)\Delta_p.
\end{equation*}
With $D_1 = \bar D_1$, this gives
\begin{equation*}
D_2(\alpha) = \bar D_2 - \frac{(1 - \alpha)\Delta_p}{a}=D_2^\ast-\frac{\alpha\Delta_p}{a}.
\end{equation*}

Then, the profit loss incurred under the $\alpha$-fairness constraint, i.e., $\Pi(0)-\Pi(\alpha)$, is given by
\begin{equation}
\pi\left(\bar D_1 + D_2^\ast\right)-p_1^\ast \bar D_1-p_2^\ast D_2^\ast 
-\pi\left(\bar D_1 + \left(D_2^\ast -\frac{\alpha \Delta_p}{a}\right)\right)+p_1^\ast\bar D_1 
+ \left(p_2^\ast-\alpha\Delta_p\right)
\left(D_2^\ast -\frac{\alpha \Delta_p}{a}\right)\\
=\frac{\alpha^2\Delta_p^2}{a}.
\label{eq:profit_diff_2}
\end{equation}
This profit loss is smaller than that in case~(i), where $p_1(\alpha)$ is strictly increasing. Indeed, subtracting~\eqref{eq:profit_diff_1} from~\eqref{eq:profit_diff_2} yields
$$\frac{\alpha^2\Delta_p^2}{a} - \left(\frac{1}{a}\epsilon ^2
+\left(
\bar D_1
+\frac{2\alpha\Delta_p}{a}
\right)\epsilon 
+\frac{\alpha^2\Delta_p^2}{a}\right)=-\frac{1}{a}\epsilon ^2
-\left(
\bar D_1
+\frac{2\alpha\Delta_p}{a}
\right)\epsilon<0.$$
Therefore, the case in which $p_1(\alpha)$ is strictly increasing is dominated and can be safely ignored.

We next characterize $\eta(\alpha)$ and $\nu(\alpha)$. Substituting $D_2(\alpha)$ into the expression for $\eta$ yields
\begin{equation*}
\eta(\alpha)= \frac{\pi - \big(2a(\bar D_2 - \tfrac{(1 - \alpha)\Delta_p}{a}) + b - a\bar D_2\big)}{a}= \frac{\pi - b - a\bar D_2}{a} + \frac{2(1 - \alpha)\Delta_p}{a}.
\end{equation*}
Finally, substituting $\eta(\alpha)$ into \eqref{eq:first_order}, gives
\begin{equation*}
\nu_1(\alpha)= \pi - (b + a\bar D_1) + a\,\eta(\alpha) = 2\pi - 2b - a(\bar D_1 + \bar D_2) + 2(1 - \alpha)\Delta_p.
\end{equation*}

Since $D_1=\bar{D}_1$, $\nu_1(0) > 0$. Both dual variables decrease linearly in $\alpha$ because
\begin{equation*}
\eta'(\alpha) = -\frac{2\Delta_p}{a} < 0\quad\text{and}
\quad
\nu_1'(\alpha) = -2\Delta_p < 0.
\end{equation*}
Thus, $D_1(\alpha)$ is binding to $\bar{D}_1$ for small $\alpha$. The associated dual variable 
$\nu_1(\alpha)$ declines until it reaches zero at
\begin{equation*}
{\alpha}_4
= 1 - \frac{2b + a(\bar D_1 + \bar D_2) - 2\pi}{2\Delta_p}.
\end{equation*}
The if-and-only-if condition of ${\alpha}_4<1$ is
$$\frac{\pi - b}{2a} < \frac{\bar D_1 + \bar D_2}{4},$$
which is feasible under~\eqref{eq:cond_1b}.

We also need to consider the point $D_1(\alpha) + D_2(\alpha) = D_{\mathrm{s}}$. Let 
$$
\alpha_5 := \inf\{\alpha \mid D_1(\alpha) + D_2(\alpha) = D_{\mathrm{s}}\}.
$$
Then, $\alpha_5$ satisfies
\begin{equation*}
\begin{aligned}
D_1(\alpha_5) + D_2(\alpha_5) = D_{\mathrm{s}} \iff&
\bar{D}_1 + \frac{p_2(\alpha) - b}{a} + \bar{D}_2 = D_{\mathrm{s}} \iff \bar{D}_1 + \frac{p_2^\ast + \alpha_5 \Delta_p - b}{a} + \bar{D}_2 = D_{\mathrm{s}} \\[3pt]
\iff& \bar{D}_1 + D_2^\ast + \frac{\alpha_5 \Delta_p}{a} = D_{\mathrm{s}} \iff \bar{D}_1 + \bar{D}_2 + \frac{\pi - b}{2a} + \frac{\alpha_5 \Delta_p}{a} = D_{\mathrm{s}}.
\end{aligned}
\end{equation*}
Hence,
$$
\alpha_5 = \frac{a(D_{\mathrm{s}} - \bar{D}_1 - \bar{D}_2) - \tfrac{\pi - b}{2}}{\Delta_p}.
$$

The if-and-only-if condition for $\alpha_5 > 0$ is
$$
D_{\mathrm{s}} > \bar{D}_1 + \bar{D}_2 + \frac{\pi - b}{2a}
= \frac{\bar{D}_1}{2} + \bar{D}_2 + \frac{\bar{D}_1}{2} + \frac{\pi - b}{2a}
\ge \frac{\bar{D}_1}{2} + \bar{D}_2 + \bar{D}_1,
$$
where the last inequality follows from $\tfrac{\pi - b}{2a} + \tfrac{\bar{D}_1}{2} \ge \bar{D}_1$, which is a necessary condition for $D_1^\ast = \bar{D}_1$. However, this contradicts the assumption $D_{\mathrm{s}} < \bar{D}_1 + \bar{D}_2$. Therefore, we can rule out the case where the aggregated energy constraint becomes binding.

For $\alpha \le \alpha_{4}$, we have $D_1(\alpha)=\bar{D}_1$, and only $D_2(\alpha)$ increases. Consequently, $U_1$ remains constant while $U_2$ increases, implying that the total consumer utility $U$ rises and, as a result, $W_{\mathrm{CNW}}$ also increases.

Regarding social welfare, we obtain
$$W'_{\mathrm{SW}}(\alpha)
=\sum_{i\in[2]} D'_i(\alpha)\left(\pi-b + a(\bar{D}_i - D_i(\alpha))\right)
= aD_2'(\alpha)\left(\bar{D}_2- D_2(\alpha)\right)> 0,$$ 
where the inequality holds because $D_2'(\alpha) > 0$ and $D_2(\alpha) \le D_2(\alpha_4) <\bar{D}_2.$ Note that $D_2(\alpha)$ cannot reach $\bar{D}_2$, because doing so would imply $D_1^\ast + D_2^\ast = \bar{D}_1 + \bar{D}_2$, which exceeds the system capacity $D_{\mathrm{s}}$ under the assumption $\bar{D}_1 + \bar{D}_2 > D_{\mathrm{s}}$. In summary,
\[
\boxed{
U_1~\text{remains constant},\quad
U_2~\text{increases},\quad
U~\text{increases},~W_{\mathrm{CNW}}~\text{increases}~\text{and}~
W_{\mathrm{SW}}~\text{increases}.\quad(\text{Regime}~$1$)
}
\]
For $\alpha>\alpha_4$ the upper bound on consumer~$1$ releases and the path coincides with Regime~$2$.


\textbf{\boldmath Case $\lambda > 0$.} Suppose $D_1^\ast + D_2^\ast = D_{\mathrm{s}}$, $D_1^\ast = \bar{D}_1$, and $D_2^\ast = D_{\mathrm{s}} - \bar{D}_1$. 
This case is similar to the one where $D_1^\ast = \bar{D}_1$ in Case~$\lambda=0$. 
However, the aggregated energy constraint is now binding, i.e., $D_1^\ast + D_2^\ast = D_{\mathrm{s}}$. 
In this situation, $D_1$ cannot remain at the same value, because doing so would require $D_2^\ast$ to increase in order to satisfy the fairness constraint, which would violate the aggregated energy constraint. Therefore, this case coincides with Regime~$2$ with Case~$\lambda=0$, as the total provided energy $D_1+D_2$ remains constant.\hfill \Halmos
\end{proof}

\medskip

\begin{proof}{Proof of Theorem~\ref{thm:utility_fairness}.}
If $p_i \in [b-a\bar D_i,b]$, then utility can be expressed solely as a function of $D_i$. In this case,
\begin{equation}
U_i
= p_i D_i - C(D_i)
= (aD_i + b - a\bar D_i)D_i
- \left(\frac{1}{2}aD_i^2 + (b-a\bar D_i)D_i\right)
= \frac{1}{2}aD_i^2.
\label{eq:utility_linear_response}
\end{equation}
By Lemma~\ref{lemma:equivalent_model}, the optimization problem can be solved with respect to $D_i^\ast$, and the corresponding price $p_i$ can be recovered using \eqref{linear_response}. If $D_i^\ast=0$, the utility is still given by \eqref{eq:utility_linear_response}, since $U_i = p_i\cdot0 - C(0) = 0$. When $D_i^\ast=\bar D_i$, the profit-only optimal price is $p_i=b$, and the resulting utility is again given by \eqref{eq:utility_linear_response}. Therefore, in all cases at $\alpha=0$, utility admits the representation \eqref{eq:utility_linear_response}.

Since $D_2^\ast > D_1^\ast$, as established in the proof of Proposition~\ref{thm:energy_fairness}, $U_2^\ast=\tfrac{1}{2}a{D_2^\ast}^2 > \tfrac{1}{2}a{D_1^\ast}^2 = U_1^\ast$. 
Accordingly, the utility fairness constraint for any $\alpha$ can be written as
$$|U_2-U_1| \le (1 - \alpha) \Delta_U,$$
where $\Delta_U:=U_2^\ast-U_1^\ast=\frac{1}{2}a\left({D_2^\ast}^2-{D_1^\ast}^2\right)$.

In the following proof, we analyze each regime defined by the relevant boundary conditions. For example, Regimes~$3$ corresponds to the cases in which $D_1^\ast=\bar D_1$. Within each regime (i.e., as long as the set of binding constraints remains unchanged), the optimal solution varies continuously with $\alpha$. This follows from the continuity of the KKT system with respect to $\alpha$ when the active constraint set is fixed. Continuity of $D_i(\alpha)$ implies continuity of utilities $U_i(\alpha)$ for $i\in\{1,2\}$. Define the utility gap $\Delta_U(\alpha):=U_2(\alpha)-U_1(\alpha)$. At $\alpha=0$, the optimal solution satisfies $\Delta_U(0)>0$, while by construction $\Delta_U(1)=0$. Since $\Delta_U(\alpha)$ is continuous on $[0,1]$, the intermediate value theorem implies that $\Delta_U(\alpha)$ cannot change sign without crossing zero. Therefore, the utility ordering $U_2(\alpha)\ge U_1(\alpha)$ is preserved for all $\alpha\in[0,1]$.

We reuse $\lambda \ge 0$ denote the multiplier for the aggregated energy constraint $D_1 + D_2 \le D_{\mathrm s}$, $\mu_i \ge 0$ for the lower bound $-D_i \le 0$, and $\nu_i \ge 0$ for the upper bound $D_i - \bar D_i \le 0$, for $i \in [2]$. 
Let $\eta \ge 0$ be the multiplier for the utility fairness constraint.

\noindent\textbf{\boldmath Regimes~$3$ and $4$: $D_1^\ast = \bar D_1$.}\\
As in Lemma~\ref{lemma:opt_sol}, we analyze two cases: slack $(\lambda = 0)$ and binding $(\lambda > 0)$.

\textbf{\boldmath Case $\lambda = 0$.} When $D_1^\ast = \bar D_1$, there are two conceivable adjustment paths as $\alpha$ increases:  
(1) both $D_1$ and $D_2$ strictly decrease, or  
(2) $D_1$ remains fixed at its upper bound $\bar D_1$ while $D_2$ changes.

However, path~(1) is impossible.  
For the sake of contradiction, suppose that both $D_1$ and $D_2$ decrease as $\alpha$ increases. More specifically, let $D_1(\alpha) = \bar D_1 - \Delta D_1$ and $D_2(\alpha) = D_2^\ast - \Delta D_2$, where small $\Delta D_i > 0$ for $i\in[2]$. They satisfy
\begin{equation*}
\frac{1}{2} a D_2(\alpha)^2 - \frac{1}{2} a D_1(\alpha)^2
\le (1 - \alpha)\Delta_U.
\end{equation*}
Now consider the alternative feasible solution $(\bar D_1, D_2(\alpha))$. It satisfies the utility fairness constraint because
\begin{equation*}
\frac{1}{2} a D_2(\alpha)^2-\frac{1}{2} a \bar D_1^2
< \frac{1}{2} a D_2(\alpha)^2 - \frac{1}{2} a D_1(\alpha)^2
\le (1 - \alpha)\Delta_U,
\end{equation*}
where the first inequality holds since $D_1(\alpha) < \bar D_1$. Moreover, for fixed $D_2(\alpha)$, the profit $\Pi(\cdot, D_2(\alpha))$ is strictly concave in $D_1$ and the unconstrained maximizer $D_1^\ast = \tfrac{\pi - b}{2a} + \tfrac{\bar D_1}{2}$ is unique. Since $D_1(\alpha) < \bar D_1 \le D_1^\ast $, increasing $D_1$ from $D_1(\alpha)$ to $\bar D_1$ strictly increases profit. Thus, $\Pi(\bar D_1, D_2(\alpha)) > \Pi(D_1(\alpha), D_2(\alpha))$.
This yields a contradiction, therefore, choosing $D_1(\alpha) = \bar D_1 - \Delta D_1 < \bar D_1$ cannot be optimal.

Now, in path (2), let $D_2(\alpha) = D_2^\ast + \Delta D_2$ and $p_2^\ast=aD_2^\ast+b-a\bar{D}_2$. Define $p_1(\alpha) = b + \Delta p_1$, where $\Delta p_1\ge0$ because $D_1(\alpha)=\bar{D}_1$ is fixed at its upper bound. 
The difference between the aggregator profit without the utility fairness constraint and the profit under $\alpha$–utility fairness can therefore be written as
\begin{equation}
\begin{aligned}
\Delta \Pi 
&= \Pi(0) - \Pi(\alpha) \\
&= \Bigl[\pi\bigl(\bar D_1 + D_2^\ast\bigr) - b \bar D_1 - p_2^\ast D_2^\ast\Bigr] - \Bigl[\pi\bigl(\bar D_1 + D_2^\ast + \Delta D_2\bigr) - (b + \Delta p_1)\bar D_1 - (p_2^\ast + a \Delta D_2)\bigl(D_2^\ast + \Delta D_2\bigr)\Bigr]\\
&=a (\Delta D_2)^2
+ \bigl(a D_2^\ast + p_2^\ast - \pi\bigr)\Delta D_2
+ \bar D_1 \Delta p_1\\
&=a (\Delta D_2)^2
+ \bigl(2a D_2^\ast + b - a\bar D_2 - \pi\bigr)\Delta D_2
+ \bar D_1 \Delta p_1.
\end{aligned}
\label{eq:profit_loss_utility}
\end{equation}

The utility function is defined as $U_i := p_i D_i - C_i(D_i)$, which simplifies to $\frac{1}{2} a D_i^2$ when $p_i \in [b - a\bar{D}_i, b]$. Assuming $\Delta D_2\le\bar{D}_2-D_2^\ast$, the difference in total consumer utility is 
\begin{equation*}
\begin{aligned}
U_2(\alpha)-U_1(\alpha)
&= \frac{1}{2}a\left(D_2^\ast+\Delta D_2\right)^2
   - \bigl[(p_1+\Delta p_1)\bar{D}_1 - C(\bar{D}_1)\bigr] \\
&= \Bigl(\tfrac{1}{2}a(D_2^\ast)^2 - p_1 \bar D_1 + C(\bar D_1)\Bigr)
   + a D_2^\ast \Delta D_2
   + \frac{a}{2}(\Delta D_2)^2
   - \bar D_1 \Delta p_1 \\
&= (U_2^\ast - U_1^\ast)
   + a D_2^\ast \Delta D_2
   + \frac{a}{2}(\Delta D_2)^2
   - \bar D_1 \Delta p_1 \\
&= \Delta_U
   + a D_2^\ast \Delta D_2
   + \frac{a}{2}(\Delta D_2)^2
   - \bar D_1 \Delta p_1.
\end{aligned}
\end{equation*}
The utility fairness constraint is binding. Suppose, for contradiction, that the optimal solution $(\Delta p_1^\ast, \Delta D_2^\ast)$ does not bind this constraint. Then, by slightly decreasing $\Delta p_1^\ast$, the utilities can still satisfy the utility fairness constraint, while strictly decreasing the aggregator's profit loss~\eqref{eq:profit_loss_utility}. This contradicts the optimality of the slack solution, implying that the utility fairness constraint must bind.
Therefore, the fairness constraint can be written as
$$a D_2^\ast \Delta D_2+ \frac{a}{2}(\Delta D_2)^2- \bar D_1 \Delta p_1+\alpha\Delta_U= 0.$$

Since we want to maximize the profit, we equivalently minimize the profit loss $\Delta \Pi$. Thus, we consider
\begin{equation}
\begin{aligned}
\min_{\Delta p_1,\, \Delta D_2}\quad 
a(\Delta D_2)^2 + \bigl(2a D_2^\ast + b - a\bar D_2 - \pi\bigr)\Delta D_2 + \bar D_1 \Delta p_1 \\
\text{s.t.}\quad 
a D_2^\ast \Delta D_2 + \frac{a}{2}(\Delta D_2)^2 - \bar D_1 \Delta p_1 + \alpha\Delta_U = 0, \\
\Delta D_2 \in\left[- D_2^\ast,\bar D_2 - D_2^\ast\right], \\
D_2^\ast + \Delta D_2 + \bar D_1 \le D_{\mathrm s}, \\
\Delta p_1 \ge 0,
\end{aligned}
\label{eq:u1}
\end{equation}
where the second box constraint follows from the box constraint $D_2 \in [0,\bar D_2]$, the third constraint corresponds to the aggregated energy constraint, and the last constraint reflects the case in which $D_1$ is fixed at its upper bound $\bar D_1$, so that $p_1$ can only stay the same or increase.

The problem~\eqref{eq:u1} can be reformulated by eliminating $\Delta p_1$ based on the first constraint. 
\begin{subequations}
\begin{align}
\min_{\Delta D_2}\quad 
& \frac{3a}{2}(\Delta D_2)^2 
  + \bigl(3a D_2^\ast + b - a\bar D_2 - \pi\bigr)\Delta D_2
  + \alpha\,\Delta_U \nonumber\\
\text{s.t.}\quad
& a D_2^\ast \Delta D_2 
  + \frac{a}{2}(\Delta D_2)^2 
  + \alpha\Delta_U \ge 0, \label{subeq:utility_fairness}\\
& \Delta D_2 \in\left[- D_2^\ast,\bar D_2 - D_2^\ast\right], \label{subeq:D2_1}\\
& D_2^\ast + \Delta D_2 + \bar D_1 \le D_{\mathrm s}.\label{subeq:D2_2}
\end{align}
\label{eq:u1_reduced}
\end{subequations}
The objective function is strictly convex, and its unconstrained minimizer is
\begin{equation*}
\Delta D_2^{\mathrm{uc}}
= -\frac{3a D_2^\ast + b - a\bar D_2 - \pi}{3a}
= \frac{\pi - b + a\bar D_2}{3a} - D_2^\ast,
\end{equation*}
which satisfies \eqref{subeq:D2_1} and \eqref{subeq:D2_2}.

Let $r^{-}(\alpha)$ and $r^{+}(\alpha)$ denote the roots of the fairness constraint~\eqref{subeq:utility_fairness},
\begin{equation*}
r^{-}(\alpha)
= -D_2^\ast -
\sqrt{(D_2^\ast)^2 - \frac{2\alpha\Delta_U}{a}}
\quad\text{and}\quad
r^{+}(\alpha)
= -D_2^\ast +
\sqrt{(D_2^\ast)^2 - \frac{2\alpha\Delta_U}{a}},
\end{equation*}
Therefore, for the unconstrained minimizer $\Delta D_2^{\mathrm{uc}}$ to be feasible, for $\alpha \in [0,1]$, it must satisfy either (i) $(D_2^\ast)^2 - \frac{2\alpha \Delta_U}{a} \le 0$ or (ii) $\Delta D_2^{\mathrm{uc}} \le r^{-}(\alpha)$ or $\Delta D_2^{\mathrm{uc}} \ge r^{+}(\alpha)$. 
Condition (i) holds when
$$\alpha \ge \frac{a(D_2^\ast)^2}{2\Delta_U}=\frac{(D_2^\ast)^2}{(D_2^\ast)^2-\bar D_1^2}>1.$$
Since $\alpha \in [0,1]$, this condition cannot be satisfied and is therefore not observable.
The former case in (ii) is impossible, as it would imply $\Delta D_2^{\mathrm{uc}} < r^{-}(\alpha) \le -D_2^\ast$, which violates~\eqref{subeq:D2_1}.
Thus, the optimal solution can be written as
\begin{equation*}
\Delta D_2^\ast(\alpha)=
\begin{cases}
r^+(\alpha) &\text{if}\quad \Delta D_{2}^{uc}  < r^+(\alpha),\\
\Delta D_{2}^{uc} &\text{if}\quad\Delta D_{2}^{uc} \in [r^+(\alpha),\bar D_2-D_2^\ast]\\
\bar D_2-D_2^\ast & \text{otherwise.}
\end{cases}
\end{equation*}
The upper bound $\bar D_2 - D_2^\ast$ follows from the upper bound on $\Delta D_2$ in~\eqref{eq:u1_reduced}. More specifically, we consider
$$\min\left\{\bar D_2 - D_2^\ast,D_{\mathrm s} - D_2^\ast - \bar D_1\right\}=\bar D_2 - D_2^\ast,$$
where the equality follows from the condition $\bar D_1 + \bar D_2 > D_{\mathrm s}$.

We now analyze the conditions $\Delta D_2^{\mathrm{uc}} \ge r^{+}(\alpha)$ and $\Delta D_2^{\mathrm{uc}} \le \bar D_2 - D_2^\ast$. These conditions hold if and only if
\begin{equation}
\begin{aligned}
\Delta D_2^{\mathrm{uc}} \ge r^{+}(\alpha) &\iff \alpha\ge\frac{a}{2\Delta_U}
\left[
(D_2^\ast)^2
- \left(\frac{a\bar D_2 + \pi - b}{3a}\right)^2
\right]=: \alpha_1,\\
\Delta D_2^{\mathrm{uc}} \le \bar D_2 - D_2^\ast &\iff \bar D_2 \ge \frac{\pi-b}{2a}.
\end{aligned}
\label{eq:threshold_R3_to_R4}
\end{equation}
Note that, when $\lambda=0$, the last inequality always holds because $D_2^\ast=\frac{\pi-b}{2a}+\frac{\bar D_2}{2}\le \bar D_2$ implies 
$\bar D_2\ge \frac{\pi-b}{2a}$.
Then, we can write the optimal solution $D_2^\ast(\alpha)$ as follows.
\begin{equation}
D_2^\ast(\alpha)=D_2^\ast+\Delta D_2^\ast(\alpha)=
\begin{cases}
D_2^\ast + r^+(\alpha)=\sqrt{(D_2^\ast)^2 - \frac{2\alpha\Delta_U}{a}} &\text{if}\quad \alpha< \alpha_1\\[5pt]
D_2^\ast+\Delta D_{2}^{uc}=\frac{\pi-b+a\bar D_2}{3a}  &\text{otherwise.}
\end{cases}
\label{eq:opt_D_for_R3R4}
\end{equation}
We can also derive $p_1^\ast(\alpha)$ by getting $\Delta p_1$ from the equality constraint in~\eqref{eq:u1}.
\begin{equation*}
p_1^\ast(\alpha)=p_1^\ast+\Delta p_1^\ast(\alpha)=
\begin{cases}
p_1^\ast &\text{if}\quad \alpha < \alpha_1\\[5pt]
p_1^\ast + \Delta \tilde p (\alpha)  &\text{otherwise,}
\end{cases}
\end{equation*}
where
$$\Delta \tilde p (\alpha)=\frac{
aD_2^\ast \Delta D_2^{\mathrm{uc}}
+\frac{a}{2}\big(\Delta D_2^{\mathrm{uc}}\big)^2
+\alpha\Delta_U
}{\bar D_1}> 0\quad \text{for all}~\alpha>\alpha_1.$$
Therefore, if $\alpha < \alpha_1$, only $D_2$ adjusts, whereas if $\alpha \ge \alpha_1$, only $p_1$ adjusts.

We now verify whether the transition can occur. The condition $\alpha_1 > 0$ holds if and only if
\begin{equation}
    3aD_2^\ast > a\bar D_2+\pi-b.
    \label{eq:alpha_1_greater_then_0_utility}
\end{equation}
This condition always holds because $2aD_2^\ast = a\bar D_2 + \pi - b > 0$. The condition $\alpha_1\le1$ holds if and only if
\begin{equation}
    a\bar D_2+\pi-b\ge 3a\bar D_1.
    \label{eq:alpha1_less_1}
\end{equation}
When $\lambda=0$, this condition simplifies to 
$$a\bar D_2+\pi-b\ge 3a\bar D_1 \iff  \frac{2}{3}D_2^\ast\ge\bar{D}_1.$$
This always holds because
$$\frac{2}{3}D_2^\ast = \frac{2}{3}\left(\frac{\pi-b}{2a}+\frac{\bar D_2}{2}\right)=\frac{2}{3}\frac{\pi-b}{2a} + \frac{\bar D_2}{3}\ge \frac{2}{3}\frac{\pi-b}{2a} + \frac{\bar D_1}{3}\ge\frac{2}{3}\bar D_1+\frac{\bar D_1}{3}=\bar D_1,$$
where the last inequality follows $\frac{\pi-b}{2a}\ge\frac{\bar D_1}{2}$. Therefore, Regime~$3$ cannot be terminate regime when $\lambda=0$.


Overall, when $\alpha\le\alpha_1$, only $D_2$ decreases. This implies that $U_2$ decreases. In contrast, $D_1$ remains constant, and hence $U_1$ also remains constant. Consequently, both the total consumer utility ($U$) and the CNW ($U_{\mathrm{CNW}}$) decrease. Social welfare decreases since $U$ and $\Pi$ decrease. In summary,
\[
\boxed{
U_1~\text{remains constant},\quad
U_2~\text{decreases},\quad
U~\text{decreases},\quad
W_{\mathrm{CNW}}~\text{decreases},~\text{and}~ W_{\mathrm{SW}}~\text{decreases}. \qquad (\text{Regime~$3$})
}
\]

When $\alpha>\alpha_1$, only $p_1$ increases, implying that $U_2$ remains constant while $U_1$ increases. Consequently, both $U$ and $W_{\mathrm{CNW}}$ increase. Regarding social welfare,
$$
W'_{\mathrm{SW}}(\alpha)
=\sum_{i\in[2]} D'_i(\alpha)\left(\pi-b + a(\bar{D}_i - D_i(\alpha))\right)
= 0,
$$
because $D_i(\alpha)$ remains constant for all $i$, and thus no change in provided energy occurs even though $p_1$ varies.
Consequently,
\[
\boxed{
U_1~\text{increases},~
U_2~\text{remains constant},~
U~\text{increases},~
W_{\mathrm{CNW}}~\text{increases},~\text{and}~W_{\mathrm{SW}}~\text{remains constant}. \quad (\text{Regime~$4$})
}
\]
When $\lambda=0$, Regime~$4$ cannot be the initial regime because~\eqref{eq:alpha_1_greater_then_0_utility} always holds.

\textbf{\boldmath Case $\lambda > 0$.} 
When $\lambda>0$, the logic is very similar to the case $\lambda=0$. However, since $D_2^\ast$ differs, we need to verify whether~\eqref{eq:threshold_R3_to_R4} holds. In this case, the condition $\bar D_2 \ge \frac{\pi-b}{2a}$ does not always hold. In particular, when $\bar D_2 < \frac{\pi-b}{2a}$ and $D_{\mathrm{s}} > \frac{3\bar D_1 + \bar D_2}{2}$, we obtain the case in which $D_1^\ast = \bar D_1$ and $D_2^\ast = D_{\mathrm{s}} - \bar D_1$. Therefore, under $\lambda>0$, Regime~$4$ can arise as the initial regime.

On the other hand, Regime~$3$ cannot be the final regime, as in the case $\lambda = 0$. Suppose, to the contrary, that Regime~$3$ is the final regime, then $\alpha_1 \ge 1$, which is equivalent to
\begin{equation}
    \frac{3}{2}\bar D_1 \ge \frac{\pi-b}{2a} + \frac{\bar D_2}{2}
    \label{eq:cont1}
\end{equation}
by~\eqref{eq:alpha1_less_1}.
For $\lambda>0$ to hold, a necessary condition is
$$
\min\left(
\bar D_1 + \frac{\pi-b}{2a} + \frac{\bar D_2}{2},
\frac{\pi-b}{a} + \frac{\bar D_1 + \bar D_2}{2}
\right) \ge D_{\mathrm s}
$$
Combining this with~\eqref{eq:cont1} yields
\begin{equation}
\min\left(
\frac{5}{2}\bar D_1,
2\bar D_1 + \frac{\pi-b}{2a}
\right) \ge D_{\mathrm s}.
\label{eq:contradiction1_for_R3}
\end{equation}
Under Regime~$3$, we have $D_1^\ast = \bar D_1$.  
The necessary condition is
\begin{equation}
    \frac{D_{\mathrm s}}{2} + \frac{\bar D_1 - \bar D_2}{4} \ge \bar D_1
    \label{eq:D1_barD1_positive_lambda}
\end{equation}
Substituting this lower bound of $D_{\mathrm s}$ into the previous inequality~\eqref{eq:contradiction1_for_R3} implies
$$
\min\left(
\frac{3}{2}\bar D_1 - \frac{1}{4}\bar D_2,
\frac{5\bar D_1 - \bar D_2}{4} + \frac{\pi-b}{2a}
\right) \ge \bar D_1
$$
which immediately implies
\begin{equation}
    2\bar D_1 \ge \bar D_2
    \label{eq:contradiction2_for_R3}
\end{equation}
However, the standing assumption $\bar D_1 + \bar D_2 > D_{\mathrm s}$ together with~\eqref{eq:D1_barD1_positive_lambda} implies
$$
\bar D_1 + \bar D_2 > D_{\mathrm s}
\ge \frac{3}{2}\bar D_1 + \frac{1}{2}\bar D_2
\implies \bar D_2 > 2\bar D_1
$$
which contradicts~\eqref{eq:contradiction2_for_R3}.  
Therefore, Regime~$3$ cannot be the final regime even when $\lambda>0$.

In summary, for both the cases $\lambda = 0$ and $\lambda > 0$, when $D_1^\ast = \bar D_1$, there are two possible regimes, Regime~$3$ and Regime~$4$. We show that both Regime $3$ and Regime $4$ can serve as initial regimes, and that a transition from Regime $3$ to Regime $4$ must necessarily occur. In other words, Regime $3$ cannot be a terminal regime.

In the following analysis, we consider the case in which $D_1^\ast < \bar D_1$. If, under this case, the dynamics evolve such that $D_1(\alpha)$ reaches the upper bound $\bar D_1$, the system then possibly enters one of the regimes described above. Note that the expressions derived above are formulated in terms of the initial optimal solution $(D_1^\ast, D_2^\ast)$. If a regime arises following a transition from another regime, the corresponding quantities are instead expressed as $(D_1(\hat\alpha), D_2(\hat\alpha))$, where $\hat\alpha$ denotes the value at which the regime transition occurs.

\noindent\textbf{\boldmath Regimes~$1$ and $2$: $D_1^\ast<\bar{D}_1$.}\\
As in the previous case, we consider two cases: $\lambda = 0$ (slack) and $\lambda > 0$ (binding).

\textbf{\boldmath Case $\lambda = 0$.} Suppose $D_1^\ast + D_2^\ast < D_\mathrm{s}$ and $0 < D_i^\ast < \bar D_i$ for all $i \in [2]$. Note that $D_i^\ast>0$, when $D_1^\ast+D_2^\ast<D_{\mathrm{s}}$. Consider $\alpha$ such that $(D_1(\alpha), D_2(\alpha))$ is the interior of all constraints, i.e., $\alpha<\tilde{\alpha}$ where
$$\tilde{\alpha}:=\inf\left\{\alpha|D_1(\alpha)+D_2(\alpha)=D_\mathrm{s}~\text{or~} D_i(\alpha)\in\{0,\bar{D}_i\},~\forall i\in[2]\right\}.$$

The utility fairness constraint can be simplified as follows.
\begin{equation*}
    \begin{aligned}
    U_2-U_1\le(1-\alpha)\left(U_2^\ast-U_1^\ast\right)&\iff \frac{1}{2}aD_2^2-\frac{1}{2}aD_1^2\le(1-\alpha)\left(\frac{1}{2}a(D_2^\ast)^2-\frac{1}{2}a(D_1^\ast)^2\right)\\
    &\iff D_2^2-D_1^2 \le (1 - \alpha) \left({D_2^\ast}^2-{D_1^\ast}^2\right).
    \end{aligned}
\end{equation*}

Then, the Lagrangian for $\alpha<\tilde \alpha$ is
\begin{equation*}
\begin{aligned}
\mathcal{L} = a D_1^2 + (b-a\bar{D}_1 - \pi) D_1 + a D_2^2 + (b-a\bar{D}_2 - \pi) D_2+\eta\left(D_2^2 - D_1^2 - (1-\alpha) \Delta_{D^2}\right),
\end{aligned}
\end{equation*}
where $\Delta_{D^2}:={D_2^\ast}^2-{D_1^\ast}^2$.
The stationarity conditions are:
\begin{equation}
\begin{aligned}
2a (1-\eta) D_1 + (b-a\bar{D}_1 - \pi) &= 0,\\
2a (1+\eta) D_2 + (b-a\bar{D}_2 - \pi) &= 0. 
\end{aligned}
\label{eq:stationary_utility}
\end{equation}
Using $D_i^{*} = \frac{\pi - b + a\bar{D}_i}{2a}$, we can rewrite \eqref{eq:stationary_utility} as
\begin{equation}
D_1(\alpha) = \frac{D_1^{\ast}}{1 - \eta(\alpha)}\quad\text{and}\quad
D_2(\alpha) = \frac{D_2^{\ast}}{1 + \eta(\alpha)},
\label{eq:D_alpha}
\end{equation}
where $\eta(\alpha)$ is determined by the utility fairness constraint
\begin{equation}
\left( \frac{D_2^{\ast}}{1+\eta(\alpha)} \right)^2
 - \left( \frac{D_1^{\ast}}{1-\eta(\alpha)} \right)^2
= (1-\alpha)\,\Delta_{D^2}.
\label{eq:eta_eq}
\end{equation}
Additionally, note that $\eta(\alpha)<1$, which is implied by primal feasibility, $D_1(\alpha)=\frac{D_1^\ast}{1-\eta(\alpha)}>0$.

Differentiating \eqref{eq:eta_eq} with respect to $\alpha$ yields
\begin{equation}
\eta'(\alpha)=\frac{\Delta_{D^2}}{
2\left[
\dfrac{{D_2^\ast}^2}{(1+\eta(\alpha))^3}
+ \dfrac{{D_1^\ast}^2}{(1-\eta(\alpha))^3}
\right]}>0,
\label{eq:increasing_eta}
\end{equation}
which implies that $\eta(\alpha)$ is an increasing function of $\alpha$.

Note that the sign of $U_i'(\alpha)$ is determined by that of $D_i'(\alpha)$, since 
$U_i'(\alpha)=a\,D_i(\alpha)D_i'(\alpha)$. 
Since $\eta(\alpha)$ is increasing in $\alpha$ by \eqref{eq:increasing_eta}, and \eqref{eq:D_alpha} implies that $D_1(\alpha)$ is increasing while $D_2(\alpha)$ is decreasing in $\alpha$, we conclude that $U_1'(\alpha) > 0$ and $U_2'(\alpha) < 0$ for all $\alpha$.

We now clarify the definition of $\tilde{\alpha}$. It is the minimum value of $\alpha$ at which the interior solution first hits a boundary. Here, we can exclude the cases (i) $D_1(\alpha)=0$, (ii) $D_2(\alpha)=\bar D_2$, and (iii) $D_2(\alpha)=0$. The multiplier $\eta(\alpha)$ is strictly increasing in $\alpha$ (see~\eqref{eq:increasing_eta}) and remains finite. 
$0 \le \eta(\alpha) < \eta(1) < \infty$.
Therefore, $D_1(\alpha)$ is increasing based on~\eqref{eq:D_alpha}, which means (i) $0\le D_1^\ast<D_1(\alpha)$. Similarly, $D_2(\alpha)$ is decreasing, which implies (ii) $D_2(\alpha)<D_2^\ast<\bar{D}_2$. Lastly, $D_2(\alpha)$ does not hits $0$ since the minimum value is achieved when $\alpha=1$, and it should satisfy $D_2(\alpha)>D_2(1)=D_1(1)>0$.

At the boundary $D_2(\alpha)=\bar D_2$, the system may lie in either Regime~$3$ or Regime~$4$, as shown in the analysis of Regimes~$3$ and $4$. Furthermore, since the dynamics of Regime~$4$, characterized by $p_1(\alpha)\ge b$, cannot be fully explained by the provided energy dynamics in \eqref{eq:D_alpha}. Therefore, this case can also trigger a transition. Therefore, $\tilde{\alpha}$ can be expressed as
\begin{equation*}
\begin{aligned}
\tilde\alpha &:= \min\left(\alpha_2,\alpha_3,\alpha_4\right),~\text{where}\\
\alpha_2&:=\inf\left\{\alpha \mid D_1(\alpha) + D_2(\alpha) = D_{\mathrm s}\right\}\\
\alpha_3&:=\inf\left\{\alpha \,\Big|\, D_1(\alpha) = \bar D_1, D_2(\alpha)>\frac{a\bar D_2+\pi-b}{3a}\right\},\\
\alpha_4&:=\inf\left\{\alpha | \Pi^{(4)}(\alpha)\ge \Pi^{(1,2)}(\alpha)\right\},
\end{aligned}
\end{equation*}
where $\Pi^{(4)}(\alpha)$ denotes the profit induced in Regime~$4$ with $D_1(\alpha)=\bar D_1$, and $\Pi^{(1,2)}(\alpha)$ denotes the profit induced in Regimes~$1$ or Regime~$2$, which is fully determined by the provided energy in \eqref{eq:D_alpha}. The condition defining $\alpha_3$ follows from the requirement that $\alpha_1>0$ in~\eqref{eq:alpha_1_greater_then_0_utility}.

We now examine $\alpha_2$, $\alpha_3$, and $\alpha_4$. For~$\alpha_2$, the case $D_1(\alpha) + D_2(\alpha) = D_{\mathrm{s}}$, which is corresponding to $\alpha_2$, the left-hand side varies with respect to $\alpha$ according to
\begin{equation*}
\frac{dD_1(\alpha)}{d\alpha} + \frac{dD_2(\alpha)}{d\alpha}
= \eta'(\alpha)\left[
\frac{D_1^{\ast}}{(1-\eta(\alpha))^2}
- \frac{D_2^{\ast}}{(1+\eta(\alpha))^2}
\right].
\end{equation*}
At $\alpha = 0$, the term in brackets is negative because $D_2^{\ast} > D_1^{\ast}$ and $\eta(\alpha) = 0$, implying that the total provided energy $D_1(\alpha) + D_2(\alpha)$ initially decreases as $\alpha$ increases. However, as $\eta(\alpha)$ increases, the first term grows while the second term diminishes. Therefore, their sum may increase beyond a certain threshold. Consequently, the total provided energy $D_1(\alpha) + D_2(\alpha)$ attains its maximum at one of the endpoints, that is, at $\alpha = 0$ or $\alpha = 1$. Note that when $\alpha = 0$, we have $D_1^\ast + D_2^\ast < D_{\mathrm{s}}$, since we are considering the case in which $\lambda = 0$.

When $\alpha=1$, using $D_1(1)^2=D_2(1)^2$, we can get $\eta(1)$,
\begin{equation}
\eta(1) = 
\frac{D_2^{\ast} - D_1^{\ast} }{D_2^{\ast} + D_1^{\ast}}=\frac{c-1}{c+1}<\infty,
\label{eq:eta_1}
\end{equation}
where 
\begin{equation}
    c:=\frac{D_2^{\ast}}{D_1^{\ast}}=\frac{\pi-b+a\bar{D}_2}{\pi-b+a\bar{D}_1}>1.
    \label{eq:definition_c}
\end{equation}

Substituting $\eta(1)=\frac{D_2^{\ast}-D_1^{\ast}}{D_2^{\ast}+D_1^{\ast}}$ into \eqref{eq:D_alpha} gives
\begin{equation*}
D_1(1)+D_2(1)
=\frac{D_1^{\ast}}{1-\eta(1)}+\frac{D_2^{\ast}}{1+\eta(1)}
=\frac{D_1^{\ast}}{2D_1^{\ast}/(D_1^{\ast}+D_2^{\ast})}
+\frac{D_2^{\ast}}{2D_2^{\ast}/(D_1^{\ast}+D_2^{\ast})}
=D_1^{\ast}+D_2^{\ast}<D_{\mathrm{s}}.
\end{equation*}
Therefore, no $\alpha \in [0,1]$ satisfies $D_1(\alpha)+D_2(\alpha)=D_{\mathrm{s}}$, so this case requires no further consideration.

For $\alpha_3$, we first derive the $\eta(\alpha_3)$ as follows.
\begin{equation}
\frac{D_1^\ast}{1 - \eta(\alpha_3)} = \bar D_1
\quad \iff \quad
\eta(\alpha_3) = 1 - \frac{D_1^\ast}{\bar D_1}.
\label{eq:eta_3}
\end{equation}
By~\eqref{eq:D_alpha}, the corresponding $D_2(\alpha_3)$ is
\begin{equation*}
D_2(\alpha_3)
= \frac{D_2^\ast}{1+\eta(\alpha_3)}
= \frac{D_2^\ast}{\,2 - \frac{D_1^\ast}{\bar D_1}\,}
= \frac{\bar D_1 D_2^\ast}{\,2\bar D_1 - D_1^\ast\,}.
\end{equation*}
Therefore, a necessary condition for Regime~$3$ to arise is
$$D_2(\alpha_3)>\frac{\pi-b+a\bar{D}_2}{3}=\frac{2}{3}D_2^\ast \iff D_1^\ast>\frac{\bar D_1}{2} \iff \frac{\pi-b}{2a}>0.$$
Moreover, by~\eqref{eq:eta_eq},
$$\alpha_3= 1 - \frac{D_2(\alpha_3)^2 - D_1(\alpha_3)^2}{\Delta_{D^2}}
= 1 - \frac{\left( \frac{\bar D_1 D_2^\ast}{\,2\bar D_1 - D_1^\ast\,} \right)^2
- \bar D_1^2}{\Delta_{D^2}}.
$$
If $\alpha_3<1$, then Regime~$3$ must occur. The if-and-only-if condition for $\alpha_3<1$ is
$$\eta(\alpha_3)<\eta(1) \iff 1-\frac{D_1^\ast}{\bar{D}_1}<\frac{c-1}{c+1} \iff \frac{D_1^\ast}{\bar{D}_1}>\frac{2}{c+1} \iff D_1^\ast+D_2^\ast>2\bar{D}_1.$$
When $\lambda = 0$, this condition holds if
$$D_1^\ast + D_2^\ast=\left(\frac{\pi-b}{2a}+\frac{\bar D_1}{2}\right)+\left(\frac{\pi-b}{2a}+\frac{\bar D_2}{2}\right)>2\bar D_1\iff \frac{\pi-b}{2a}>\frac{3\bar D_1-\bar D_2}{4}.$$
To conclude, the condition for Regime~$3$ to arise is,
$$ \max\left(0,\frac{3\bar D_1-\bar D_2}{4}\right)<\frac{\pi-b}{2a}< \frac{\bar D_1}{2},$$
where the last inequality follows from $D_1^\ast <\bar{D}_1$. There exists a parameter value for which this interval is feasible
(e.g., $(a,b,\pi,\bar D_1,\bar D_2,D_{\mathrm s}) = (1,9,9.4,1.2,3.5,4.5)$)
or infeasible
(e.g., $(a,b,\pi,\bar D_1,\bar D_2,D_{\mathrm s}) = (1,9,9.4,1.5,3.5,4.8)$).
Thus, Regime~$3$ may or may not arise.

We next analyze $\alpha_4$. We first derive the profit expressions $\Pi^{(4)}$ and $\Pi^{(1,2)}$ and then compare them.  
In Regime~$4$, provided energy of consumer~$1$ is saturated so that $D_1(\alpha)=\bar D_1$, while $D_2(\alpha)$ remains interior and equals $D_2(\alpha_4)$.  
Therefore, for all $\alpha \ge \alpha_4$,
\begin{equation*}
U_2(\alpha)=U_2(\alpha_4)=\tfrac12 a D_2(\alpha_4)^2,
\qquad
U_1(\alpha)=p_1(\alpha)\bar D_1 - C(\bar D_1).
\end{equation*}
Since the utility fairness constraint binds,
$U_2(\alpha)-U_1(\alpha)=(1-\alpha)\Delta_U,$
which implies
$$p_1(\alpha)\bar D_1
=U_2(\alpha_4)+C(\bar D_1)-(1-\alpha)\Delta_U.$$
The aggregator profit in Regime~$4$ is given by
$\Pi^{(4)}(\alpha)
=\pi(\bar D_1+D_2(\alpha_4))-p_1(\alpha)\bar D_1-p_2(\alpha_4)D_2(\alpha_4)$.
Substituting the expression for $p_1(\alpha)\bar D_1$ yields
\begin{equation*}
\Pi^{(4)}(\alpha)
=\pi\bar D_1+\pi D_2(\alpha_4)-p_2(\alpha_4)D_2(\alpha_4)-U_2(\alpha_4)-C(\bar D_1)+(1-\alpha)\Delta_U.
\end{equation*}
Using $p_2(\alpha_4)=aD_2(\alpha_4)+(b-a\bar D_2)$ and $U_2(\alpha_4)=\tfrac12 aD_2(\alpha_4)^2$, we obtain
\begin{equation*}
\Pi^{(4)}(\alpha)
=\pi\bar D_1+(\pi-b+a\bar D_2)D_2(\alpha_4)-\tfrac32 aD_2(\alpha_4)^2-C(\bar D_1)+(1-\alpha)\Delta_U.
\end{equation*}
By~\eqref{eq:opt_D_for_R3R4}, $D_2(\alpha_4)=\tfrac{\pi-b+a\bar D_2}{3a}=\frac23 D_2^\ast$.
Substituting this expression back, we obtain
\begin{equation*}
\Pi^{(4)}(\alpha)
=\pi\bar D_1+\frac23 a(D_2^\ast)^2-C(\bar D_1)+(1-\alpha)\Delta_U.
\end{equation*}
Using $\pi-b=2aD_1^\ast-a\bar D_1$ and $C(\bar D_1)=b\bar D_1-\tfrac12 a\bar D_1^2$, we can equivalently write
\begin{equation*}
\Pi^{(4)}(\alpha)
=\frac23 a(D_2^\ast)^2-\tfrac12 a\bar D_1^2+2aD_1^\ast\bar D_1+(1-\alpha)\Delta_U.
\end{equation*}

For $\Pi^{(1,2)}(\alpha)$, in either Regime~$1$ or Regime~$2$, the solution is interior, and prices are pinned down by provided energy through the inverse price response functions. Consequently, the induced profit is fully determined by $(\tilde D_1(\alpha), \tilde D_2(\alpha))$, which evolve according to~\eqref{eq:D_alpha}. Depending on $\alpha$, however, $(\tilde D_1(\alpha), \tilde D_2(\alpha))$ may or may not constitute the optimal solution. For example, if Regime~$4$ is reached before $\tilde D_1(\alpha)$ attains $\bar D_1$, then $(\tilde D_1(\alpha), \tilde D_2(\alpha))$ is no longer optimal.  
Using $\pi - b + a \bar D_i = 2 a D_i^\ast$, we can write
\begin{equation*}
\begin{aligned}
\Pi^{(1,2)}(\alpha)&=(\pi-b+a\bar D_1)\tilde D_1(\alpha)-a\tilde D_1(\alpha)^2+(\pi-b+a\bar D_2)\tilde D_2(\alpha)-a\tilde D_2(\alpha)^2\\
&=a\big(2D_1^\ast\tilde D_1(\alpha)-\tilde D_1(\alpha)^2+2D_2^\ast\tilde D_2(\alpha)-\tilde D_2(\alpha)^2\big).
\end{aligned}
\end{equation*}

Then, the difference between $\Pi^{(4)}$ and $\Pi^{(1,2)}$ is given by
\begin{equation}
\begin{aligned}
\Delta^\Pi(\alpha)
:=& \Pi^{(4)}(\alpha)-\Pi^{(1,2)}(\alpha)\\
=& a\left(
\frac{2}{3}(D_2^\ast)^2
-\frac12 \bar D_1^2
+2D_1^\ast\bar D_1
+\frac12(1-\alpha)\Delta_D^2
-2D_1^\ast \tilde D_1(\alpha)
-2D_2^\ast \tilde D_2(\alpha)
+\tilde D_1(\alpha)^2
+\tilde D_2(\alpha)^2
\right)\\
=& a\left(
\frac{2}{3}(D_2^\ast)^2
-\frac12 \bar D_1^2
+2D_1^\ast\bar D_1
+\frac12\big(\tilde D_2(\alpha)^2-\tilde D_1(\alpha)^2\big)
-2D_1^\ast \tilde D_1(\alpha)
-2D_2^\ast \tilde D_2(\alpha)
+\tilde D_1(\alpha)^2
+\tilde D_2(\alpha)^2
\right)\\
=& a\left(
\frac{2}{3}(D_2^\ast)^2
-\frac12 \bar D_1^2
+2D_1^\ast\bar D_1
-2D_1^\ast \tilde D_1(\alpha)
-2D_2^\ast \tilde D_2(\alpha)
+\frac12 \tilde D_1(\alpha)^2
+\frac32 \tilde D_2(\alpha)^2
\right).
\end{aligned}
\label{eq:diff_profit_3_and_4}
\end{equation}
Therefore, if $\Delta^\Pi(\alpha_3)>0$, it follows that $\alpha_4<\alpha_3$.
At $\alpha_3$, we have $\tilde D_1(\alpha_3)=\bar D_1$, and thus
\begin{equation*}
\Delta^\Pi(\alpha_3)
=a\left(
\frac{2}{3}(D_2^\ast)^2
-2D_2^\ast \tilde D_2(\alpha_3)
+\frac32 \tilde D_2(\alpha_3)^2
\right)
=\frac{3a}{2}\left(
\tilde D_2(\alpha_3)-\frac{2}{3}D_2^\ast
\right)^2 \ge 0,
\end{equation*}
with equality if and only if $\tilde D_2(\alpha_3)=\frac23 D_2^\ast$.
Since the above condition implies that the profit under Regime~$4$ exceeds the profit under the energy adjustment dynamics in~\eqref{eq:D_alpha}, it remains to consider the condition under which $p_1(\alpha)\ge b$ holds.
\begin{equation*}
\begin{aligned}
p_1(\alpha)\bar D_1\ge b\bar D_1 
&\iff
U_2(\alpha_4)+C(\bar D_1)-(1-\alpha)\Delta U\ge b\bar D_1\\
&\iff
\frac12 a \left(\frac23D_2^\ast\right)^2 
+ b\bar D_1
-\frac12 a \bar D_1^2
-\frac12 a\left(D_2(\alpha)^2-\bar D_1^2\right)
\ge b\bar D_1 \\
&\iff
D_2(\alpha)\le \frac23 D_2^\ast.
\end{aligned}
\end{equation*}
Therefore, the necessary condition that regime transitions directly to Regime~$4$ without passing through Regime~$3$ is
\begin{equation}
    D_2(\alpha_3)\le \frac23 D_2^\ast 
    \iff 
    \frac{\bar D_1 D_2^\ast}{2\bar D_1 - D_1^\ast}\le \frac23 D_2^\ast 
    \iff 
    \frac{\pi-b}{2a}\le 0.
\label{eq:alpha_4}
\end{equation}


We now examine how the performance measures vary with $\alpha\le\tilde\alpha$.
Regarding the total consumer utility, 
\begin{equation*}
    \frac{dU}{d\alpha} = a D_1(\alpha) D_1'(\alpha) + a D_2(\alpha) D_2'(\alpha)
= \eta'(\alpha) \underbrace{\left[ \frac{{D_1^{\ast}}^2}{(1-\eta(\alpha))^3} - \frac{{D_2^{\ast}}^2}{(1+\eta(\alpha))^3} \right]}_{(\star)}.
\end{equation*}

Since $\eta(0)=0$, it follows that $\frac{dU(0)}{d\alpha}=-\eta'(0)\Delta_{D^2}<0$.
Therefore, $\frac{dU(\alpha)}{d\alpha}$ remains negative for $\alpha < \alpha_4$, where $\alpha_4$ denotes the value of $\alpha$ at which the bracketed term $(\star)$ becomes zero. Formally, with $c$~defined in~\eqref{eq:definition_c}, $\alpha_5$ is defined by
\begin{equation}
\frac{{D_1^{\ast}}^2}{(1-\eta(\alpha_5))^3}
= \frac{{D_2^{\ast}}^2}{(1+\eta(\alpha_5))^3}
\quad\Longrightarrow\quad
0<\eta(\alpha_5)=\frac{c^{2/3}-1}{c^{2/3}+1}<\frac{c-1}{c+1}=\eta(1),
\label{eq:eta_5}
\end{equation}
where the last inequality holds because the function $f(x)=\tfrac{x-1}{x+1}$ is increasing for all $x \ge 1$. Since $\eta(\alpha)$ is also increasing, we obtain $\alpha_5 < 1$. Moreover, since $\eta(0)=0$ and $\eta(\alpha_5)>0$, it follows that $\alpha_5>0$. Consequently, Regime~$1$ cannot be terminal, and Regime~$2$ cannot arise as the initial regime.

For CNW, we have
\begin{equation*}
\begin{aligned}
\frac{d{W_{\mathrm{CNW}}}}{d\alpha}
=& 2\left( \frac{D_1'(\alpha)}{D_1(\alpha)} + \frac{D_2'(\alpha)}{D_2(\alpha)} \right)
= 2\left( \frac{\eta'(\alpha)}{1-\eta(\alpha)} - \frac{\eta'(\alpha)}{1+\eta(\alpha)} \right)= \frac{4\,\eta(\alpha)\,\eta'(\alpha)}{1-\eta(\alpha)^2}>0.
\end{aligned}
\end{equation*}
Thus, for $\alpha<\min(\alpha_3,\alpha_4,\alpha_5)$, $U$ decreases, and so does $W_{\mathrm{SW}}$. Overall, Regime~$1$ is characterized as follows.
\[
\boxed{
U_1~\text{increases},\quad
U_2~\text{decreases},\quad
U~\text{decreases},~W_{\mathrm{CNW}}~\text{increases},~\text{and}~
W_{\mathrm{SW}}~\text{decreases}.\quad(\text{Regime}~$1$)
}
\]

On the other hand, if $\alpha_5 < \min(\alpha_3,\alpha_4)$, then Regime~$1$ arises first for $\alpha < \alpha_5$, and Regime~$2$ follows for $\alpha_5 \le \alpha < \min(\alpha_3,\alpha_4)$. In Regime~$2$, $U_1$, $U_2$, and $W_{\mathrm{CNW}}$ move in the same direction as in Regime~$1$, however, the total consumer utility $U$ increases.
Social welfare is
\begin{equation*}
\begin{aligned}
\frac{d W_{\mathrm{SW}}}{d \alpha}
=&\sum_{i\in[2]} D'_i(\alpha)\left(\pi-b + a\bar{D}_i - 2aD_i(\alpha)\right)\\
=& \eta'(\alpha)\left(\frac{D_1^\ast}{(1-\eta(\alpha))^2}\left(\pi-b+a\bar{D}_1-2aD_1(\alpha)\right)-\frac{D_2^\ast}{(1+\eta(\alpha))^2}\left(\pi-b+a\bar{D}_2-2aD_2(\alpha)\right)\right)\\
=& -2a\eta(\alpha)\eta'(\alpha)\left(\frac{D_1^\ast}{(1-\eta(\alpha))^2}D_1(\alpha)+\frac{D_2^\ast}{(1+\eta(\alpha))^2}D_2(\alpha)\right)\\
=& -2a\eta(\alpha)\eta'(\alpha)\left(\frac{{D_1^\ast}^2}{(1-\eta(\alpha))^3}+\frac{{D_2^\ast}^2}{(1+\eta(\alpha))^3}\right)<0.\\
\end{aligned}
\end{equation*}
where the third equality follows from~\eqref{eq:stationary_utility} and the last equality follows from~\eqref{eq:D_alpha}. The last inequality follows since $\eta(\alpha)$, $\eta'(\alpha)$, and $\eta(\alpha)<1$. Therefore,
\[
\boxed{
U_1~\text{increases},\quad
U_2~\text{decreases},\quad
U~\text{increases},~W_{\mathrm{CNW}}~\text{increases},~\text{and}~
W_{\mathrm{SW}}~\text{decreases}.\quad(\text{Regime}~$2$)
}
\]

Finally, we examine which of the following regime transitions are possible: (i) $2 \to 3$, (ii) $2 \to 4$, (iii) $1 \to 3$, and (iv) $1 \to 4$.
The possibility of transitions (iii) and (iv) will be addressed in the case $\lambda>0$. Transitions (i) and (ii) correspond, respectively, to
$$
(i)\ \alpha_5 < \alpha_3 < \alpha_4,
\quad\text{and}\quad
(ii)\ \alpha_5 < \alpha_4 < \alpha_3.
$$

First, by~\eqref{eq:alpha_4}, if $\frac{\pi-b}{2a} > 0$, Regime~$4$ follows Regime~$3$. Therefore, it suffices to show that the condition $\alpha_5 < \alpha_3$ can be satisfied. The condition $\alpha_5 < \alpha_3$ holds if and only if
$$
\alpha_5 < \alpha_3
\iff \eta(\alpha_5) < \eta(\alpha_3)
\iff \frac{c^{2/3}-1}{c^{2/3}+1} < 1 - \frac{D_1^\ast}{\bar D_1}
\iff (D_1^\ast)^{1/3}(D_2^\ast)^{2/3} < 2\bar D_1 - D_1^\ast
$$
Substituting the closed-form expressions $D_i^\ast = \tfrac{\pi-b}{2a} + \tfrac{\bar D_i}{2}$, the inequality simplifies to
$$
\left(\frac{\pi-b}{a} + \bar D_1\right)
\left(\frac{\pi-b}{a} + \bar D_2\right)^2
<
\left(3\bar D_1 - \frac{\pi-b}{a}\right)^3
$$
This inequality can be satisfied for some parameter values.  
For instance, when $(a,b,\pi,\bar D_1,\bar D_2,D_{\mathrm s}) = (1,9,9.4,1.2,3.5,4.5)$, the inequality holds.

For case (ii), we show that there exists a set of parameters
$(a,b,\pi,\bar D_1,\bar D_2,D_{\mathrm s})$
such that
$$\alpha_5 < \alpha_4 < \alpha_3 <1 \iff \eta(\alpha_5) < \eta(\alpha_4) < \eta(\alpha_3)<\eta(1),$$
where the equivalence holds because $\eta(\alpha)$ is increasing in $\alpha$. This ordering implies that the dynamics proceed through Regime~$1$, Regime~$2$, and then Regime~$4$. As discussed above, once Regime~$4$ is reached, Regime~$3$ cannot occur. We do not pursue a fully analytical characterization of this case for two reasons. First, computing $\eta(\alpha_4)$ requires solving an equation of cubic order in $\eta$. Second, since our objective is to establish the feasibility of the transition, it suffices to provide an explicit parameter instance for which the ordering holds.

Consider the parameter set
$$(a,b,\pi,\bar D_1,\bar D_2,D_{\mathrm s})=(1,11,9.5,2,10,10).$$
Then
$$D_1^\ast=\frac{\pi-b}{2a}+\frac{\bar D_1}{2}=0.25,
\qquad
D_2^\ast=\frac{\pi-b}{2a}+\frac{\bar D_2}{2}=4.25.$$
First, by~\eqref{eq:eta_5}, with $c=D_2^\ast/D_1^\ast$, we obtain
$$\eta(\alpha_5)
=\frac{c^{2/3}-1}{c^{2/3}+1}
\approx 0.737.$$
Next, $\eta(\alpha_4)$ is defined as the value of $\eta$ that makes
$\Delta^\Pi$ in~\eqref{eq:diff_profit_3_and_4} equal to zero.  
Writing~\eqref{eq:diff_profit_3_and_4} as a function of $\eta$, we have
$$
\Delta^\Pi(\eta)
=
a\Bigg(
\frac{2}{3}(D_2^\ast)^2
-\frac12 \bar D_1^2
+2D_1^\ast\bar D_1
-2D_1^\ast \frac{D_1^\ast}{1-\eta}
-2D_2^\ast \frac{D_2^\ast}{1+\eta}
+\frac12\left(\frac{D_1^\ast}{1-\eta}\right)^2
+\frac32\left(\frac{D_2^\ast}{1+\eta}\right)^2
\Bigg).
$$
Solving $\Delta^\Pi(\eta)=0$ yields
$$
\eta(\alpha_4)\approx 0.851.
$$
Finally, by~\eqref{eq:eta_3},
$$
\eta(\alpha_3)=1-\frac{D_1^\ast}{\bar D_1}=0.875.$$
Therefore, for this parameter set,
$$\eta(\alpha_5) < \eta(\alpha_4) < \eta(\alpha_3) < \eta(1)=\frac{c-1}{c+1} \approx 0.889,$$
where $\eta(1)$ from~\eqref{eq:eta_1}. This establishes that the transition corresponding to case (ii) is feasible.

\textbf{\boldmath Case $\lambda > 0$.} Suppose $D_1^\ast + D_2^\ast = D_\mathrm{s}$ and $0 < D_i^\ast < \bar D_i$ for all $i \in [2]$. We first consider $(D_1(\alpha), D_2(\alpha))$ which is the interior of all constraints except $D_1+D_2\le D_{\mathrm{s}}$. Then, the fairness constraint and the supply constraint jointly imply
\begin{equation*}
D_2(\alpha)^2 - D_1(\alpha)^2 = (1-\alpha)\Delta_{D^2},\quad D_1(\alpha) + D_2(\alpha) = D_s,
\end{equation*}
where $\Delta_{D^2}:={D_2^\ast}^2-{D_1^\ast}^2$.
Then,
\begin{equation}
D_1(\alpha) = \frac{D_s - \delta(\alpha)}{2},
\qquad
D_2(\alpha) = \frac{D_s + \delta(\alpha)}{2},
\label{eq:D_alpha_delta}
\end{equation}
where
\begin{equation*}
\delta(\alpha):=D_2(\alpha) - D_1(\alpha)
= \frac{D_2(\alpha)^2 - D_1(\alpha)^2}{D_1(\alpha)+D_2(\alpha)} = \frac{(1-\alpha)\Delta_{D^2}}{D_s}.
\end{equation*}
$\delta(\alpha)$ is decreasing function with $\alpha$. i.e., $\delta'(\alpha)<0$.

With an interior solution, $U_i=\tfrac{1}{2}aD_i^2$, therefore
$$U'_1(\alpha)=aD_1(\alpha)D_1'(\alpha)=-aD_1(\alpha)\frac{\delta'(\alpha)}{2}>0\quad{and}\quad U'_2(\alpha)=aD_2(\alpha)D_2'(\alpha)=aD_2(\alpha)\frac{\delta'(\alpha)}{2}<0.$$
Total consumer utility is
$$\frac{dU}{d\alpha}=U_1'(\alpha)+U_2'(\alpha)=\frac{a\delta'(\alpha)}{2}\left(D_2(\alpha)-D_1(\alpha)\right)<0.$$
The last inequality holds because $D_1(\alpha)< D_1(1)=D_2(1)<D_2(\alpha)$ for all $\alpha\in[0,1)$.
For CNW, we have
\begin{equation*}
\begin{aligned}
\frac{d{W_{\mathrm{CNW}}}}{d\alpha}
=& 2\left( \frac{D_1'(\alpha)}{D_1(\alpha)} + \frac{D_2'(\alpha)}{D_2(\alpha)} \right)
= \delta'(\alpha)\left( \frac{1}{D_2(\alpha)} - \frac{1}{D_1(\alpha)} \right)>0.
\end{aligned}
\end{equation*}
Lastly, since $U$ decreases, social welfare is also decreasing with $\alpha$. Overall, this corresponds to Regime~$1$.

The dynamics of~\eqref{eq:D_alpha_delta} may change when either the constraint $D_1 + D_2 \le D_s$ switches from binding to slack or $D_1 \le \bar D_1$ becomes binding. The former case does not introduce a new regime. When the supply constraint becomes slack, i.e., $D_1 + D_2 < D_s$, the solution remains interior and therefore corresponds to either Regime~$1$ or Regime~$2$ under $\lambda = 0$.

Unlike the case $\lambda = 0$, Regime~$1$ can be a terminal regime. By~\eqref{eq:D_alpha_delta}, we have $D_1(1) = D_2(1) = \frac{D_{\mathrm s}}{2}$, which can be chosen to be smaller than $\bar D_1$. For instance, let $(a,b,\pi,\bar D_1,\bar D_2) = (1, 11, 12, 2, 10)$ with $D_{\mathrm s} = 4$. If the constraint $D_1(\alpha)=\bar D_1$ becomes binding before the constraint $D_1 + D_2 \le D_s$ turns slack, let $\alpha_6$ denote the smallest value of $\alpha$ at which this occurs. At $\alpha=\alpha_6$, we have
\begin{equation*}
D_1(\alpha_6)=\bar D_1
\quad\text{and}\quad
D_2(\alpha_6)
=\frac{D_s}{2}+\frac{\delta(\alpha_6)}{2}
=\frac{D_s}{2}+\frac{D_s}{2}-\bar D_1
=D_s-\bar D_1.
\end{equation*}
At this point, the solution transitions from Regime~$1$ to either Regime~$3$ or Regime~$4$. As characterized in the analysis of Regimes~$3$ and~$4$, the resulting regime is determined by whether
\begin{equation}
D_2(\alpha_6)=D_s-\bar D_1>\frac{\pi-b+a\bar D_2}{3a}.
\label{eq:cond_R3R4}
\end{equation}
where $D_2^\dagger$ is defined in Lemma~\ref{lemma:opt_sol}. If this condition holds, the solution transitions to Regime~$3$, otherwise, it transitions to Regime~$4$. This condition may or may not hold. For instance, with $(a,b,\pi,\bar{D}_1,\bar{D}_2)=(1, 11, 12,2,10)$, if $D_{\mathrm{s}}=6$, \eqref{eq:cond_R3R4} holds. On the other hand, if $D_{\mathrm{s}}=5$, \eqref{eq:cond_R3R4} does not hold.

Finally, we consider the remaining case that has not been addressed so far, namely
\begin{equation*}
D_1^{\ast}=0
\quad\text{and}\quad
D_2^{\ast}=D_s.
\end{equation*}
Since $D_1^\ast$ (respectively, $D_2^\ast$) attains its minimum (respectively, maximum) value at $\alpha=0$, we parameterize deviations from this initial solution as
\begin{equation*}
D_1(\alpha)=\Delta D_1
\quad\text{and}\quad
D_2(\alpha)=D_s-\Delta D_2,~\text{where}~\Delta D_1, \Delta D_2\ge 0,
\end{equation*}
where $\Delta D_1$ and $\Delta D_2$ represent the corresponding changes in provided energy with $\alpha$.

We evaluate the profit loss incurred when moving from $\alpha=0$ to a given $\alpha>0$ as
\begin{equation*}
\begin{aligned}
\Pi(0,D_s)-\Pi(\Delta D_1,D_s-\Delta D_2)
=&\;\pi D_s-(aD_s+b-a\bar D_2)D_s
-\pi(\Delta D_1+D_s-\Delta D_2) \\
&+(a\Delta D_1+b-a\bar D_1)\Delta D_1
+\big(a(D_s-\Delta D_2)+b-a\bar D_2\big)(D_s-\Delta D_2) \\
=&\;
-a(\Delta D_1)^2
+(b-a\bar D_1-\pi)\Delta D_1 +a(\Delta D_2)^2
+(\pi-b-2aD_s+a\bar D_2)\Delta D_2.
\end{aligned}
\end{equation*}
Hence, characterizing the optimal response amounts to minimizing this profit reduction subject to feasibility and fairness constraints. The resulting optimization problem is given by
\begin{equation*}
\begin{aligned}
\min_{\Delta D_1,\Delta D_2}\quad
&-a(\Delta D_1)^2 + (b-a\bar D_1-\pi)\Delta D_1
+ a(\Delta D_2)^2 + (\pi-b-2aD_s+a\bar D_2)\Delta D_2 \\
\text{s.t.}\quad
&\Delta D_1-\Delta D_2 \le 0, \\
&(D_s-\Delta D_2)^2-(\Delta D_1)^2=(1-\alpha)\Delta_{D^2}, \\
&\Delta D_1 \ge 0,\quad \Delta D_2 \ge 0.
\end{aligned}
\end{equation*}

Using the equality constraint, $(\Delta D_1)^2=(D_s-\Delta D_2)^2-(1-\alpha)\Delta_{D^2}$. Substituting this expression into the quadratic part of the objective function yields
\begin{equation*}
-a(\Delta D_1)^2 + a(\Delta D_2)^2 = -a\big((D_s-\Delta D_2)^2-(1-\alpha)\Delta_{D^2}\big)
+ a(\Delta D_2)^2 = -aD_s^2 + 2aD_s\Delta D_2 + a(1-\alpha)\Delta_{D^2},
\end{equation*}
which differs from a linear function only by constants. Dropping constant terms, the objective reduces to
\begin{equation*}
\min_{\Delta D_1,\Delta D_2}
\quad
(b-a\bar D_1-\pi)\Delta D_1
+(\pi-b+a\bar D_2)\Delta D_2.
\end{equation*}

Since the objective is linear and the feasible set is defined by the above constraints, the inequality constraint is binding at the optimum, implying
\begin{equation*}
\Delta D_1=\Delta D_2.
\end{equation*}
Therefore, the constraint $D_1 + D_2 \le D_s$ remains binding for all $\alpha$, with $D_1$ increasing and $D_2$ decreasing as $\alpha$ varies. This behavior coincides with the interior solution under $\lambda>0$, which corresponds to Regime~$1$.
\hfill \Halmos

\end{proof}

\newpage
\section{Experimental Details}
\label{appendix:experimental_details}

To produce the numerical results, we implement the model in \texttt{Pyomo} and utilize off-the-shelf optimization solvers. We first solve the optimization problem under energy fairness in terms of the decision variable $D$. Because this formulation yields a convex optimization problem, we employ the solver \texttt{IPOPT} to obtain its global optimum efficiently.

In contrast, when enforcing price fairness or utility fairness, the induced provided energy becomes piecewise linear in prices. This creates a nonconvex feasible region, as regime-dependent price transitions introduce kinks in both the objective function and the constraints. Such nonconvexities prevent standard convex optimization solvers from certifying global optimality. Consequently, we adopt a partition-enumeration procedure: we decompose the feasible region into a finite collection of partitions, solve the subproblem associated with each partition, and retain the best solution across all partitions.

For the price fairness case, each partition induces a convex optimization problem. In contrast, the utility fairness formulation yields a nonconvex problem. To ensure global optimality, we solve these subproblems using the global nonlinear solver \texttt{Couenne} and complement this with a grid-search procedure to verify solution quality. 

More specifically, the optimization problem in~\eqref{eq:conjunction} exhibits a piecewise-defined price-energy relationship. Depending on whether $p_i$ lies below $b-a\bar{D}_i$, between $b-a\bar{D}_i$ and $b$, or above $b$, the provided energy takes the value $0$, an affine function of $p_i$, or the saturation level $\bar{D}_i$, respectively. These breakpoints introduce non-differentiabilities in the objective and constraints, yielding a nonconvex feasible region. However, once the regime of each consumer is fixed, the price-energy map becomes linear. This motivates our approach: instead of solving the full nonconvex problem directly, we decompose the feasible region into finitely many convex subregions and solve each subproblem separately.

Each consumer $i$ belongs to one of the following sets,
\begin{equation*}
\mathcal{I}_0 := \{i : p_i < b-a\bar{D}_i\},\qquad
\mathcal{I}_1 := \{i : b-a\bar{D}_i \le p_i \le b\},\qquad
\mathcal{I}_2 := \{i : p_i > b\}.
\end{equation*}
With $N$ consumers, there are at most $3^N$ such assignments. Let $\mathcal{Z}$ denote this finite collection of partitions, and let $\mathcal{I}=(\mathcal{I}_0,\mathcal{I}_1,\mathcal{I}_2)\in\mathcal{Z}$ denote one of them. Once $\mathcal{I}$ is fixed, provided energy expressions remain unchanged with respect to prices, eliminating all kinks in the feasible set. The price fairness problem under a given partition therefore admits a tractable representation.
\begin{equation}
\begin{aligned}
P(\mathcal{I}):=\max_{\mathbf{p}} \quad 
& \sum_{i\in\mathcal{I}_1} \left(\pi - p_i\right)
  \left(\frac{p_i-b}{a} + \bar D_i\right)
+ \sum_{i\in\mathcal{I}_2} \left(\pi - p_i\right)\bar{D}_i \\[3pt]
\text{s.t.}\quad
& p_i \le b-a\bar{D}_i,\quad i\in\mathcal{I}_0,\\
& b-a\bar{D}_i \le p_i \le b,\quad i\in\mathcal{I}_1,\\
& p_i \ge b,\quad i\in\mathcal{I}_2,\\[3pt]
& p_i - p_j \le (1-\alpha)\Delta_p,\quad
  p_j - p_i \le (1-\alpha)\Delta_p,\quad \forall i\ne j.
\end{aligned}
\label{eq:conjunction}
\end{equation}

Because both the objective and constraints are linear or quadratic in $\mathbf{p}$, $P(\mathcal{I})$ is a convex quadratic program. Hence, although the original problem is nonconvex, the feasible region decomposes as
\begin{equation*}
\mathcal{F} = \bigcup_{\mathcal{I}\in\mathcal{Z}} \mathcal{F}(\mathcal{I}),
\end{equation*}
and the global maximizer can be obtained by
\begin{equation*}
\Pi^\ast
= \max_{\mathcal{I}\in\mathcal{Z}} \; 
\max_{\mathbf{p}\in\mathcal{F}(\mathcal{I})} \Pi(\mathbf{p}),
\end{equation*}

We solve $P(\mathcal{I})$ for each $\mathcal{I}\in\mathcal{Z}$ using
\texttt{IPOPT} and retain the solution attaining the highest objective value.

Under utility fairness, the fairness constraint involves consumers' utilities,
\begin{equation*}
U_i - U_j \le (1-\alpha)\Delta_U,\qquad
U_j - U_i \le (1-\alpha)\Delta_U.
\end{equation*}
Utilities are given by
\begin{equation*}
U_i =
\begin{cases}
0, & i\in\mathcal{I}_0,\\[2pt]
\dfrac{1}{2a}\,(p_i-b+a\bar{D}_i)^2, & i\in\mathcal{I}_1,\\[6pt]
p_i\bar{D}_i + \dfrac{1}{2}a\bar{D}_i^2 - b\bar{D}_i, & i\in\mathcal{I}_2.
\end{cases}
\end{equation*}
These nonlinear expressions destroy convexity even within a fixed partition. Accordingly, each subproblem is solved using the global nonlinear solver \texttt{Couenne}, and we verify optimality using a supplementary grid search. This ensures that the best solution across all partitions is globally optimal.

Lastly, Algorithm~\ref{alg:partition-enumeration} summarizes our approach. By enumerating all partitions and solving the induced subproblems exactly, we obtain a certified global maximizer for both price and utility fairness formulations.

\begin{algorithm}[htbp]
\caption{Partition–enumeration solver}
\label{alg:partition-enumeration}
\begin{algorithmic}[1]
\STATE \textbf{Input:} $(a,b,\pi,\bar{\mathbf{D}},D_{\mathrm{s}})$, fairness level $\alpha$.
\STATE Construct $\mathcal{Z}$, the finite set of partitions
$\mathcal{I}=(\mathcal{I}_0,\mathcal{I}_1,\mathcal{I}_2)$ of $[N]$.
\STATE Initialize $\Pi^\ast \leftarrow -\infty$ and $\mathcal{I}^\ast \leftarrow \emptyset$.
\FOR{each $\mathcal{I}\in\mathcal{Z}$}
    \STATE Solve $P(\mathcal{I})$ and obtain $\Pi(\mathcal{I})$ and $\mathbf{p}(\mathcal{I})$.
    \IF{$\Pi(\mathcal{I}) > V^\ast$}
        \STATE $\Pi^\ast \leftarrow \Pi(\mathcal{I})$, $\mathcal{I}^\ast \leftarrow \mathcal{I}$,
        $\mathbf{p}^\ast \leftarrow \mathbf{p}(\mathcal{I})$.
    \ENDIF
\ENDFOR
\STATE \textbf{Output:} $(\mathcal{I}^\ast,\mathbf{p}^\ast)$ and $\Pi^\ast$.
\end{algorithmic}
\end{algorithm}

\section{Supplementary Case Study}
\subsection{Details of the Case Study}
\label{appendix:case_study_preprocessing}
This section describes the data preprocessing procedures in detail. We restrict our sample to participants who consented to take part in the experiment. The iFlex field experiment was conducted in two phases: a pilot phase in the winter of $2019/20$ (Phase~$1$) and a full-scale pricing experiment in the winter of $2020/21$ (Phase~$2$). Phase~$1$ involved a smaller number of households and primarily served as a pilot study. Participants may have joined Phase~$1$ only, Phase~$2$ only, or both phases. To avoid potential learning or familiarity effects and to ensure sufficient sample size, we exclude Phase~$1$ participants and focus on households that participated exclusively in Phase~$2$.

Within Phase~$2$, participants are assigned either to a control group or to a price (treatment) group. Households in the control group are not exposed to any experimental price signals throughout the experiment period, whereas households in the treatment group are randomly assigned hourly different incentive price for some randomly selected days. Since our analysis focuses on price responsiveness, we further restrict the sample to households in the price treatment group. Finally, to incorporate household characteristics such as income, we retain only participants who completed the post-experiment survey. As a result, our analysis includes $1{,}233$ of the $7{,}410$ participants. We then retain hourly observations that vary over time for each household, such as electricity consumption and incentive price signals. The detailed sample selection process is summarized in Figure~\ref{fig:filtering_flow}.
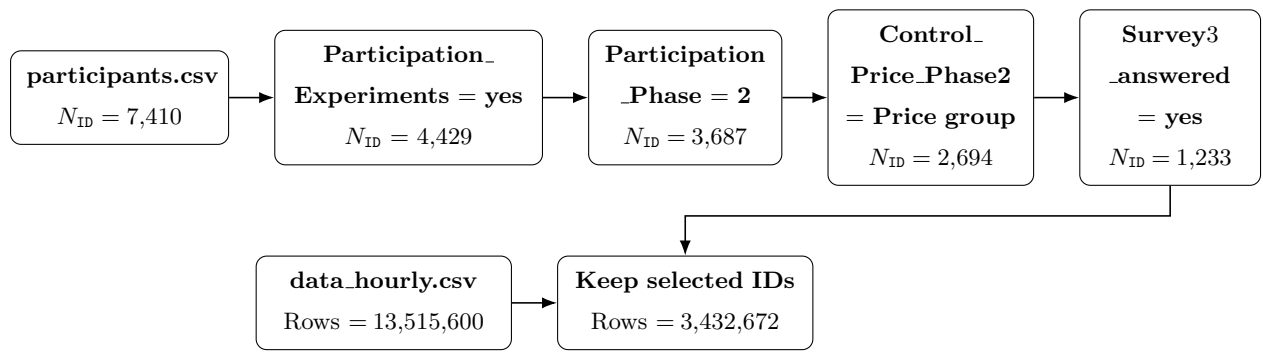
\begin{figure}[htbp]
\caption{Data Filtering Process.\vspace{5mm}}
\centering
\begin{tikzpicture}[
  font=\footnotesize,
  node distance=6mm and 6mm,
  box/.style={
    draw, rounded corners,
    align=center,
    inner xsep=6pt, inner ysep=6pt,
    text width=0.15\linewidth
  },
  arrow/.style={-Latex, line width=0.6pt}
]

\node[box, text width=0.15\linewidth] (hh0) {\textbf{participants.csv}\\
$N_{\texttt{ID}}=7{,}410$};

\node[box, text width=0.19\linewidth, right=of hh0] (hh1) {\textbf{Participation\_ Experiments $=$ yes}\\
$N_{\texttt{ID}}=4{,}429$};

\node[box, text width=0.13\linewidth, right=of hh1] (hh2) {\textbf{Participation}\\
\textbf{\_Phase $=$ 2}\\
$N_{\texttt{ID}}=3{,}687$};

\node[box, text width=0.14\linewidth, right=of hh2] (hh3) {\textbf{Control\_ Price\_Phase2}\\
\textbf{$=$ Price group}\\
$N_{\texttt{ID}}=2{,}694$};

\node[box, text width=0.12\linewidth,right=of hh3] (hh4) {\textbf{Survey$3$ \_answered}\\
\textbf{$=$ yes}\\
$N_{\texttt{ID}}=1{,}233$};

\draw[arrow] (hh0) -- (hh1);
\draw[arrow] (hh1) -- (hh2);
\draw[arrow] (hh2) -- (hh3);
\draw[arrow] (hh3) -- (hh4);

\node[box, text width=0.18\linewidth, below=12mm of hh2] (row1) {\textbf{Keep selected IDs}\\
Rows $=3{,}432{,}672$};

\node[box, text width=0.18\linewidth, left=of row1] (row0) {\textbf{data\_hourly.csv}\\
Rows $=13{,}515{,}600$};


\draw[arrow] (row0) -- (row1);

\draw[arrow] (hh4.south) -- ($(hh4.south)+(0,-4mm)$) -| (row1.north);

\end{tikzpicture}
\label{fig:filtering_flow}
\end{figure}

\subsection{Additional Results under Small $D_{\mathrm{s}}$}
\label{appendix:additional_case_study}
We consider the case $D_{\mathrm{s}} = 0.3\sum_{g\in[3]} n_g \bar{D}_g$ in this section. Figure~\ref{fig:demand_case_study_30} presents the outcomes under energy fairness.
With a limited aggregation target, aggregated energy is primarily constrained by $D_\mathrm{s}$. In contrast to the large $D_\mathrm{s}$ case, more flexible consumers (cluster~$1$) reduce their provided energy, accompanied by increases from less flexible consumers (clusters~$2$ and~$3$). Treating clusters~$2$ and~$3$ as a combined group, the resulting provided energy adjustment directions are fully consistent with our theoretical predictions for the two-agent setting. As more flexible consumers reduce their provided energy, their utility also decreases. Similar to the large $D_\mathrm{s}$ case, the magnitude of more flexible consumers’ utility change dominates those of less flexible consumers and the aggregator, leading to a decline in both total consumer utility and social welfare. In contrast, the individual utility of less flexible consumers increases, which is consistent with policymakers’ fairness objectives that prioritize benefits for low-income households. Moreover, as the absolute level of provided energy becomes more evenly distributed across consumers, the CNW increases.

\begin{figure}[htbp]
\FIGURE{
    \includegraphics[width=\textwidth]{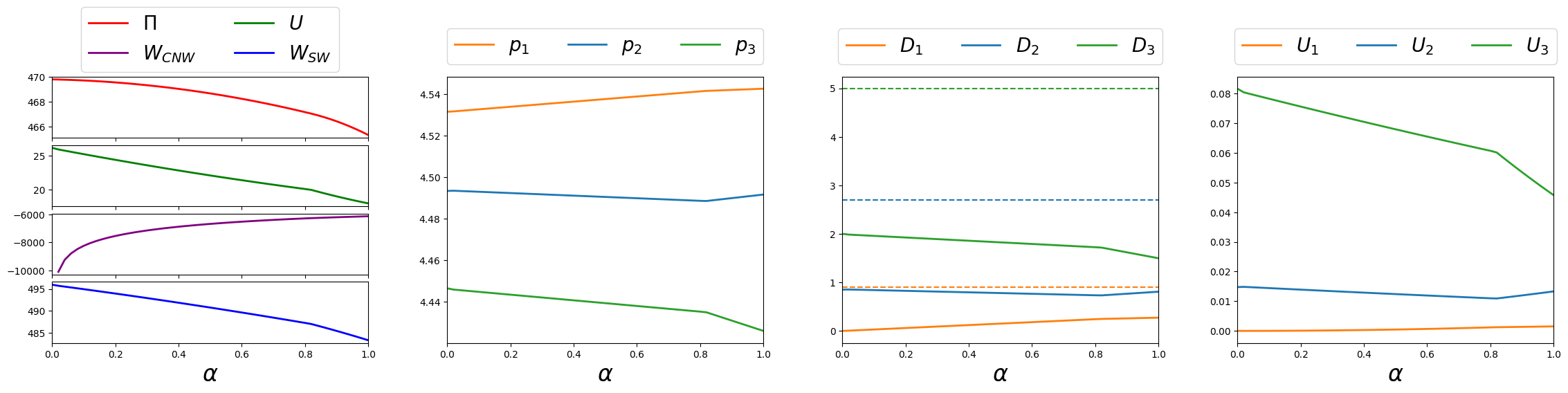}
}
{Energy fairness \label{fig:demand_case_study_30}
\vspace{.3cm}}
{}
\end{figure}

Figure~\ref{fig:price_case_study_30} reports the outcomes under price fairness. In this setting, the aggregator does not need to aggregate any energy from cluster~$1$, and consumers in this group therefore do not participate in the program regardless of the price they face. As a result, the initial price for cluster~$1$ is economically irrelevant. We thus exclude the price of cluster~$1$ when computing the initial price gap. Without this adjustment, the indeterminacy in price~$1$ would mechanically inflate the measured price gap, since it could take arbitrary values without affecting any performance measures. For any value of~$\alpha$, price fairness does not induce participation from cluster~$1$, and provided energy from this group remains zero throughout. Consequently, CNW is ill-defined in this case due to zero utility for cluster~$1$, and we therefore exclude it from the plot. Although the lack of participation by cluster~$1$ is driven by the small $D_{\mathrm{s}}$, the main message of price fairness remains unchanged--price fairness does not harm more flexible consumers, nor does it benefit less flexible consumers.

\begin{figure}[htbp]
\FIGURE{
    \includegraphics[width=\textwidth]{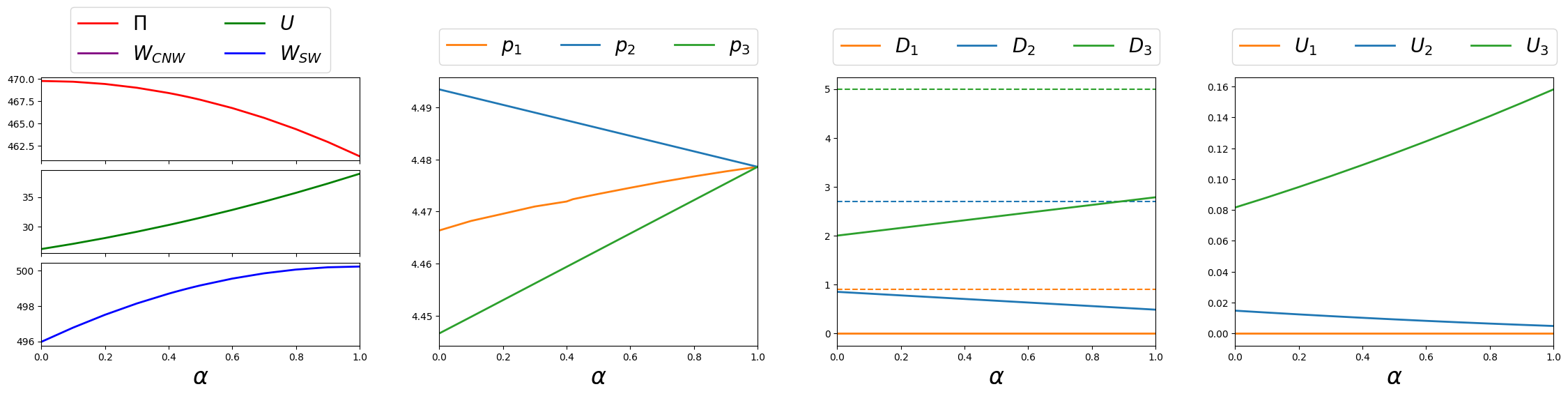}
}
{Price Fairness \label{fig:price_case_study_30}
\vspace{.3cm}}
{consumer Nash welfare is not reported because it takes the value $-\infty$.}
\end{figure}

Figure~\ref{fig:utility_case_study_30} reports the outcomes under utility fairness. Similar to the large $D_{\mathrm{s}}$ case, utility fairness benefits the low-utility group (cluster~$1$) while harming the high-utility group (cluster~$3$). By enforcing utility fairness, consumers in cluster~$1$ gain the opportunity to participate in the program, a feature not observed under price fairness. However, as in the large $D_{\mathrm{s}}$ case, the associated aggregator's profit loss is the largest among all fairness criteria considered. A key difference arising in the small $D_{\mathrm{s}}$ setting is that, with the exception of CNW, all performance measures decline as $\alpha$ increases. Taken together, these results indicate that while utility fairness achieves explicit inclusion of previously excluded consumers (cluster~$1$), it does so at the cost of both substantial aggregator's profit loss and a reduction in total consumer utility.

\begin{figure}[htbp]
\FIGURE{
    \includegraphics[width=\textwidth]{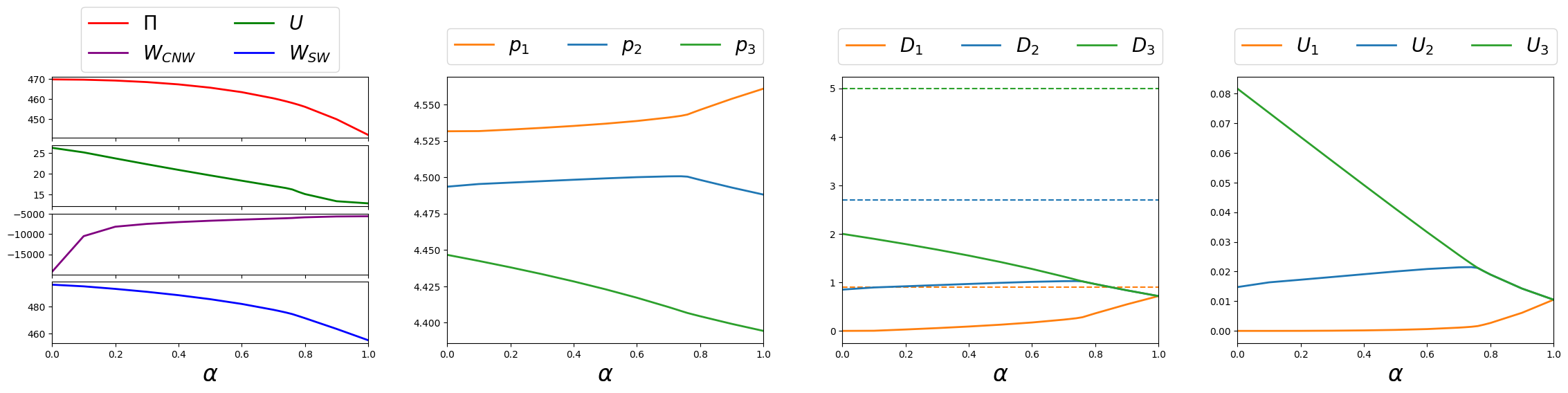}
}
{Utility Fairness \label{fig:utility_case_study_30}
\vspace{.3cm}}
{}
\end{figure}
\end{APPENDIX}

\end{document}